%% file: percentiles.tex
\documentclass{llncs}
\usepackage{fullpage}
\usepackage{amsmath}
\usepackage{amssymb}
\usepackage{mathtools}
\usepackage{xspace}
\usepackage{tikz}
\usepackage{fixme}
\usepackage{paralist,mdwlist}
\usepackage{manfnt,mparhack}
\usepackage{caption}
\usepackage{subcaption}
\usepackage{multirow}
\usepackage{wrapfig}
\usepackage{todonotes}

\usepackage[colorlinks=true, citecolor=blue,linkcolor=red,urlcolor=black]{hyperref}

\let\doendproof\endproof
\renewcommand\endproof{~\hfill\qed\doendproof}

\input{commands}

\newcommand\Val{\textrm{\sf Val}}
\newcommand\calD{\ensuremath{\mathcal{D}}}

\newcommand\calI{\ensuremath{\mathcal{I}}}
\newcommand\calT{\ensuremath{\mathcal{T}}}
\newcommand\calA{\ensuremath{\mathcal{A}}}

\newcommand\calM{\ensuremath{\mathcal{M}}}
\newcommand\mecs{\ensuremath{\textrm{\sf MEC}}}
\newcommand\Safe{\ensuremath{\textrm{\sf Safe}}}
\newcommand\PTIME{\textrm{\sf P}}
\newcommand\PSPACE{\textrm{\sf PSPACE}}
\newcommand\EXPTIME{\textrm{\sf EXPTIME}}
\newcommand\restr[2]{\ensuremath{\left.#1\right|_{#2}}}
\newcommand\bis{\ensuremath{\textrm{\sf bis}}}

\usetikzlibrary{automata,calc}
\usetikzlibrary{arrows}

\title{Percentile Queries in Multi-Dimensional Markov Decision Processes\thanks{M.~Randour is an F.R.S.-FNRS Postdoctoral Researcher, J.-F.~Raskin is supported by ERC Starting Grant (279499: inVEST). Work partly supported by European project CASSTING (FP7-ICT-601148).}}
\author{Mickael Randour\inst{1} \and Jean-Fran\c{}cois Raskin\inst{1} \and Ocan Sankur\inst{2}}

\institute{
D\'epartement d'Informatique, Universit\'e libre de Bruxelles (ULB), Belgium
\and CNRS, Irisa, Rennes, France
}

\begin{document}
\maketitle

\begin{abstract}
Markov decision processes (MDPs) with multi-dimensional weights are useful to analyze systems with multiple objectives that may be conflicting and 
require the analysis of trade-offs. We study the complexity of percentile queries in such MDPs and give algorithms to synthesize strategies that enforce such constraints. 
Given a multi-dimensional weighted MDP and a quantitative payoff function~$f$, thresholds $v_i$ (one per dimension), and probability thresholds $\alpha_i$, 
we show how to compute a single strategy to enforce that for all dimensions $i$, the probability of outcomes $\rho$ satisfying $f_i(\rho) \geq v_i$ is at least $\alpha_i$. 
We consider classical quantitative payoffs from the literature (sup, inf, lim sup, lim inf, mean-payoff, truncated sum, discounted sum).  
Our work extends to the quantitative case  the multi-objective model checking problem studied by Etessami et al.~\cite{EKVY-lmcs08} in unweighted MDPs.
\end{abstract}

\input intro.tex

\section{Preliminaries}

\smallskip\noindent\textbf{Markov decision processes.} A finite \textit{Markov decision process} (MDP) is a tuple $M = (S,A,\delta)$ where
$S$ is the finite set of \emph{states}, $A$~is the finite set of \emph{actions} and $\delta\colon S\times A \rightarrow \calD(S)$
is a partial function called the \emph{probabilistic transition function}, where~$\calD(S)$ denotes the set of rational probability distributions over~$S$.
The set of actions that are available in a state $\state \in \states$ is denoted by $A(s)$. 
We use $\delta(s,a,s')$ as a shorthand for $\delta(s,a)(s')$.
An \emph{absorbing state} $s$ is such that for all $a \in A(s)$, $\delta(s,a,s) = 1$. We assume w.l.o.g.~that MDPs are \textit{deadlock-free}: for all $s \in S$, $A(s) \neq \emptyset$ (if not the case, we simply replace the deadlock by an absorbing state with a unique action). An MDP where for all $s \in \states$, $\vert A(s)\vert = 1$ is a fully-stochastic process called a \textit{Markov chain}.

A \textit{weighted} MDP is a tuple $M = (S,A,\delta,\weight)$, where $w$ is a \emph{$d$-dimension weight function}~$w\colon A \rightarrow \mathbb{Z}^d$.
For any $l \in \{1,\ldots,d\}$, we denote $\weight_l \colon A \rightarrow \integ$ 
the projection of~$\weight$ to the~$l$-th dimension, i.e., the function mapping each action $a$ to the $l$-th element of vector $w(a)$.
A \emph{run} of~$M$ is an infinite sequence $s_1a_1 \ldots a_{n-1}
s_n\ldots{}$ of
states and actions such that $\delta(s_i,a_i,s_{i+1})>0$ for all~$i\geq 1$.
Finite prefixes of runs are called \emph{histories}.

Fix an MDP $M = (S,A,\delta)$.
An \emph{end-component} (EC) of $M$ is
an MDP $C = (S',A',\delta')$ with $S' \subseteq S$, $\emptyset \neq A'(s) \subseteq A(s)$ for all~$s \in S'$,
and~$\supp(\delta(s,a)) \subseteq S'$ for all~$s \in S', a \in A'(s)$
(here $\supp(\cdot)$ denotes the support), $\delta' = \restr{\delta}{S'\times A'}$ and such that $C$ is \textit{strongly connected}, i.e., there is a run between any pair of states in $S'$.
The union of two ECs with non-empty intersection is
an EC; one can thus define \emph{maximal} ECs.
We let $\mecs(M)$ denote the set of maximal ECs of~$M$, computable in
polynomial time~\cite{DeAlfaro-phd97}.

\smallskip\noindent\textbf{Strategies.} A~\emph{strategy}~$\sigma$ is a function $(SA)^*S\rightarrow \calD(A)$ such that
for all~$h \in (SA)^*S$ ending in~$s$, we have~$\supp(\sigma(h)) \subseteq A(s)$. The set of all strategies is $\strats$. A strategy is \textit{pure} if all histories are mapped to \textit{Dirac distributions}. A strategy~$\sigma$ can be encoded by a 
\emph{Moore machine}, $(\calM,\sigma_a,\sigma_u,\alpha)$ 
where~$\calM$ is a finite or infinite set of memory elements,~$\alpha$ the \emph{initial
distribution on~$\calM$}, 
$\sigma_u$ the \emph{memory update
function} $\sigma_u : A\times S \times \calM \rightarrow \calM$, 
and $\sigma_a : S \times \calM \rightarrow \calD(A)$ the \emph{next action
function} where $\supp(\sigma_a(s,m)) \subseteq A(s)$ for any~$s\in S$ and~$m \in \calM$.
We say that~$\sigma$ is \emph{finite-memory} if~$|\calM|<\infty$, and \emph{$K$-memory} if~$|\calM|=K$;
it is memoryless if~$K=1$, thus only depends on the last state of the history. 
We see such strategies as functions $s \mapsto \calD(A(s))$ for $s \in S$. 
A strategy is \emph{infinite-memory} if~$|\calM|$ is infinite.
For a class of problems, we say that strategies use linear (resp. polynomial, exponential) memory 
if there exist strategies for which $K$ is linear (resp. polynomial, exponential) in the size of~$M$. The entity choosing the strategy is often called the \textit{controller}.

An MDP~$M$, a strategy~$\sigma$ encoded by $(\calM,\sigma_a,\sigma_u,\alpha)$, and a state~$s$ determine a
Markov chain $M_s^\sigma$ defined on the state space $S\times \calM$ as follows.
The initial distribution is such that for any~$m \in \calM$,
state $(s,m)$ has probability $\alpha(m)$, and~$0$ for other states. For any pair of states
$(s,m)$ and~$(s',m')$, the probability of the transition $(s,m),a,(s',m')$ is
equal to $\sigma_a(s,m)(a) \cdot \delta(s,a,s')$
if $m' = \sigma_u(s,m,a)$, and to~$0$ otherwise.
A~\emph{run} of~$M_s^\sigma$ is an infinite sequence of the form
$(s_1,m_1),a_1,(s_2,m_2),a_2,\ldots$, where each
$(s_i,m_i),a_i,(s_{i+1},m_{i+1})$ is a transition with nonzero probability
in~$M_s^\sigma$, and~$s_1=s$. 
When considering the probabilities of events in $M_s^\sigma$, we
will often consider sets of runs of~$M$. Thus, given $E \subseteq
(SA)^\omega$, we denote by $\pr_{M,s}^\sigma[E]$ the probability of the runs 
of~$M_s^\sigma$ whose projection\footnote{The projection of a run
$(s_1,m_1),a_1,(s_2,m_2),a_2,\ldots$ in $M_s^\sigma$ to $M$ is simply the run $s_{1}a_{1}s_{2}a_{2}\ldots{}$ in $M$.} to~$M$ is in~$E$, i.e., the probability of event $E$ when $M$ is executed with initial state~$s$ and strategy $\sigma$. Note that every event has a uniquely defined probability \cite{vardi_FOCS85} (Carath\'eodory's extension theorem induces a unique probability measure on the Borel $\sigma$-algebra over $(SA)^\omega$).

\smallskip\noindent\textbf{Almost-sure reachability of ECs.} Let~$\Inf(\rho)$ denote the random variable representing the disjoint union of states and actions that occur
infinitely often in the run~$\rho$. By an abuse of notation,
we see $\Inf(\rho)$ as a sub-MDP~$M'$ if it contains
exactly the states and actions of~$M'$. 
It was shown that for any MDP~$M$, state~$s$, strategy~$\sigma$,
${\pr_{M,s}^\sigma[\Inf \text{ is an EC}]=1}$~\cite{DeAlfaro-phd97}.

\smallskip\noindent\textbf{Multiple reachability.} Given a subset~$T$ of states,
let $\Diamond T$ be the \emph{reachability objective w.r.t.~$T$}, defined as 
the set of runs visiting a state of~$T$ at least once.
The \emph{multiple reachability} problem consists, given MDP~$M$, state~$\initState$, target sets
$T_1,\ldots,T_{\queries}$, and probabilities~$\alpha_1,\ldots,\alpha_{\queries} \in [0,1] \cap \rat$, 
in deciding whether there exists a strategy~$\sigma \in \strats$
such that
$\bigwedge_{i = 1}^{q} \pr_{M,\initState}^\sigma[\Diamond T_i]\geq \alpha_i.$
The \emph{almost-sure multiple reachability} problem restricts to $\alpha_1=\ldots=\alpha_{\queries} = 1$.

\label{sec:percentilesProblem}
\smallskip\noindent\textbf{Percentile problems.} 
We consider \textit{payoff functions} among $\inf$, $\sup$, $\liminf$,
$\limsup$, mean-payoff, truncated sum (shortest path) and discounted sum. For
any run~$\rho=s_1a_1s_2a_2\ldots$, dimension~$l \in \{1,\ldots,d\}$, and weight
function $w$,
\vspace{-1mm}
\begin{itemize*}
\item
  $\inf_l(\rho) = \inf_{j\geq 1} w_l(a_j)$,
  $\sup_l(\rho) = \sup_{j\geq 1} {w_l(a_j)}$,
\item
  $\liminf_l(\rho) = \liminf_{j \rightarrow \infty} w_l(a_j)$,
  $\limsup_l(\rho) = \limsup_{j \rightarrow \infty} w_l(a_j)$,
\item 
  $\mpinf_l(\rho) = \liminf_{n \rightarrow \infty} \frac{1}{n} \sum_{j=1}^n w_l(a_j)$,   $\mpsup_l(\rho) = \limsup_{n \rightarrow \infty} \frac{1}{n} \sum_{j=1}^n w_l(a_j)$,
\item
  $\discSum{\discount_l}_l(\rho) = \sum_{j=1}^{\infty} \discount_l^{j}\cdot w_l(a_j)$, with $\discount_l \in \left] 0, 1\right[ \cap \rat$ a rational discount factor,
\item 
$\truncatedSum{\truncatedTarget}_l(\rho) = \sum_{j=1}^{n-1} w_l(a_j)$ with $s_{n}$ the first visit of a state in $\truncatedTarget \subseteq \states$. If $\truncatedTarget$ is never reached, then we assign $\truncatedSum{\truncatedTarget}_l(\rho) = \infty$.
\end{itemize*}
\vspace{-1mm}

For any payoff function~$f$, $f_l \geq v$ defines the runs~$\rho$ that satisfy $f_l(\rho)\geq v$.
A~\emph{percentile constraint} is of the form~$\pr_{M,\initState}^\sigma[ f_{l} \geq v] \geq \alpha$, where~$\sigma$ is to be synthesized given threshold value~$v$ and probability~$\alpha$. We study \emph{multi-constraint percentile queries} requiring to simultaneously satisfy~$q$ constraints each referring to a possibly different dimension.
Formally, given a $\dimension$-dimensional weighted MDP~$\markovProcess$, initial state $\initState \in \states$, payoff function~$f$,
dimensions $l_1,\ldots,l_q \in \{1,\ldots,\dimension\}$, value thresholds $v_1,\ldots,v_\queries \in \rat$ and probability thresholds $\alpha_1,\ldots, \alpha_\queries \in [0,1] \cap \rat$, the \emph{multi-constraint percentile problem} asks if there exists a strategy~$\strat \in \strats$ such that query
\[
\query \coloneqq \bigwedge_{i = 1}^{q}\; \pr_{M,\initState}^\strat\big[f_{l_{i}} \geq v_i\big] \geq
	\alpha_i
	\]
	holds. We can actually solve queries
$\exists?\, \sigma$, $\bigvee_{i=1}^m \bigwedge_{j=1}^{n_i} \pr_{M,\initState}^\strat\big[f_{l_{i,j}} \geq v_{i,j} \big]\geq \alpha_{i,j}$.
We present our results for conjunctions of constraints only since
the latter is equivalent to verifying the disjuncts independently:
in other terms, to
$\bigvee_{i=1}^m \exists \sigma \bigwedge_{j=1}^{n_i} \pr_{M,\initState}^\strat\big[f_{l_{i,j}} \geq v_{i,j} \big]\geq \alpha_{i,j}$.

We distinguish \textit{single-dimensional percentile problems} ($\dimension = 1$) from
\textit{multi-dimensional} ones ($\dimension > 1$).
We assume w.l.o.g.~that $\queries \geq \dimension$ otherwise one can simply neglect unused dimensions.
For some cases, we will consider the \emph{$\varepsilon$-relaxation} of the problem,
which consists in ensuring each value~$v_i-\varepsilon$ with probability~$\alpha_i$.

\smallskip\noindent\textbf{Complexity.} We assume binary encoding of constants,
and define the \emph{model size} $|M|$, a polynomial in $\vert \states \vert$ and the size of the \textit{encoding} of weights and probabilities (e.g., $\log_{2} W$ with $W$ the largest absolute weight), as the size of the representation of~$M$; and the \emph{query size} $|\query|$, a polynomial in the number of constraints $q$ and the encoding of thresholds, that of the query.
The \emph{problem size} refers to the sum of the two.

\smallskip\noindent\textbf{Memory and randomness.} 
Throughout the paper, we will study the memory requirements for strategies w.r.t.~different classes of percentile queries.
Here, we show, by a simple example, that randomness is always needed for all payoff functions.

\begin{lemma}
Randomized strategies are necessary for multi-dimensional percentile queries for any payoff function.
\end{lemma}

\begin{proof}
Let $M$ be a 2-dim.~deterministic MDP with $\states = \{s_{0}, s_{1}, s_{2}\}$, $A = \{a, b\}$ and the transition function defined as $\delta(s_0, a, s_1) = 1$, $\delta(s_0, b, s_2) = 1$, $\delta(s_1, a, s_1) = 1$ and $\delta(s_2, b, s_2) = 1$. Essentially there are only two possible runs in this MDP: $\rho_{1} = s_{0} (a\,s_{1})^{\omega}$ and $\rho_{2} = s_{0} (b\,s_{2})^{\omega}$. Assume that the weight and the payoff functions are chosen such that $f(\rho_{1}) = (1,0)$ and $f(\rho_{2}) = (0,1)$: they are incomparable. Consider the query \[\query \coloneqq \pr_{M,s_{0}}^\strat \big[f_{1} \geq 1/2\big] \geq 1/2\quad \wedge \quad \pr_{M,s_{0}}^\strat \big[f_{2} \geq 1/2\big] \geq 1/2.\] It is easy to see that $\query$ can only be satisfied by a strategy that chooses between $a$ and $b$ with equal probability, hence no pure strategy satisfies the query. 
Note that here $f$ can be chosen anything among $\sup,\limsup,\mpinf,\mpsup$, $\discSum{\discount_l}$ with appropriate $\discount_l$,
and $\truncatedSum{T_l}$ with target sets~$T_1=\{s_1\}$ and~$T_2=\{s_2\}$ respectively for each query. For $\inf$, and~$\liminf$, we may switch the weight vectors to obtain the same result.
\end{proof}

\section{Multiple Reachability and Contraction of MECs}
\label{section:reachsafe}

\smallskip\noindent\textbf{Multiple reachability.} The multiple reachability  problem
was studied \cite{EKVY-lmcs08} where an algorithm based on a linear program (LP) of size polynomial in the model and exponential in the query was given.
As a particular case, it was proved that restricting the target sets to absorbing states yields a polynomial-size LP. We will use this LP later in Fig.~\ref{fig:mpinf-lp} in Section~\ref{section:mp}.

\begin{theorem}[\cite{EKVY-lmcs08}]
  \label{thm:absorbing-reachsafe}
  Memoryless strategies suffice for multiple reachability with absorbing target states, and
  can be decided and computed in polynomial time.
  With arbitrary targets, exponential-memory strategies (in query size) can be computed in time polynomial in the model and exponential in the query.
\end{theorem}

In this section, we improve over this result by showing that the case of almost-sure multiple reachability is \PSPACE-complete,
with a recursive algorithm and a reduction from QBF satisfiability.
This also shows the \PSPACE-hardness of the general problem. Moreover, we show that exponential memory is required for strategies,
following a construction of \cite{DBLP:journals/acta/ChatterjeeRR14}.

\begin{theorem}
  \label{thm:asreach}
  The almost-sure multiple reachability problem is \PSPACE-complete,
  and strategies need exponential memory in the query size.
\end{theorem}

We first show the \PSPACE-completeness of the almost-sure multiple reachability problem.

\begin{lemma}
  The almost-sure multiple reachability problem is \PSPACE-complete.
\end{lemma}

\begin{proof}
  We start by showing \PSPACE-membership. Let~$M$ be an MDP,~$s_0$ a state, and $T_1,\ldots,T_q$ target sets. We write
  $T=T_1\cup \ldots \cup T_q$. Note first that we know how to solve the problem in polynomial time for~$q=1$.
  Let~$M'$ be the MDP obtained by~$M$ by making all states in~$T$ absorbing. The procedure works as follows. For each state~$x \in T$,
  let us define $I = \{ 1\leq i\leq q \mid x \not \in T_i\}$; we clearly have $|I|< n$.
  We recursively verify whether there is a strategy almost-surely satisfying the multiple reachability objective $(T_i)_{i \in I}$.
  Let~$\calT$ denote all states of~$T$ for which the recursive call returned positively. We now check in polynomial time whether
  the set~$\calT$ can be reached almost-surely from~$s_0$.
  Note that the recursive call depth is linear, so the whole procedure uses polynomial space.

  We now prove the equivalence between~$M$ and~$M'$. Assume that there is a strategy~$\sigma$ almost-surely reaching~$\calT$ in~$M'$.
  This strategy can be followed in~$M$ until some state~$x$ of~$\calT$ is reached, which happens almost-surely. But we know, by the recursive callof our procedure, that from any such state~$x \in \calT$ there exists a strategy almost-surely satisfying the rest of the reachability objectives.
  Thus, by extending~$\sigma$ in each state~$x \in \calT$ by these strategies, we construct a solution to the multiple reachability problem in~$M$. Notice that the constructed strategy uses linear memory since it is ``memoryless'' between each switch.

  Conversely, assume that there is a strategy~$\sigma$ satisfying the multiple reachability query in~$M$ from~$s_0$.
  Towards a contradiction, assume that some state $x \in T \setminus \calT$ is reached with positive probability in~$M$ under~$\sigma$, 
  thus also in~$M'$ under the same strategy. We know by the recursive call of our procedure that the remaining targets cannot be satisfied 
  almost-surely by any strategy from state~$x$ in~$M$. It follows that strategy $\sigma$ fails to satisfy all targets almost-surely 
  from~$s_0$, a contradiction.

  To show \PSPACE-hardness, we reduce the truth value of a quantified Boolean formula (QBF) to our problem. An instance of QBF is a quantified Boolean formula over $X=\{x_1,x_2,\dots,x_n\}$

   $$\Psi \equiv \exists x_1 \forall x_2 \exists x_2 \dots \forall x_{n-1} \exists x_n \cdot C_1 \land C_2 \land \dots C_m$$
   \noindent
    where each clause $C_i$ is the disjunction of 3 literals taken in $\{ x , \neg x \mid x \in X\}$. 
    From $\Psi$, we construct an (acyclic) MDP as shown in Fig.~\ref{fig:qbf}. For each variable $x_i$, there are three states called $x_i$, $f_i$ and $t_i$ in the MDP. In a state $x_i$ that corresponds to an {\em existentially quantified} variable, there are two actions that are available: $\top$ and $\bot$. The action $\top$ visits (deterministically) the state $t_i$ while the action $\bot$ visits the state $f_i$, and then in the two cases, the run proceeds to the state for the next variable.  Intuitively, choosing $\top$ in $x_i$ corresponds to the choice of truth value {\sf true} for $x_i$, and $\bot$ to truth value {\sf false}. In a state $x_i$ that corresponds to an {\em universally quantified} variable, there is only the action $*$ available and the successor is chosen uniformly at random between $f_i$ and $t_i$. The targets are defined as follows: for each clause $C_j$, the target set $T_j=\{ t_i \mid x_i \in C_j \} \cup \{ f_i \mid \neg x_i \in C_j \}$ must be visited with probability one. Clearly, given the value assigned to a variable $x_i$, we visit exactly the set of target sets $T_j$ that correspond to the clauses that are made true by the valuation of $x_i$. It should be clear now that the histories in the MDP are in bijection with the valuation of the Booelan variables in $\Psi$ and that the set of valuations that satisfies $\Psi$ correspond exactly to the histories that visits all the sets $T_j$, $1 \leq j \leq n$ with probability one.

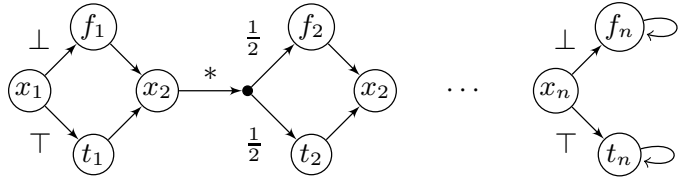
\begin{figure}[t]
  \centering  
  \scalebox{1.2}{\begin{tikzpicture}
    \tikzstyle{every state}=[node distance=1cm,minimum size=10pt, inner sep=1pt];
    \node[state] at (0,0) (x1){$x_1$};
    \node[state,below right of=x1] (t1) {$t_1$};
    \node[state,above right of=x1] (f1) {$f_1$};
    \node[state,below right of=f1] (x2) {$x_2$};
    \node[state,right of=x2, fill=black,minimum size=3pt] (x2b) {};
    \node[state,below right of=x2b] (t2) {$t_2$};
    \node[state,above right of=x2b] (f2) {$f_2$};
    \node[state,below right  of=f2] (x3) {$x_2$};
    \node[right of=x3] (etc) {$\cdots$};
    \node[state,right of=etc] (xn) {$x_n$};
    \node[state,below right of=xn] (tn) {$t_n$};
    \node[state,above right of=xn] (fn) {$f_n$};
    \path[-latex']
    (x1) edge node[below left] {$\top$} (t1) 
    (x1) edge node[above left] {$\bot$} (f1)
    (t1) edge (x2)
    (f1) edge (x2)
    (x2) edge node[above] {$*$} (x2b)
    (x2b) edge node[below left] {$\frac{1}{2}$} (t2)
    (x2b) edge node[above left] {$\frac{1}{2}$} (f2)
    (t2) edge (x3)
    (f2) edge (x3)
    (xn) edge node[below left] {$\top$} (tn)
    (xn) edge  node[above left] {$\bot$} (fn)
    (fn) edge[loop right] (fn)
    (tn) edge[loop right] (tn);
  \end{tikzpicture}}
  \caption{Reduction for the QBF formula
    $\exists x_1 \forall x_2 \ldots \exists x_n C_1 \land \ldots \land C_m$.
  The objectives are $T_j = \{t_i \mid x_i \in C_j\}  \cup \{f_i \mid \lnot x_i \in C_i\}$
for all $1\leq j\leq m$.}
  \label{fig:qbf}
\end{figure}

  Now, we claim that there is a strategy to reach each set $T_j$, $1 \leq j \leq n$, with probability one if and only if the formula $\Psi$ is true. Indeed, if $\Psi$ is true, we know that there exists for each existentially quantified variable $x_i$ a choice function $g_{x_i}$ which assign a truth value to $x_i$ given the truth values chosen for the variables that appears before $x_i$ in the quantification block. These choice functions naturally translate into a (deterministic memryfull) strategy that mimics the choices of truth values by choosing between $\bot$ and $\top$ accordingly. We get that if the formula is true (all closed are made true) then the associated strategy visits all the target sets with probability one.   
  
  For the other direction, we first note that it is not useful for the scheduler to play a randomised strategy. As the graph of the MDP is acyclic (except for the two states $t_n$ and $f_n$ that have a self loop), all the target sets are visited with probability one if and only if all the outcomes of the strategy visits all the target sets. So, if the scheduler plays randomly say in state $x_i$ then all the resulting outcomes for action $\bot$ and all the resulting outcomes for action $\top$ must visit all the target sets, so both choices need to be good and so there is no need for randomisation and the scheduler can safely choose one of the two arbitrarily. So, pure strategies are sufficient and but we have seen that pure strategies corresponds exactly to the choice functions in the QBF problem. So is clear that from a winning strategy for the scheduler, we can construct a choice function that makes the formula true.
  \end{proof}

We establish an exponential lower bound on the memory requirements based on a family of MDPs depicted in Fig.~\ref{fig:multiReach_exp_mem} and inspired from \cite[Lemma 8]{DBLP:journals/acta/ChatterjeeRR14}.

\begin{lemma}
  \label{lem:multiReach_expMemoryLB}
  Exponential-memory in the query size is necessary for almost-sure multiple reachability.
\end{lemma}
\begin{proof}
Consider the unweighted MDP $\markovProcess$ depicted in Fig.~\ref{fig:multiReach_exp_mem}. The MDP is composed of $k$ gadgets where a state between $s_{i,L}$ and $s_{i,R}$ is stochastically chosen (they are equiprobable), followed by $k$ gadgets where the controller can decide to visit either $s'_{i,L}$ or $s'_{i,R}$. We define an almost-sure multiple reachability problem for target sets 
\begin{equation*}
T_{i} = \{s_{1,L}, s'_{1,L}\}, \{s_{1,R}, s'_{1,R}\}, \{s_{2,L}, s'_{2,L}\}, \ldots{}, \{s_{k,L}, s'_{k,L}\}, \{s_{k,R}, s'_{k,R}\}.
\end{equation*}
Hence, this problem requires $q = 2\cdot k$ constraints to be defined. We claim that a strategy satisfying this problem cannot be expressed by a Moore machine containing less than $2^{k} = 2^{\frac{q}{2}}$ memory states.

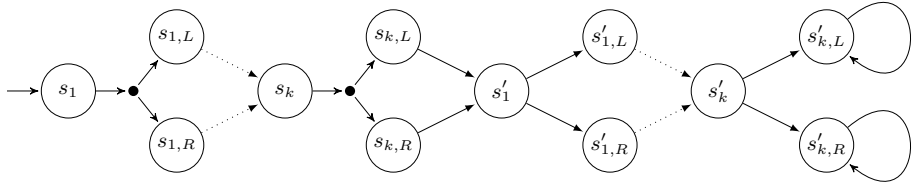
\begin{figure}[htb]
  \centering   
  \scalebox{0.9}{\begin{tikzpicture}[->,>=stealth',shorten >=1pt,auto,node
    distance=2.5cm,bend angle=45,scale=0.4, inner sep=0pt, font=\small]
    \tikzstyle{p1}=[draw,circle,text centered,minimum size=8mm]
    \tikzstyle{p2}=[fill,circle,text centered,minimum size=1.5mm]
    \node[p1]  (0)  at (0, 0) {$s_{1}$};
    \node[p2]  (0b)  at (2.4, 0) {};
    \node[p1]  (1) at (4, 2) {$s_{1,L}$};
    \node[p1]  (2) at (4, -2)  {$s_{1,R}$};
    \node[p1]  (3) at (8, 0)  {$s_{k}$};
    \node[p2]  (3b)  at (10.4, 0) {};
    \node[p1]  (4)  at (12, 2) {$s_{k,L}$};
    \node[p1]  (5)  at (12, -2) {$s_{k,R}$};
    \node[p1]  (6)  at (16, 0) {$s'_{1}$};
    \node[p1]  (7) at (20, 2) {$s'_{1,L}$};
    \node[p1]  (8) at (20, -2)  {$s'_{1,R}$};
    \node[p1]  (9) at (24, 0)  {$s'_{k}$};
    \node[p1]  (10)  at (28, 2) {$s'_{k,L}$};
    \node[p1]  (11)  at (28, -2) {$s'_{k,R}$};
    \coordinate[shift={(-5mm,0mm)}] (init) at (0.west);
    \path     
    (0) edge (0b)
    (0b) edge[shorten <=1pt] (1)
    (0b) edge[shorten <=1pt] (2)
    (3) edge (3b)
    (3b) edge[shorten <=1pt] (4)
    (3b) edge[shorten <=1pt] (5)
    (10) edge [loop right, out=40, in=320,looseness=3, distance=4cm] (10)
    (11) edge [loop right, out=40, in=320,looseness=3, distance=4cm] (11)
    (init) edge (0);
	\draw[dotted,->,>=latex] (1) to (3);
	\draw[dotted,->,>=latex] (2) to (3);
	\draw[->,>=latex] (4) to (6);
	\draw[->,>=latex] (5) to (6);
	\draw[->,>=latex] (6) to (7);
	\draw[->,>=latex] (6) to (8);
	\draw[dotted,->,>=latex] (7) to (9);
	\draw[dotted,->,>=latex] (8) to (9);
	\draw[->,>=latex] (9) to (10);
	\draw[->,>=latex] (9) to (11);
      \end{tikzpicture}}
      \vspace{-10mm}
      \caption{Family of multiple reachability problems requiring exponential memory.}
\label{fig:multiReach_exp_mem}
  \end{figure}

Indeed, it is clear that to ensure almost-sure reachability of all sets $T_{i}$, the controller has to chose in state $s'_{i}$ the exact opposite action of the one stochastically chosen in $s_{i}$. Remembering the $k$ choices made in states $s_{i}$ requires $k$ bits of encoding. Hence, a satisfying strategy requires a Moore machine with $2^{k}$ memory states to encode those choices.

It is easy to see that if the controller uses a - possibly randomized - strategy $\strat$ with less than $2^{k}$ memory states, then there exists $i \in \{1, \ldots{}, k\}$ such that $\strat(s_{1}\ldots{}s_{i}s_{i,L}\ldots{}s'_{i}) = \strat(s_{1}\ldots{}s_{i}s_{i,R}\ldots{}s'_{i})$, i.e., the controller chooses to go to $s'_{i,L}$ (resp. $s'_{i,R}$) with identical probability against both stochastic choices in~$s_{i}$. Assume that the controller chooses to go toward $s'_{i,L}$ with probability $p \in \left[ 0, 1\right] $ and toward $s'_{i,R}$ with probability $1-p$: this implies that the probability that the target set $\{s_{i,L},s'_{i,L}\}$ (resp. $\{s_{i,R},s'_{i,R}\}$) is never visited is equal to $\frac{1}{2} \cdot (1-p)$ (resp. $\frac{1}{2} \cdot p$). Clearly, it is impossible to have both those probabilities equal to zero simultaneously, which proves that such a strategy cannot satisfy the almost-sure multiple reachability problem defined above, and concludes our proof.
\end{proof}

Despite the above lower bounds, it turns out that the polynomial time algorithm for the case of absorbing targets can be extended:
we identify a subclass of the multiple reachability problem that admits a polynomial-time solution.
In the \emph{nested multiple reachability} problem, the target sets are nested, i.e., 
$T_1 \subseteq T_{2} \subseteq \ldots{} \subseteq T_q$.
The memory requirement for strategies is reduced as well to linear memory.

\begin{theorem}
  \label{thm:monotonic-reach}
  The nested multiple reachability problem can be solved in polynomial time.
  Strategies have memory linear in the query size, which is optimal.
\end{theorem}

 Intuitively, we use $q+1$ copies of the original MDP, one for each target set, plus one last copy. The idea is then to travel between those copies in a way that reflects the nesting of target sets whenever a target state is visited. The crux to obtain a polynomial-time algorithm is then to reduce the problem to a multiple reachability problem \textit{with absorbing states} over the MDP composed of the $q+1$ copies, and to benefit from the reduced complexity of this case.
 
\begin{proof}
  Assume an MDP~$M$, $s_0$ and the target sets $T_1 \subseteq \ldots \subseteq T_q$ are given.
  We make~$q+1$ copies of the MDP~$M$, namely, $M_1,\ldots,M_q,M_{q+1}$. We start $M_{q+1}$ at state~$\initState$.
  We redirect some of the edges as follows.
  For any~$M_i$, state~$s$, action~$a$, and $t \in \supp(\delta(s,a))$, if $t \in T_j$ for some $j<i$, then we direct this edge to state~$t$ in $M_{j'}$
  where $j'$ is the smallest index with $t \in T_{j'}$.
  Hence, along any run, we are in copy~$M_j$ if, and only if we have already satisfied all targets $T_j,\ldots,T_q$.
  Now, we add a fresh absorbing state~$\bot_i$ to each copy. From all states of~$M_i$ a fresh action~$a^\bot$ leads to~$\bot_i$.
  Let us call this new MDP~$M'$. Note that the size of~$M$ is $\mathcal{O}(q|M|)$.

  For each~$i=1\ldots q$, we define $T_i' = \{\bot_i,\bot_{i-1},\ldots,\bot_1\}$.
  We claim that the multiple reachability problem query $(T_i,\alpha_i)_{1\leq i\leq q}$ for~$M$
  is equivalent to~$(T_i',\alpha_i)_{1\leq i\leq q}$ for~$M'$. But the latter query has absorbing target states, thus the problem can be solved in polynomial time
  by~\cite{EKVY-lmcs08}.
  
  Consider a strategy~$\sigma$ for~$M$ achieving the objectives $(T_i,\alpha_i)_{1\leq i\leq q}$.
  We can assume w.l.o.g. that~$\sigma$ is finite-memory by~\cite{EKVY-lmcs08}.
  Let~$S_i$ denote the set of states of~$M_i$.
  We define strategy~$\sigma'$ for~$M'$ as follows. Let us define a mapping~$p(\cdot)$ from the histories of~$M'$ to those of~$M$, where a state of any copy
  is projected to the original state in~$M$. The mapping is actually a bijection from histories of~$M'$ that do not use the action~$a^\bot$
  to histories of~$M$.
  Now, for all histories~$h$ of~$M'$ that end in copy~$M_i$,
  if~$\pr_{M, p(h)}^\sigma[\Diamond \cup_{j <i} T_i ] = 0$, we set $\sigma'(h) = a^\bot$. Otherwise, we let~$\sigma'(h) = \sigma(p(h))$.

  We prove that for all~$i=1\ldots q$, $\pr_{M,s_0}^\sigma[\Diamond T_i] \leq \pr_{M',s_0'}^{\sigma'}[\Diamond T_i']$.
  Let~$\iota(h)$ denote the copy in which~$h$ ends in~$M'$.
  For all histories~$h$ of~$M$ from which the probability of satisfying $\Diamond \cup_{j < \iota(p^{-1}(h))}T_i$ is nonzero,
  we have $\pr_{M,s_0}^\sigma[h] = \pr_{M',s_0'}^{\sigma'}[p^{-1}(h)]$ by definition.
  Define $H_i = \{h \mid \forall i=1\ldots |h|-1, h_i \not \in T_i, h_{|h|} \in T_i, \pr_{M,s_0}^\sigma[h]>0\}$,
  that is, the histories that visit~$T_i$ for the first time at their last state.
  Clearly $\pr_{M,s_0}^\sigma[H_i] = \pr_{M,s_0}^{\sigma}[\Diamond T_i]$.
  But the probability of reaching  $T_i$ is always nonzero along these histories, so 
  we also have $H_i = H_i' := \{h \in H_i \mid \forall i=1\ldots|h|, \pr_{M,h_{1\ldots i}}^\sigma[\Diamond T_i]>0\}$,
  and we get $\pr_{M',s_0'}^{\sigma'}[H_i'] = \pr_{M,s_0}^{\sigma}[H_i']$.
  In other words, $\pr_{M',s_0'}^{\sigma'}[\Diamond \cup_{j \leq i} S_j] \geq \pr_{M,s_0}^\sigma[\Diamond T_i]$,
  that is, the target sets $T_1,\ldots,T_q$ are reached in~$M'$ with at least the same probabilities as in~$M$.
  We now need to show that from any history ending in copy~$M_i$, some state~$\bot_j$ with~$j\leq i$
  is reached almost-surely in~$M'$ under~$\sigma'$. It will follow that
  $\pr_{M',\initState'}^{\sigma'}[\Diamond T_i'] \geq \pr_{M,\initState}^\sigma[\Diamond T_i]$.
  To see this, notice that strategy~$\sigma$ is finite-memory, and so is~$\sigma'$. 
  So there exists $\nu>0$ such that for any state~$s$, and memory element~$m$
  if the probability of satisfying $\Diamond \cup_{j < i} S_j$ is nonzero from~$s$ and~$m$, 
  then it is at least~$\nu$. Note that the probability of never satisfying
  $\Diamond \cup_{j < i} S_j$ while staying in such states is~$0$.
  So, whenever the run reaches copy~$M_i$, almost-surely,
  either some copy~$M_j$ with~$j<i$ is reached, or we reach a history~$h$ such that $\pr_{M,p(h)}^{\sigma}[\Diamond \cup_{j < i} T_j] = 0$, in which case we end in~$\bot_i$.
  The inequality follows.

  Conversely, consider any strategy $\sigma'$ for~$M'$ achieving the reachability objectives $(T_i',\alpha_i)_{1\leq i\leq q}$.
  We assume $\sigma'$ to be memoryless by~\cite{EKVY-lmcs08}.
  We define~$\sigma(h) = \sigma'(p^{-1}(h))$ whenever the action~$\sigma'$ prescribes is different than~$a^\bot$, and otherwise
  $\sigma'$ switches to an arbitrary memoryless strategy. Since all histories of~$M'$ that end in~$\bot_i$ satisfy the objectives
  $T_i\cup\ldots\cup T_q$, strategy~$\sigma$ achieves the objectives $(T_i,\alpha_i)_{1\leq i\leq q}$.
  The memory of $\sigma$ is $\mathcal{O}(q)$ since $\sigma'$ is memoryless in~$M'$ which is made of~$q$ copies of~$M$.

  We now show that linear memory is necessary.
  Consider an MDP~$M$ with states~$s,t_1,\ldots,t_n,\bot$. State~$s$ has~$n$ actions~$a_1,\ldots,a_n$.
  For each~$1\leq i\leq n$, action~$a_i$ leads from~$s$ to~$t_i$ with probability~$1-\frac{1}{i+1}$,
  and with probability~$\frac{1}{i+1}$ leads to absorbing state~$\bot$. From all states~$t_i$ with $i>1$, $s$ is reachable by a deterministic action,
  but from~$t_1$ one can only reach $\bot$.
  The MDP is depicted in Fig.~\ref{fig:nested} ($\bot$ is not shown).

\vspace{-4mm}
\begin{figure}[h]
\begin{center}
\scalebox{0.9}{
  \begin{tikzpicture}
    \tikzstyle{every state}=[node distance=1.5cm,minimum size=20pt, inner sep=1pt];
    \node[state] at (0,0) (s){$s$};
    \node[state,minimum size=1pt, fill=black] at (-2.5, 1) (sn) {};
    \node[state, minimum size=1pt, fill=black] at (0.4,1) (sn1) {};
    \node[state, minimum size=1pt, fill=black] at (2.5, 1) (s1) {};    
    \node[state] at (-3, 2) (tn) {$t_n$};
    \node[state] at (0, 2) (tn1) {$t_{n-1}$};
    \node[right of=tn1] (etc) {$\cdots$};
    \node[state] at (3, 2) (t1) {$t_1$};
    \path[-latex']
    (s) edge[bend left] node[left]{$a_n$} (sn)
    (sn) edge node[left]{\tiny $1-\frac{1}{n+1}$} (tn)
    (sn) edge node[below]{\tiny $\frac{1}{n+1}$} ($(sn)+(-1,0)$)
    (s) edge node[right]{$a_{n-1}$} (sn1)
    (sn1) edge node[left]{\tiny $1-\frac{1}{n}$} (tn1)
    (sn1) edge node[below]{\tiny $\frac{1}{n}$} ($(sn1)+(-0.5,0)$)
    (s) edge[bend right]  node[right]{$a_{1}$} (s1)
    (s1) edge node[left]{\tiny $\frac{1}{2}$} (t1)
    (s1) edge node[left]{\tiny $\frac{1}{2}$} ($(s1)+(-0.5,0)$)
    (tn) edge[gray,bend left=-10]  (s)
    (tn1) edge[gray,bend right=60]  (s)
    (t1) edge[gray]  ($(t1)+(1,0)$);
  \end{tikzpicture}}
  \caption{Linear memory is required for the nested multiple reachability problem.}
  \label{fig:nested}
\end{center}
\vspace{-6mm}
\end{figure}
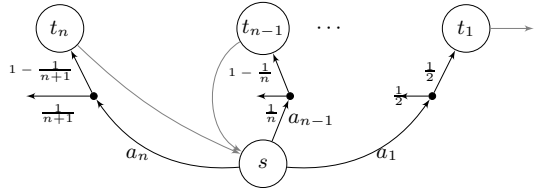

We consider the nested multiple reachability targets $T_1,\ldots,T_n$
with $T_i = \{t_i,\ldots,t_1\}$ for each~$i$, and consider threshold probabilities
$\alpha_1,\ldots,\alpha_n$ defined by
$\alpha_n = 1-\frac{1}{n+1}$, and $\alpha_i = \alpha_{i+1} (1-\frac{1}{i+1})$
for $1\leq i\leq n-1$.

Let us first describe a strategy that satisfies these constraints.
Define a strategy that deterministically
chooses, at each visit to state~$s$, the following actions: $a_n,a_{n-1},\ldots,a_1$. A simple calculation
shows that the constraints are satisfied: the probability of satisfying $T_n$ is at least $1-\frac{1}{n+1}$ by the first action,
that of $T_{n-1}$ is $\frac{1}{n+1}(1-\frac{1}{n})$ by the sequence $a_na_{n-1}$ of actions, and so on.
We argue that this is the only strategy that satisfies these reachability queries showing that $\mathcal{O}(n)$ memory is necessary.

Consider any strategy~$\sigma$ satisfying the multiple reachability queries. We show that $\sigma$ must deterministically choose~$a_n$ in the first step.
In fact, assume that some action~$a_i$ with $i\neq n$ is chosen with probability~$\eta>0$. The probability of moving to $\bot$ under any such action is at least
$\frac{1}{n}$. Thus the probability of going to $\bot$ in the first step (without seeing any $t_i$) is at least
$\eta\frac{1}{n} + (1-\eta) \frac{1}{n+1} > \frac{1}{n+1}$, which is a contradiction.
Now, assume that $\sigma$ deterministically chooses $a_na_{n-1}\ldots a_{n-i+1}$ in the first~$i$ steps. The probability of reaching $T_{n-i}$ in the first $i$ steps is thus $0$,
while the probability of being in~$s$ (and not in~$\bot$) after $i$ steps is $\gamma = (1-\frac{1}{n+1})\cdots (1-\frac{1}{n-i+2})$.
Assume that $\sigma$ does not deterministically choose $a_{n+i}$.
Target~$T_{n-i}$ is reached by histories that eventually choose some action $a_{n-i},\ldots, a_1$. Let~$H$ denote the set of histories
stopping at the first action from this set, i.e.,
$H = s((a_n+\ldots+a_{n-i+1})s)^*(a_{n-i}+\ldots+a_1)$.
Note that at these histories, either we satisfy $T_{n-i}$ or we end in $\bot$, so the probability of satisfying $T_{n-i}$ under~$\sigma$ can be written as
$\gamma \sum_{h \in H} \alpha_h p_h$, where $\alpha_h$ is the probability of~$\sigma$ of choosing the actions of $h$ from the current history,
and $p_h$ is the probability of the resulting run.
We have, for all~$h \in H$, $p_h \leq (1-\frac{1}{n+1})^{\frac{|h|-1}{2}}(1-\frac{1}{n-i+1})$ since 
$h$ contains $\frac{|h|-1}{2}$ actions outside $a_1,\ldots,a_{n-i}$ and  after each such action 
we must come back to~$s$.
For all $h \in H$ with $\frac{|h|-1}{2}>1$, we must have $\alpha_h = 0$ since otherwise we would get $\gamma \sum_{h \in H} \alpha_h p_h< \gamma(1-\frac{1}{n-i+1})$.
Furthermore, if $\sigma$ chooses an action some action~$a_j$ with $1\leq j \leq n-i$ in the first step, the probability of going to $\bot$
is $\frac{1}{j+1} > 1-\frac{1}{n-i+1}$. It follows that $\alpha_h=1$ for the unique history that chooses action $a_{n-i}$.
\end{proof}

\smallskip\noindent\textbf{Contraction of MECs.} In order to solve percentile queries, we sometimes reduce our problems to multiple reachability by first contracting MECs of given MDPs,
which is a known technique~\cite{DeAlfaro-phd97}. We define a transformation of MDP~$M$ to represent the events $\Inf(\rho) \subseteq C$ for~$C \in \mecs(M)$
as fresh states. Intuitively, all states of a MEC will now lead to an absorbing state that will abstract the behavior of the MEC.

Consider~$M$ with~$\mecs(M)=\{C_1,\ldots,C_m\}$. We define MDP~$M'$ from~$M$ as follows. For each~$C_i$, we add state~$s_{C_i}$ and action~$a^*$ from each state~$s \in C_i$ to~$s_{C_i}$.
All states~$s_{C_i}$ are absorbing, and $A(s_{C_i}) = \{a^*\}$. 
The probabilities of events $\Inf(\rho) \subseteq C_i$ in~$M$ are captured by the reachability of states $s_{C_i}$ in~$M'$, as follows. We use the classical temporal logic symbols $\Diamond$ and $\Box$ to represent the \textit{eventually} and \textit{always} operators respectively.

\begin{lemma}
  \label{lemma:mec-contract}
	Let~$M$ be an MDP and $\mecs(M)=\{C_1,\ldots,C_m\}$. For any strategy~$\sigma$ for $M$, there exists a strategy~$\tau$ for~$M'$ such that
  for all $i\in\{1,\ldots,m\}$, $\pr_{M,\initState}^\sigma[\Diamond\Box C_i] =
	\pr_{M',\initState}^\tau[\Diamond s_{C_i}]$. 
  Conversely, for any strategy~$\tau$ for~$M'$ such that $\sum_{i=1}^m
	\pr_{M',\initState}^\tau[\Diamond s_{C_i}]=1$, there exists~$\sigma$
  such that for all
  $i$, $\pr_{M,\initState}^\sigma[\Diamond\Box C_i] = \pr_{M',s}^\tau[\Diamond s_{C_i}]$. 
\end{lemma}

\begin{proof}
  Consider any strategy~$\tau$ in~$M'$ with $\sum_{i=1}^m \pr_{M',s}^\tau[\Diamond s_{C_i}]=1$.
  We define strategy~$\sigma$ for~$M$ by imitating~$\tau$ except that whenever it chooses action~$a^*$ from some state~$s \in C_i$, we switch 
  to a memoryless strategy that surely stays inside~$C_i$. The desired equality follows. The other direction was proved in \cite[Lemma 4.6]{BBCFK-lmcs14}.
\end{proof}

Under some hypotheses, solving multi-constraint percentile problems on ECs
yield the result for all MDPs, by  the transformation of Lemma~\ref{lemma:mec-contract}.
We prove a general theorem and then derive particular results as corollaries.

\begin{theorem}
  \label{thm:general}
  Consider all prefix-independent payoffs~$f$ such that for all strongly connected MDPs~$M$, 
  and all $(l_i,v_i)_{1\leq i\leq q} \in \{1,\ldots,\dimension\}\times \mathbb{Q}$, 
there exists a strategy~$\sigma$ such that
  \[\forall i \in \{1,\ldots,\dimension\}, \pr_{M,\initState}^\sigma[f_{l_i} \geq
	v_i] \geq \sup_{\tau} \pr_{M,\initState}^\tau[f_{l_i}\geq v_i].\]
  If the value $\sup_{\tau}$ is computable in polynomial time for strongly connected MDPs, then the multi-constraint percentile problem for~$f$ is decidable in polynomial time.
  Moreover, if strategies achieving $\sup_\tau$ for strongly connected MDPs use $\mathcal{O}(g(M,q))$ memory,
  then the overall strategy use $\mathcal{O}(g(M,q))$ memory.
\end{theorem}

The hypotheses are crucial. Essentially, we require payoff functions that are prefix-independent and for which strategies can be combined easily inside MECs (in the sense that if two constraints can be satisfied independently, they can be satisfied simultaneously). Prefix-independence also implies that we can forget about what happens before a MEC is reached. Hence, by using the MEC contraction, we can reduce the percentile problem to multiple reachability for absorbing target states.

\begin{proof}
  Consider an MDP~$M$, an initial state~$\initState$, and an instance of the multi-constraint percentile problem $(l_i,v_i,\alpha_i)_{1\leq i\leq q}$ for payoff function~$f$.
  
  Let~$C_1,\ldots,C_m$ denote the MECs of~$M$.
  Consider the MDP~$M'$ of Lemma~\ref{lemma:mec-contract}.
  For each~$1\leq j\leq m$, let~$\vec{u(j)}$ denote the component-wise optimal value vector achievable inside~$C_j$ and~$\sigma_j$ a witness strategy,
  which can be computed by hypothesis in polynomial time. Note that because $f$ is prefix-independent and each~$C_j$ strongly connected,
  it follows by~\cite{Chatterjee-tcs07} that $\sup_{\tau} \pr_{M,\initState}^\tau[f_{l_i}\geq v_i] \in \{0,1\}$.
  In fact, for strongly connected MDPs, if a prefix-independent measure can be satisfied with nonzero probability, then there exists a state from which 
  the threshold can be satisfied with probability~$1$.
  Moreover, because the MDP is strongly connected, such a state is reachable almost-surely from any other state.

  Now, for each~$1\leq i\leq q$, we define~$T_i = \{ s_{C_j} \mid 1\leq j\leq m, \sup_{\tau} \pr_{M,\initState}^\tau [f_{l_i} \geq u(j)_i] = 1\}$,
  where states $s_{C_i}$ were defined in Lemma~\ref{lemma:mec-contract}.
  We solve the multiple reachability with absorbing targets~$T_1,\ldots,T_m$ in~$M'$, 
  with probabilities~$\alpha_1,\ldots,\alpha_q$, by Theorem~\ref{thm:absorbing-reachsafe}.
  All computations are in polynomial time. We now establish the connection with the multi-constraint percentile problem.

  Assume there is a strategy~$\tau$ in~$M'$ witnessing the multiple reachability problem. 
  Recall that the strategy~$\sigma$ for~$M$ of Lemma~\ref{lemma:mec-contract} derived from~$\tau$ consists in following~$\tau$ until an action~$a^*$ is taken, upon which 
  one switches to an arbitrary strategy inside the current MEC. Let us define strategy~$\sigma'$ in this manner,
  by switching to the optimal strategy~$\sigma_j$, where~$C_j$ is the current MEC. It follows that, for each~$1\leq i\leq q$, the probability of switching to~$\sigma_j$ for~$j$ such that
  $s_{C_j} \in T_i$ is at least~$\alpha_i$. But such~$\sigma_j$ satisfy $f_{l_i}\geq v_i$ almost-surely in~$C_j$. Because~$f$ is prefix-independent,
  we get $\pr_{M,\initState}^{\sigma'}[f_{l_i}\geq v_i] \geq \alpha_i$.
  Strategy~$\sigma'$ thus just needs one additional bit compared to~$\sigma$ to remember whether it has switched to a strategy inside a MEC.

  Conversely, consider any strategy~$\sigma$ satisfying the multi-constraint percentile problem for~$f$. Let~$\tau$ be the strategy 
  for~$M'$ given by Lemma~\ref{lemma:mec-contract}.
  We have,
  \[\begin{array}{ll}
  \pr_{M,\initState}^\sigma[ f_{l_i} \geq v_i ] &= \sum_{j=1}^m \pr_{M,s}^\sigma[ f_{l_i} \geq v_i \mid \Diamond\Box C_j] \pr_{M,\initState}^\sigma[\Diamond\Box C_j]\\
  &= \sum_{j=1}^m \pr_{M,s}^\sigma[ f_{l_i} \geq v_i \mid \Diamond\Box C_j] \pr_{M',\initState}^\tau[\Diamond s_{C_j}]
\end{array}
\]
Furthermore, we $\pr_{M,s}^\sigma[ f(w_i) \geq v_i \mid \Diamond\Box C_j]>0$ implies that
$\sup_{\sigma'} \pr_{C_j,s}^{\sigma'}[f(w_i)\geq v_i] = 1$ as observed above. It follows that,
$\pr_{M,s}^\sigma[ f(w_i) \geq v_i \mid \Diamond\Box C_j] \leq \sup_{\sigma'}\pr_{C_j,s}^{\sigma'}[f(w_i)\geq v_i)]$.
We obtain 
\[\alpha_i \leq \pr_{M,\initState}^\sigma[ f(w_i)\geq v_i ]\leq \sum_{j : C_j \in T_i} \pr_{M',\initState}^\tau[\Diamond s_{C_j}],\]
which concludes the proof.
\end{proof}

\section{Inf, Sup, LimInf, LimSup Payoff Functions}
\label{section:quant-reg}
\label{subsection:single-simple}

\smallskip\noindent\textbf{Single-dimensional queries.} We give polynomial-time algorithms for the \emph{single}-dimensional multi-cons\-traint percentile problems.
For $\inf$ and $\sup$ we reduce the problem to nested multiple reachability, while $\liminf$ and $\limsup$ are solved by applying Theorem~\ref{thm:general}. 

\begin{theorem}
\label{thm:quant_reg_single_dim}
  The single-dimensional multi-constraint percentile problems can be solved in polynomial time in the problem size
  for $\inf$, $\sup$, $\liminf$, and $\limsup$ functions. Computed strategies use memory linear in the query size for~$\inf$ and $\sup$,
  and constant memory for $\liminf$ and~$\limsup$.
\end{theorem}
\begin{proof}
  Let us fix MDP~$M$, and a starting state~$\initState$. We start with $\sup$. The result will be derived by Theorem~\ref{thm:monotonic-reach}.
  Consider an instance $(v_i,\alpha_i)_{1\leq i\leq q}$ of the problem, where we assume w.l.o.g. that $v_1\leq \ldots\leq v_q$.
  To simplify the argument, let us assume that weights are assigned to states rather than edges;
  one can always transform the given MDP (in polynomial time) to ensure this.
  We define $T_i$ as the set of states whose weights are at least~$v_i$. The problem
  of ensuring that $\pr_{M,\initState}^\sigma[\sup \geq v_i]\geq \alpha_i$ by some strategy~$\sigma$ 
  is then equivalent to the nested reachability problem with targets~$T_1 \supseteq T_2 \supseteq \ldots \supseteq T_q$.
  The problem can thus be solved in polynomial time by Theorem~\ref{thm:monotonic-reach}.
  The resulting strategies use linear memory by this theorem.

  For $\Inf$, consider an instance $(v_i,\alpha_i)_{1\leq i\leq q}$ of the problem with~$v_1\leq \ldots\leq v_q$.
  We make $q+1$ copies of~$M$, each named~$M_i$. For any state~$s$ of~$M$, we refer as~$s(i)$ to the corresponding copy in~$M_i$.
  The starting state is~$\initState(q+1)$.
  In each~$M_i$, any edge from~$s(i)$ to~$t(i)$ 
  of weight~$w < v_i$ is redirected to~$t(j)$, where~$j\leq i$ is the least index such that~$w < v_j$.
  Intuitively, if the run is in~$M_i$, this means that the current history~$h$ violates all constraints $\inf\geq v_j$ for all~$j=i\ldots q$.  
  For each $1\leq i\leq q$, let $\Safe_M^i$ denote the set of states of $M_{i+1}$ from which $\inf\geq v_i$ can be surely satisfied.
  These sets can be computed in polynomial time.
  Now, we add an absorbing state $\top_i$ for each copy $M_i$, and a fresh action $a^\top$ deterministically leads to~$\top_i$ from all states $\Safe_M^i$.
  Note that~$M'$ has size $\mathcal{O}(q|M|)$. Define $T_i = \{\top_i,\ldots,\top_q\}$.
  Now the multiple reachability instance $(T_i,\alpha_i)_{1\leq i\leq q}$ on~$M'$ (with absorbing target states) is equivalent to 
  the multi-constraint percentile problem for inf. In fact, from any strategy~$\sigma$ satisfying the reachability probabilities in~$M'$, one can clearly construct
  a strategy for $M$ by following~$\sigma$ until some state~$\top_i$ is reached, and then switching to a strategy that is surely safe for the objective $\inf\geq v_i$.
  Conversely, given a strategy $\sigma$ for~$M$ satisfying the multi-constraint percentile query, we define strategy~$\sigma'$ for~$M'$ by following~$\sigma$, and as soon as
  some state~$\Safe_M^i$ is reached, going to~$\top_i$. We argue that for each $1\leq i\leq q$, the probability of reaching $\cup_{j=i}^q \Safe_M^j$ is at least~$\alpha_i$
  in~$M$ under~$\sigma$.
  In fact, otherwise, with probability more than $1-\alpha_i$, the play always stays outside this set. Because $\inf\geq v_i$ is a safety property, this means that 
  the property is violated with probability more than $1-\alpha_i$, which is a contradiction.
  The resulting strategy uses linear memory since $M'$ is made of~$q+1$ copies of~$M$.

  For liminf and limsup, consider an instance $(v_i,\alpha_i)_{1\leq i\leq q}$ of the problem, where we assume~$v_1\leq \ldots\leq v_q$.
  We are going to use Theorem~\ref{thm:general}.

  The problem is easy to solve for an end-component~$C$: for each~$i=1\ldots q$, one removes all edges with weight smaller than~$v_i$, and checks 
  if there is an end-component~$C'$ included in~$C$. Consider the largest~$i$ with this property. We know that from any state of~$C$,
  $C'$ can be reached almost-surely, and one can stay inside~$C'$ surely. Then by such a strategy, all constraints
  $\liminf\geq v_j$ for $j=1\ldots i$ are satisfied almost-surely, 
  while other constraints are violated almost-surely by any strategy that stays inside~$C$.
  Optimal strategies inside strongly connected MDPs are thus memoryless.
  We satisfy the hypotheses of Theorem~\ref{thm:general},
  which yields a polynomial-time algorithm.

  The limsup case is solved similarly: In each end-component~$C$, if~$i_0$ denotes the largest~$v_{i_0}$ such that some edge of~$C$ has weight~$\geq v_{i_0}$, then all constraints $\limsup\geq v_j$ 
  for~$j=1\ldots i_0$ can be satisfied almost-surely, and no other constraint is satisfied by any strategy.
  
  The memory usage follows from Theorem~\ref{thm:general}.
\end{proof}

\smallskip\noindent\textbf{Multi-dimensional queries.}
\label{subsection:multi-simple}
We show that all multi-dimensional cases can be solved in time polynomial in the model size and exponential in the query size by a reduction to multiple LTL objectives studied in~\cite{EKVY-lmcs08}. 
Our algorithm actually solves a more general class of queries, where the payoff function can be different for each query.

\begin{theorem}
\label{thm:quant_reg_multi_dim}
  The multi-dimensional percentile problems for $\sup$, $\inf$, $\limsup$ and $\liminf$ can be solved in time polynomial in the model size and exponential in the query size,
  yielding strategies with memory exponential in the query.
\end{theorem}

Given an MDP~$M$, for all~$i\in\{1\ldots q\}$ and value~$v_i$, we denote~$A_{l_i}^{\geq v_i}$ the set of actions of~$M$ whose rewards are at least~$v_i$.  We fix an MDP~$M$. 
 For any constraint $\phi_i \equiv f(w_{l_i})\geq v_i$, we define an LTL formula denoted~$\Phi_i$ as follows. For $f_{l_i}=\inf$, $\Phi_i = \Box A_{l_i}^{\geq v_i}$, for $f_{l_i}=\sup$, $\Phi_i = \Diamond A_{l_i}^{\geq v_i}$, for $f_{l_i} = \liminf$, $\Phi_i = \Diamond \Box A_{l_i}^{\geq v_i}$, and for $f_{l_i} = \limsup$, $\Phi_i = \Box \Diamond A_{l_i}^{\geq v_i}$.
The percentile problem is then reduced to queries of the form $\wedge_{i=1}^q \pr_{M,\initState}^\sigma[\Phi_i]\geq \alpha_i$, 
for which an algorithm was given in~\cite{EKVY-lmcs08} that takes time polynomial in $|M|$ and doubly exponential in~$q$. We improve this complexity since our formulae have bounded sizes.

\begin{lemma}
  For all constraints $\phi_1,\ldots,\phi_q$, and probabilities~$\alpha_1,\ldots,\alpha_q$,
  there exists a strategy~$\sigma$ such that $\bigwedge_{1\leq i\leq q} \pr_{M,\initState}^\sigma[\phi_i]\geq \alpha_i$
  if, and only if, there exists a strategy~$\tau$ such that $\bigwedge_{1\leq i \leq q} \pr_{M,\initState}^\tau[\Phi_i]\geq \alpha_i.$
  This can be decided in time polynomial in the model and exponential in the query, and computed strategies use exponential memory in the query.  
\end{lemma}
\begin{proof}
  The correspondence between LTL formulae and weighted objectives are clear by construction.
  The complexity follows from \cite{EKVY-lmcs08}. In fact, 
  if~$D_{i}$ denotes the subset construction applied to B\"uchi automata recognizing~$\Phi_i$, 
  then the multiple objective LTL problem can be solved in time polynomial in the size of
  the product of~$M$ with $D_1,\ldots,D_q$. But for each formula, a B\"uchi automaton of size~$2$ can be constructed; it follows that the algorithm of~\cite{EKVY-lmcs08} has complexity polynomial in~$|M|$ and exponential in~$q$. 
  The computed strategy is memoryless on the product of~$M$ and~$D_1,\ldots,D_q$, thus the corresponding strategy for~$M$ has memory $D_1\times\ldots\times D_q$ which is a single exponential in~$q$.
\end{proof}

The exponential dependency on the query size cannot be avoided in general unless $\PTIME = \PSPACE$, as shown in the following theorem.

\begin{theorem}
\label{thm:quant_reg_multi_dim_pspace}
  The multi-dimensional percentile problem is \PSPACE-hard for $\sup$.
\end{theorem}
\begin{proof}
  Multiple reachability with arbitrary target sets can be encoded as the multi-dimensional multi-constraint percentile problem for $\sup$ with weights from $\{0,1\}$, as we show now.
  Given MDP~$M$ and targets~$T_1,\ldots,T_q$, we define~$M'$ by duplicating states as follows.
  For each state~$s$, we create a new state~$s^\bis$. All actions leaving~$s$ now leave from~$s^\bis$,
  and a single action~$a^s$ deterministically leads from~$s$ to~$s^\bis$.
  It is clear that there is a bijection between the strategies of~$M$ and those of~$M'$ and that they induce the same reachability probabilities for any subset of states of~$M$.
  We define a $q$-dimensional weight function on~$M'$ that takes values in $\{0,1\}$. At any state~$s$, 
  $w_i(a^s) = 1$ if, and only if~$s \in T_i$. All other actions have value~$0$.
  In other terms, the weight function assigns~$1$ to dimension~$i$ if the target set~$T_i$ is seen. Since the payoff function is $\sup$, along any history
  the dimensions that have the value~$1$ are exactly the target sets that have been satisfied.
  For any probabilities~$\alpha_1,\ldots,\alpha_q$, $\exists \sigma, \forall i=1\ldots q, \pr_{M',\initState}^\sigma[\sup_i\geq 1]\geq \alpha_i$,
  if, and only if $\exists \sigma, \pr_{M,s}^\sigma[\Diamond T_i]\geq \alpha_i]$.
  \PSPACE-hardness follows from Theorem~\ref{thm:asreach}.
\end{proof}

Nevertheless, the complexity can be improved for $\limsup$ functions, for which we give a polynomial-time algorithm by an application of Theorem~\ref{thm:general}.

\begin{theorem}
  \label{lemma:limsup}
  The multi-dimensional percentile problem for $\limsup$ is solvable in polynomial time.
  Computed strategies use constant-memory.
\end{theorem}

\begin{proof}
  The problem is easy to solve if~$M$ is strongly connected.
  In fact, if for some~$i$, $M$ contains no action whose weight at dimension~$l_i$ is at least~$v_i$, 
  then no strategy satisfies~$\limsup_{l_i} \geq v_i$ with positive probability. Conversely, let~$I \subseteq \{1,\ldots,q\}$ such that for each~$i \in I$,
  $M$ contains an edge~$e$ with~$w_{l_i} (e) \geq v_i$. Then, there is a strategy~$\sigma$ satisfying
  $\wedge_{i \in I} \pr_{M,s_0}^\sigma[ \limsup_{l_i} \geq v_i ] = 1$. In fact, because~$M$ is strongly connected
  each state and action can be eventually reached almost-surely from any state. 
  In particular, the strategy which assigns uniform probabilities to all available actions visits all states infinitely often almost-surely.
  
  Thus, we satisfy the hypotheses of Theorem~\ref{thm:general}, and a polynomial-time algorithm follows.
\end{proof}

\noindent
The exact query complexity of the $\liminf$ and $\inf$ cases are left open. 

\section{Mean-payoff}
\label{section:mp}
We consider the multi-constraint percentile problem both for~$\mpinf$ and~$\mpsup$.
We will see that strategies require infinite memory in both cases, in which case
it is known that the two payoff functions differ.
The \emph{single-constraint} percentile problem was first solved in~\cite{FKR-ieee95}.
The case of multiple dimensions was mentioned as a challenging problem but left open.
We solve this problem thus generalizing the previous work. 

\subsection{The Single-Dimensional Case}

We start with a polynomial-time algorithm for the single-dimensional case
obtained via Theorem~\ref{thm:general}, thus extending the results of~\cite{FKR-ieee95} to multi-constraint percentile queries.

\begin{theorem}
  The single dimensional multi-constraint percentile problems for payoffs~$\mpinf$ and~$\mpsup$ are equivalent and solvable in polynomial time.
  Computed strategies use constant memory.
\end{theorem}

\begin{proof}
  Let~$C_1,\ldots,C_m$ be the MECs of a given MDP~$M$.
  If we define $v^*(C_i) = \sup_{\sigma \in \Sigma} \expect_{C_i,s}^\sigma[\mpinf] = \sup_{\sigma \in \Sigma}\expect_{C_i,s}^\sigma[\mpsup]$,
  then for each~$1\leq i\leq m$, there exists a strategy~$\sigma_i$, computable in polynomial time, with the property
  $\pr_{M,s}^{\sigma_i}[\mpinf = v^*(C_i)] = 1$~\cite{Puterman-wiley94}. In other terms, optimal strategies exist for single dimensional mean-payoff,
  and the optimal value can be achieved almost-surely inside strongly connected MDPs.
  In contrast,  no value greater than the optimal value can be achieved with positive probability. 
  The polynomial-time algorithm then follows from Theorem~\ref{thm:general}.

  The equivalence between~$\mpinf$ and~$\mpsup$ follows from the fact that they are equivalent inside MECs since
  memoryless strategies exist, and that the strategy of Theorem~\ref{thm:general} almost-surely eventually switches to 
  an optimal strategy for a MEC.
\end{proof}

\subsection{Percentiles on Multi-Dimensional $\mpsup$}

Let $\expect_{M,\initState}^\sigma[\mpsup_i]$ be the \emph{expectation} of~$\mpsup_i$ under strategy~$\sigma$,
and $\Val_{M,\initState}^*(\mpsup_i)=\sup_{\sigma} \expect_{M,\initState}^\sigma[\mpsup_i]$, computable in polynomial time~\cite{Puterman-wiley94}.
We solve the problem inside ECs, then apply Theorem~\ref{thm:general}.
It is known that for strongly connected MDPs, for each~$i$, some strategy~$\sigma$ satisfies $\pr_{M,\initState}^\sigma[\mpsup_i = \Val^*_{M,\initState}(\mpsup_i)]=1$,
and that for all strategies~$\tau$, $\pr_{M,\initState}^\tau[\mpsup_i>v]=0$ for all $v>\Val^*_{M,\initState}(\mpsup_i)$.
By switching between these optimal strategies for each dimension, with growing
intervals, we prove that
for strongly connected MDPs, a single strategy can simultaneously optimize $\mpsup_i$ on \emph{all} dimensions.

We first recall the following result on the convergence speed of optimal memoryless strategies
in MDPs.

\begin{lemma}[\cite{Tracol09}]
  \label{lemma:boundk}
  Let~$M$ be any single-dimensional weighted MDP, $v^* = \sup_{\sigma} \expect_{M,\initState}^\sigma[\mpsup]$,
  and $\sigma$ an optimal memoryless strategy with $v^* = \expect_{M,\initState}^\sigma[\mpsup]$.
  For all~$\varepsilon>0$ and~$\eta>0$, there exists $K_0>0$ such that  for all $K\geq K_0$,
  \(
  \pr_{M,\initState}^\sigma[\{s_1a_1s_2a_2\ldots \mid \frac{1}{K} \sum_{i=1}^K w(a_i) \geq v^* - \varepsilon\}] \geq 1 - \eta.
  \)
\end{lemma}

We now show that strongly connected multi-dimensional MDPs, a single strategy can simultaneously optimize $\mpsup$, on \textit{all} dimensions.

\begin{lemma}
  For any strongly connected MDP~$M$,  there is an infinite-memory strategy~$\sigma$ such that $\forall i\in\{1,\ldots,d\}$, $\pr_{M,\initState}^\sigma[\mpsup_i \geq \Val_{M,\initState}^*(\mpsup_i)]=1$.
\end{lemma}

\begin{proof}
  Let us write $v_i^*=\Val_{M,\initState}^*(\mpsup_i)$, and let~$\sigma_i$ be a memoryless 
  optimal strategy for this dimension.
  We define a strategy that switches between these strategies~$\sigma_i$ with growing time intervals. 
  We fix $\eta \in (0,1)$, and define the sequence $\varepsilon_i = \frac{1}{i}$.
  Let $t_1 = 1$. For~$i\geq 2$, if~$K_0$ the bound given by Lemma~\ref{lemma:boundk}
  for~$\varepsilon_i$ and~$\eta$, we choose~$t_i\geq K_0$ such that ${t_i}\geq i^2 \sum_{j=1}^{i-1}t_{j}$.
  Strategy~$\sigma$ is defined by running~$\sigma_j$ during~$t_i$ steps where $j = (i \mod d) + 1$.
  Let us define $\alpha_i = \sum_{j=1}^i t_j$.
  
  We now prove that~$\sigma$ achieves the optimal value at each dimension with probability~$1$.
  Let~$A_i$ denote the random variable of the $i$-th action of an execution for a given MDP, initial state, and strategy.
  Fix any dimension $k \in \{1,\ldots,d\}$. For any~$i$ such that $(i \mod d) + 1 = k$,
  between steps $\alpha_{i-1}+1$ and~$\alpha_{i}$, strategy~$\sigma_k$ is memoryless, and by Lemma~\ref{lemma:boundk},
  we have
  \[
  \pr_{M,\initState}^\sigma[\frac{1}{t_i} \sum_{j=\alpha_{i-1}+1}^{\alpha_i} w(A_j) \geq v_k^* - \varepsilon_i] \geq \eta.
  \]
  Observe that $\frac{t_i}{\alpha_i} = \frac{i^2}{i^2+1}$, and $\frac{\alpha_{i-1}}{\alpha_i} = \frac{1}{i^2+1}$.
  So, with probability~$\eta$, we get
  \[
  \begin{array}{ll}
    \frac{1}{\alpha_i} \sum_{j=1}^{\alpha_i} w(A_j) &=\frac{1}{\alpha_i}\sum_{j=1}^{\alpha_{i-1}} w(A_j) + \frac{1}{\alpha_i}\sum_{j=\alpha_{i-1}+1}^{\alpha_i} w(A_j)\\
    &\geq \frac{\alpha_{i-1}}{\alpha_i} \min_{a \in A} w(a) + \frac{t_i}{\alpha_i} (v_i^* - \varepsilon_i).\\
    &\geq \frac{1}{i^2+1} \min_{a \in A} w(a) + \frac{i^2}{i^2+1} (v_i^* - \varepsilon_i).\\
  \end{array}
  \]
  This means that for any~$\varepsilon>0$, there exists~$i_0$ such that for all~$i\geq i_0$ with $(i \mod d) + 1  =k$, we have
  \[
  \pr_{M,\initState}^\sigma[\frac{1}{\alpha_i} \sum_{j=1}^{\alpha_i} w(A_j) \geq v_i^* - \varepsilon] \geq \eta,
  \]
  so $\pr_{M,\initState}^\sigma[\mpsup_k \geq v_k^*-\varepsilon] = 1$ for all~$\varepsilon>0$.
  It follows that $\pr_{M,\initState}^\sigma[\mpsup_k \geq v_k^*] = 1$ for all dimensions~$k$.
\end{proof}

Thanks to the above lemma, we fulfill the hypotheses of Theorem~\ref{thm:general}, and we obtain the following theorem.

\begin{theorem}
\label{thm:mpsup}
  The multi-dimensional percentile problem for~$\mpsup$
  is solvable in polynomial time. Strategies use infinite-memory, which is necessary.
\end{theorem}

To see that infinite-memory strategies are necessary for $\mpsup$, consider the MDP of Fig.~\ref{fig:mpsup-memory} where thresholds $v_1=v_2=1$ can be achieved almost-surely by the above theorem, but not by any finite-memory strategy. The proof is identical to the case of maximizing the expectation in~\cite[Lemma 7]{CDHR-fsttcs10} where it is proved for the case of deterministic MDPs (i.e., automata).

\begin{figure}[ht]
  \centering
  \scalebox{1.2}{
  \begin{tikzpicture}
    \tikzstyle{every state}=[node distance=2cm,minimum size=15pt, inner sep=1pt];
    \node[state] at (0,0) (s){$s$};
    \node[state, right of=s] (t) {$t$};
    \path[-latex'] 
    (s) edge[bend left] (t)
    (t) edge[bend left] (s)
    (s) edge[loop left, looseness = 2, distance = 6mm] node{$(1,0)$} (s)
    (t) edge[loop right, looseness = 2, distance = 6mm] node{$(0,1)$} (t);
  \end{tikzpicture}}
  \caption{Infinite-memory strategies are necessary for $\mpsup$.}
  \label{fig:mpsup-memory}
\end{figure}

\subsection{Percentiles on Multi-Dimensional $\mpinf$}
\label{section:multidim-mpinf}

In contrast with the $\mpsup$ case, our algorithm for $\mpinf$ is more involved, and requires new techniques.
In fact, the case of end-components is already non-trivial for $\mpinf$, since there is no single strategy 
that satisfies all percentile constraints in general, and one cannot hope to apply Theorem~\ref{thm:general} as we did in previous sections.
We rather need to consider the set of strategies~$\sigma_I$ satisfying \emph{maximal} subsets of percentile constraints; these are called \emph{maximal strategies}.
We then prove that any strategy satisfying all percentile queries can be written as a \emph{linear combination} of maximal strategies, that is,
there exists a strategy which chooses and executes each~$\sigma_I$ following a probability distribution.

For general MDPs, 
we first consider each MEC separately and write down the linear combination with unknown coefficients.
We know that any strategy in a MDP eventually stays forever in a MEC. Thus, we adapt the linear program
of~\cite{EKVY-lmcs08} that encodes the reachability probabilities with multiple targets, which are the MECs here.
We combine these reachability probabilities with the unknown linear combination coefficients, and obtain a linear program 
(Figure~\ref{fig:mpinf-lp}), which we prove to be equivalent to our problem.

\smallskip\noindent\textbf{Single EC.}
Fix a strongly connected $d$-dimensional MDP~$M$ and 
pairs of thresholds $(v_i,\alpha_i)_{1\leq i\leq q}$. We denote each event by~$A_i \equiv \mpinf_i \geq v_i$.
In \cite{BBCFK-lmcs14}, the problem of maximizing the \emph{joint} probability of the events~$A_i$ was solved in polynomial time.
In particular, we have the following for strongly connected MDPs.
\begin{lemma}[{\cite{BBCFK-lmcs14}}]
  \label{lemma:joint-bounds}
	If~$M$ is strongly connected, then there exists~$\sigma$ such that ${\pr_{M,s}^\sigma[\wedge_{1\leq i\leq q} A_i]>0}$ if, and only if
    there exists~$\sigma'$ such that~$\pr_{M,s}^{\sigma'}[\wedge_{1\leq i\leq q} A_i]=1$.
    Moreover, this can be decided in polynomial time, and for positive instances, for any~$\varepsilon>0$,
    a memoryless strategy~$\tau$ can be computed in polynonomial time in $M$, $\log(v_i)$ and~$\log(\frac{1}{\varepsilon})$, such that
    \(
    \pr_{M,s}^\tau[\wedge_{1\leq i\leq q} \mpinf_i \geq v_i - \varepsilon] = 1.
    \)
\end{lemma}

We give an overview of our algorithm.
Using Lemma~\ref{lemma:joint-bounds}, we define strategy $\sigma_I$ achieving $\pr_{M,s}^{\sigma_I}[\wedge_{i \in I} A_i]=1$
for any maximal subset~$I \subseteq \{1,\ldots,q\}$ for which such a strategy exists.
Then, to build a strategy for the multi-constraint problem, we look for a linear combination of these~$\sigma_I$: given $\sigma_{I_1},\ldots, \sigma_{I_m}$, we choose each~$i_0 \in \{1,\ldots,m\}$ following a probability distribution to be computed, and we run~$\sigma_{I_{i_0}}$.

We now formalize this idea. Let~$\calI$ be the set of maximal~$I$ (for set inclusion) such that some $\sigma_I$
satisfies $\pr_{M,s}^{\sigma_I}[\wedge_{i \in I}A_i]=1$. 
Note that for all $I \in \calI$, and~$j \not\in~I$,
$\pr_{M,s}^{\sigma_I}[\wedge_{i \in I} A_i \land A_j] = 0$. Assuming otherwise would contradict the maximality of~$I$, by Lemma~\ref{lemma:joint-bounds}.
We consider the events $\calA_I = \wedge_{i \in I} A_i \wedge_{i \not \in I}\lnot A_i$
for maximal~$I$.

We are looking for a non-negative family 
$(\lambda_I)_{I \in \calI}$ whose sum equals~$1$ and such that $\forall i\in\{1,\ldots,q\}$, $\sum_{I \in \calI \text{ s.t. } i \in I} \lambda_I \geq \alpha_i$.
This will ensure that if each~$\sigma_I$ is chosen with probability~$\lambda_I$ (among the set $\{\sigma_I\}_{I \in \calI}$); with probability at least~$\alpha_i$,
some strategy satisfying~$A_i$ with probability~$1$ is chosen. So each~$A_i$ is satisfied with probability at least~$\alpha_i$.
This can be written in the matrix notation as
\begin{equation}
  \label{eqn:mpinf-scc-lp}
    \calM \vec{\lambda} \geq \vec{\alpha},\quad 0\leq \vec{\lambda},\quad \mathbf{1}\cdot
		\vec{\lambda} =1,
\end{equation}
where~$\calM$ is a $q \times |\calI|$ matrix with~$\calM_{i,I} = 1$ if~$i \in I$, and~$0$ otherwise.

\begin{lemma}
  For any strongly connected MDP~$M$,
  and an instance $(v_i,\alpha_i)_{1\leq i\leq q}$ of the multi-constraint percentile problem for~$\mpinf$,
  \eqref{eqn:mpinf-scc-lp} has a solution if, and only if there exists a strategy~$\sigma$ satisfying
  the multi-constraint percentile problem.
\end{lemma}

\begin{proof}
  Assume~\eqref{eqn:mpinf-scc-lp} and consider the strategy $\sum_{I \in \calI} \lambda_I \sigma_I$, which means that
  at the beginning of the run, we choose each set~$I$ with probability~$\lambda_I$, and run
  $\sigma_I$. Clearly, the probability of satisfying $A_i$ is at least the probability of running
  a strategy $\sigma_I$ such that~$i \in I$, which is $\sum_{I \in \calI : i \in I} \lambda_I$.
  The result follows.

  Conversely, let~$\sigma$ denote a strategy satisfying $\pr_{M,\initState}^\sigma[A_i]\geq \alpha_i$ for all~$i$.
  Let us consider all events~$\calA_I$ including non-maximal~$I$.
  The events~$\calA_I$ are disjoint and we have $A_i = \cup_{I : i \in I} \calA_I$.
  It follows that $\pr_{M,\initState}^\sigma[A_i] = \sum_{I: i \in I} \pr_{M,\initState}^\sigma[\calA_I]
  = \sum_{I: i \in I}\pr_{M,\initState}^\sigma[\calA_I] \cdot \pr_{M,\initState}^{\sigma_I}[\calA_I]$
  since $\pr_{M,\initState}^{\sigma_I}[\calA_I] = 1$ by definition.
  In order to derive a probability distribution on \emph{maximal} subsets only, 
  we define a partition of~$2^{\{1,\ldots,q\}}$ by assigning each non-maximal~$J$ to a maximal~$I \in \calI$
  with~$J\subseteq I$. Formally,
  we consider sets~$\alpha(I) \subseteq 2^{\{1,\ldots,q\}}$ with~$I \in \alpha(I)$, 
   such that for all~$J \in \alpha(I)$, $J \subseteq I$, and $\{\alpha(I)\}_{I \in \calI}$ defines a partition of~$2^{\{1,\ldots,q\}}$.
  For any~$I \in \calI$, we set
  $\lambda_I = \sum_{J \in \alpha(I)} \pr_{M,\initState}^{\sigma}[\calA_J]$.
  This yields a solution of~\eqref{eqn:mpinf-scc-lp}.
\end{proof}

Now \eqref{eqn:mpinf-scc-lp} has size $O(q\cdot 2^q)$, and each subset~$I$ can be checked in time polynomial in the model size.
The computation of $\calI$, the set of maximal subsets, can be carried out in a top-down fashion;
one might thus avoid enumerating all subsets in practice.
We get the following result.

\begin{lemma}
  For strongly connected MDPs, the multi-dimensional percentile problem for $\mpinf$  can be solved in time polynomial 
  in~$M$ and exponential in~$q$. Strategies require infinite-memory in general.
  On positive instances, $2^q$-memory randomized strategies can be computed for the~$\varepsilon$-relaxation of the problem
  in time polynomial in~$|M|, 2^q, \max_i\big(\log(v_i), \log(\alpha_i)\big), \log(\frac{1}{\varepsilon})$.
\end{lemma}

\begin{proof}
  The first statement is clear from the two previous lemmas, since 
  \eqref{eqn:mpinf-scc-lp} can be solved in time polynomial in~$M$ and exponential in~$q$.
  For the $\varepsilon$-relaxation problem, notice that
  once we compute the set~$\calI$ and solve \eqref{eqn:mpinf-scc-lp}, for any set~$I \in \calI$, we compute in polynomial time a randomized strategy $\sigma_I$ ensuring $A_i^\varepsilon = \wedge_{i \in I} \mpinf_i \geq v_i -\varepsilon$. This can be done 
  as in~\cite{BBCFK-lmcs14}. Then the strategy choosing randomly each $\sigma_I$ with probability $\sigma_I$  ensures all bounds up to $\varepsilon$ (i.e., $v_i-\varepsilon$).

  The need for infinite memory was proved in~\cite[Section 5]{BBCFK-lmcs14} for the problem of
  ensuring thresholds $\pr_{M,\initState}^\sigma[ \mpinf_1 \geq v_1 \land \ldots \mpinf_2 \geq v_2 ] \geq \alpha$ for thresholds $v_1,v_2$ and probability~$\alpha$.
  It was proved that on the MDP of Fig.~\ref{fig:mpsup-memory}, $v_1=v_2=0.5$ and~$\alpha=1$ can be ensured by an infinite-memory strategy and that
  finite-memory strategies can only achieve these thresholds with probability~$0$.
  Now, if the multi-constraint percentile query $\pr_{M,\initState}^\sigma[ \mpinf_1\geq v_1] \geq 0.6 \land \pr_{M,\initState}^\sigma[\mpinf_2\geq v_2]\geq 0.6$ 
  has a solution by a strategy~$\sigma$, then we must have $\pr_{M,\initState}^\sigma[\mpinf_1\geq v_1\land \mpinf_2\geq v_2]\geq 0.2$ (this simply follows from the fact that
  $0.6+0.6=1.2$). Therefore~$\sigma$ must use infinite-memory.  
\end{proof}

\smallskip\noindent\textbf{General MDPs.} Given MDP~$M$, let us consider~$M'$ given by Lemma~\ref{lemma:mec-contract}.
We start by analyzing each maximal EC~$C$ of~$M$ as above, and compute the sets $\calI^C$ of maximal subsets.
We define a variable $\lambda_I^C$ for each~$I \in \calI^C$, and also $y_{s,a}$ for each state~$s$ and action~$a \in A'(s)$. Recall that
$A'(s) = A(s) \cup\{a^*\}$ for states~$s$ that are inside a MEC, and~$A'(s) = A(s)$ otherwise.
Let $S_{\mecs}$ be the set of states of~$M$ that belong to a MEC.
We consider the linear program (L) of Fig.~\ref{fig:mpinf-lp}.

\vspace{-5mm}
\begin{figure}[th]
	\small
  \begin{align}
    \label{eqn:general-lp}
      \mathbf{1}_{\initState}(s) + \sum_{s' \in S, a \in A(s')}y_{s',a}\delta(s',a,s) = \sum_{a \in A'(s)}y_{s,a}, 
      \qquad \forall s \in S,\\
      \sum_{s \in S_{\text{MEC}}}y_{s,a^*} = 1,\\
      \sum_{s \in C}y_{s,a^*} = \sum_{I \in \calI^C} \lambda_I^C,  \quad \forall C \in \mecs(M), \label{eqn:thirdline}\\
      {\lambda^C_I}\geq 0, \quad \forall C \in \mecs(M), \forall I \in \calI^C, \label{eqn:fourthline}\\
      \sum_{C \in \mecs(M)} \sum_{I \in \calI^C : i \in I} \lambda_I^C \geq \alpha_i, \qquad \forall i=1\ldots d.
      \label{eqn:fifthline}
  \end{align}
\vspace{-2mm}
  \caption{Linear program~(L) for the multi-constraint percentiles for~$\mpinf$.}
  \label{fig:mpinf-lp}
\end{figure}
\vspace{-3mm}

We prove the following main lemma in this section.

\begin{lemma}
  \label{lemma:mpinf-lp}
  The LP~(L) has a solution if, and only if the multi-constraint percentiles problem for $\mpinf$ has a solution. Moreover, the equation has size polynomial in~$M$ and exponential in~$q$.
  From any solution of~(L) randomized finite memory strategies can be computed for the $\varepsilon$-relaxation problem.
\end{lemma}

  The linear program follows the ideas of~\cite{EKVY-lmcs08,BBCFK-lmcs14}. 
  Note that the first two lines of~(L) corresponds to the multiple reachability LP of~\cite{EKVY-lmcs08} for absorbing target states.
The equations encode strategies that work in two phases.
  Variables~$y_{s,a}$ correspond to
  the expected number of visits of state-action~$s,a$ in the first phase. Variable~$y_{s,a^*}$ describes the probability 
  of switching to the second phase at state~$s$.
  The second phase consists in surely staying in the current MEC, so we require $\sum_{s \in S_{\text{MEC}}} y_{s,a^*} = 1$
  (and we will have $y_{s,a^*}=0$ if $s$ does not belong to a MEC).
  In the second phase, we immediately switch to some
	strategy~$\sigma_I^C$ where $C$ denotes the current MEC. Thus, variable~$\lambda_I^C$ corresponds to the probability
  with which we enter the second phase in~$C$ and switch to strategy~$\sigma_I^C$ (see~\eqref{eqn:thirdline}).
  Intuitively, given a solution $(\lambda_I)_I$ computed for one EC by \eqref{eqn:mpinf-scc-lp}, we have the correspondence
  $\lambda_I^C = \sum_{s \in C}y_{s,a^*} \cdot \lambda_I$.
  The interpretation of~\eqref{eqn:fifthline} is that each event~$A_i$ is satisfied with probability at least~$\alpha_i$.

The two following lemmas prove Lemma~\ref{lemma:mpinf-lp}. 

\begin{lemma}
  If (L) has a solution then there exists a strategy for the multi-constraint percentile problem.
  Moreover, from any solution of (L) one can derive in time polynomial in~$M$, $\log(\frac{1}{\varepsilon})$, and exponential in~$q$, a $\mathcal{O}(2^q)$-memory randomized strategy
  solving the $\varepsilon$-relaxation of the multi-constraint percentile problem.
\end{lemma}

\begin{proof}
  Let~$\bar{y_{s,a}},\bar{y_{s,a^*}}, \bar{\lambda_I^C}$ be a solution of (L).
  By \cite[Theorem 3.2]{EKVY-lmcs08}, there exists a memoryless strategy~$\rho$ for~$M'$ such that
  $\pr_{M',\initState}^\rho[\Diamond s_{C}] = \sum_{s \in C} y_{s,a^*}$ for each MEC~$C$,
  and $\sum_{C \in \mecs(M)} \pr_{M',\initState}^\rho[\Diamond s_{C}] = 1$ by the second line.  
  In this strategy, $y_{s,a^*}$ is the probability of going to~$s_{C_i}$ from~$s$.

  For each MEC~$C$, we define the strategy~$\sigma^C$ for~$M$ which, from the states of~$C$,
  executes each strategy
  $\sigma_I$ for~$I \in \calI^C$ with probability $\frac{\lambda_I^C}{\sum_{J \in
      \calI^C}\lambda_J^C}=\frac{\lambda_I^C}{\sum_{s \in C} y_{s,a^*}}$,
  if the denominators are positive, and with an arbitrary distribution otherwise.
  We combine these in a strategy~$\sigma$ for~$M$
  which starts by simulating~$\rho$ until~$\rho$ chooses takes the action~$a^*$, at which point
  $\sigma$ switches to~$\sigma^C$.

  By construction the probability of~$\sigma$ of switching to~$\sigma_I^C$ is
  $\sum_{s \in C}y_{s,a^*} \cdot \frac{\lambda_I^C}{\sum_{s \in C}y_{s,a^*}} = \lambda^C_I$, for any~$C$ and~$I \in \calI^C$.
  Moreover, thanks to the fact that $\sum_{s \in S_{\mecs}} y_{s,a^*} =1$, we know that $\sigma$ will eventually switch to some~$\sigma^C$ almost-surely. 
  Because for all~$I$ and~$C$ such that~$i \not \in I$, the probability of~$\sigma_I^C$ of satisfying~$A_i$ inside~$C$ is~$0$ (see above), 
  we get that the probability of satisfying~$A_i$ under~$\sigma$ is equal to the probability of switching to some~$\lambda^C_I$.
  But thanks to the last line of the program, this quantity is at least~$\alpha_i$.
  Hence,~$\sigma$ satisfies the multi-constraint percentile problem.

  We obtain a strategy for the relaxed problem as follows. Each strategy~$\sigma_I^C$ may be infinite-memory a priori
  but for any~$\varepsilon>0$, we can compute by Lemma~\ref{lemma:joint-bounds}, memoryless randomized strategies~$\tau_I^C$ ensuring~$\calA_I^\varepsilon$
  with probability~$1$. Now, since~$\rho$ is also memoryless, the combined strategy only needs $2^q+1$ memory elements (to store the phase, and which~$I$ it has chosen once in a MEC).
  The result follows.
\end{proof}

\begin{lemma}
  If strategy~$\sigma$ solves the multi-constraint percentile problem for~$\mpinf$, then (L) has a solution.
\end{lemma}
\begin{proof}
  Let~$C_1,\ldots,C_m$ denote the MECs,
  and define~$y_{C_i} = \pr_{M,\initState}^\sigma[\Inf(\rho) = C_i]$ for each~$i$.
  Clearly, we have $\sum_i y_{C_i} = 1$.
  Let~$\rho$ denote the strategy on~$M'$ of Lemma~\ref{lemma:mec-contract} given for~$\sigma$.
  For any action~$a \in C_i$, let~$y_{s,a}$ denote the expected number of times action~$a$ is taken at~$s$ under~$\rho$
  in~$M'$ starting at~$\initState$. 
  Now, \cite[Lemma 3.3]{EKVY-lmcs08} ensures that these variables have finite values and satisfy the 
  first two lines of (L).

  We define strategy~$\sigma'$ for~$M$ which follows~$\rho$ until action~$a^*$ is taken, at which point
  it switches to each strategy $\sigma_I^C$ with probability $\pr_{M,s_0}^{\sigma}[\calA_I\mid \Inf(\rho) = C]$	(these include non-maximal sets~$I$).
  We have that 
  \[
  \pr_{M,\initState}^{\sigma'}[A_i] = \sum_{C \in \mecs(M)} \pr_{M,\initState}^{\sigma'}[A_i \mid \Inf(\rho) = C]\pr_{M,\initState}^{\sigma'}[\Inf(\rho) = C].
  \]
  By definition of~$\sigma'$, the first term in the sum equals
  ${\pr_{M,\initState}^\sigma[A_i\mid \Inf(\rho) = C]}$.
  The second term in the sum is equal to $\pr_{M',\initState}^\rho[\Diamond s_{C}]=
  \pr_{M,\initState}^\sigma[\Inf(\rho) = C]$. 
  It follows that $\pr_{M,\initState}^{\sigma'}[A_i] = \pr_{M,\initState}^{\sigma}[A_i]$.

  Now, to obtain a solution of (L), it remains to get rid of the strategies~$\sigma_I$ for non-maximal subsets~$I$ for each MEC.
  We thus modify once more~$\sigma'$ to obtain~$\sigma''$ as follows. Whenever~$\sigma'$ switches to some strategy~$\sigma_I^C$,
  where~$I$ is not maximal, we rather switch to some $\sigma_J^C$ for some -arbitrarily chosen- maximal~$J\supset I$.
  It is clear that $\pr_{M,\initState}^\sigma[\Inf(\rho)=C] = \pr_{M,\initState}^{\sigma''}[\Inf(\rho) = C]$ and
  $\pr_{M,\initState}^{\sigma''}[\calA_I \mid \Inf(\rho) = C]\geq \pr_{M,\initState}^{\sigma}[\calA_I \mid \Inf(\rho)=C]$ for all~$I \in \calI^C$.
  
  Now, for each~$C$ and~$I \in \calI^C$, we define $\lambda_I^C = \pr_{M,\initState}^{\sigma''}[\calA_I \land \Inf(\rho) = C]$.
  It is easy to verify that $0\leq \lambda_I^C$ and $\sum_{I \in \calI^C} \lambda_I^C = \sum_{s \in C}y_{s,a^*}$, so \eqref{eqn:thirdline} and~\eqref{eqn:fourthline} are satisfied.
  We have $\sum_{I \in \calI^{C_i}} \lambda_I^{C_i} = \pr_{M,\initState}^{\sigma''}[\Inf(\rho)=C_i] = \pr_{M,\initState}^\sigma[\Inf(\rho) = C_i] =  y_{C_i}$
  for all $i=1\ldots d$. Moreover, $\pr_{M,\initState}^{\sigma''}[A_i] = \sum_{C \in \mecs(M)}\sum_{I \in \calI^C: i \in I} \pr_{M,\initState}^{\sigma''}[\calA_I \land \Inf(\rho)= C]$.
  This is at least equal to~$\pr_{M,\initState}^{\sigma}[A_i]$ as we saw above, which is at least~$\alpha_i$ by assumption; hence \eqref{eqn:fifthline} is also satisfied.
\end{proof}

Our results for the multi-dimensional problems with $\mpinf$ are summed up in the next theorem.

\begin{theorem}
\label{thm:mpinf}
  The multi-dimensional percentile problem for $\mpinf$ can be solved in 
  time polynomial in the model, and exponential in the query.
  Infinite-memory strategies are necessary, but exponential-memory (in the query) suffices
  for the $\varepsilon$-relaxation and can be computed with the same complexity.
\end{theorem}

\input{sp}

\input{ds}

\section{Conclusion}

Through this paper, we studied the strategy synthesis problem for \textit{multi-percentile queries} on multi-dimen\-sion\-al MDPs: we considered a wide range of payoff functions from the literature (sup, inf, limsup, liminf, mean-payoff, truncated sum, discounted sum), and established a complete picture of the multi-percentile framework, including algorithms, lower bounds on complexity, and memory requirements. Our results are summed up in Table~\ref{table}.

It is especially interesting to observe that for all payoff functions but the discounted sum, our algorithms require \textit{polynomial time in the size of the model} when the query size is fixed. This is of utmost practical interest as in most applications, the query size (i.e., specification) is typically small while the model (i.e., the system) can be very large. Hence, our algorithms have clear potential to be useful in practice. As future work, we aim to assess their practical efficiency through implementation in tool suites and case studies.

\bibliographystyle{plain}
\bibliography{biblio}

\end{document}

%% file: commands.tex
\newcommand{\states}{\ensuremath{S} }
\newcommand{\state}{\ensuremath{s} }

\newcommand{\supp}{\ensuremath{{\sf Supp}} }
\newcommand{\initState}{\ensuremath{s_{{\sf init}}} }

\newcommand{\weight}{\ensuremath{w} }

\newcommand{\markovProcess}{\ensuremath{{M} }}

\newcommand{\integ}{\ensuremath{\mathbb{Z}} }
\newcommand{\strat}{\ensuremath{\sigma} }
\newcommand*{\pr}{\mathbb{P}}
\newcommand*{\mpsup}{\ensuremath{\overline{\mathsf{MP}}}}
\newcommand*{\mpinf}{\ensuremath{\underline{\mathsf{MP}}}}
\newcommand*{\expect}{\mathbb{E}}

\newcommand{\Inf}{\textrm{\sf Inf}}
\newcommand{\truncatedTarget}{\ensuremath{T} }
\newcommand{\truncatedSum}[1]{\ensuremath{{\sf TS}^{#1}} }
\newcommand{\discSum}[1]{\ensuremath{{\sf DS}^{#1}} }
\newcommand{\discount}{\ensuremath{\lambda} }
\newcommand{\dimension}{\ensuremath{d} }
\newcommand{\queries}{\ensuremath{q} }
\newcommand{\nat}{\ensuremath{\mathbb{N}} }
\newcommand{\rat}{\ensuremath{\mathbb{Q}} }
\newcommand{\twoCM}{\ensuremath{\mathcal{M}} }
\newcommand{\strats}{\ensuremath{\Sigma} }
\newcommand{\round}{\ensuremath{{\sf Round}_{\gamma}} }
\newcommand{\roundedDiscSum}[1]{\ensuremath{{\sf RDS}^{#1}} }
\usetikzlibrary{shapes,arrows}
\newcommand{\query}{\ensuremath{\mathcal{Q}} }

\newcommand\NPTIME{\textrm{\sf NP}}

%% file: intro.tex
\section{Introduction}

{\em Markov decision processes} (MDPs) are central mathematical models for reasoning about (optimal) strategies in {\em uncertain environments}. For example, if rewards (given as numerical values) are assigned to actions in an MDP, we can search for a strategy (policy) that resolves the nondeterminism in a way that the {\em expected mean reward} of the actions taken by the strategy over time is maximized. See for example~\cite{Puterman-wiley94} for a solution to this problem. If we are risk-averse, we may want to search instead for strategies that ensure that the mean reward over time is larger than a given value with a high probability, i.e., a probability that exceeds a given threshold. See for example~\cite{FKR-ieee95} for a solution.

Recent works are exploring several natural extensions of those problems.
First, there is a series of works that investigate MDPs with multi-dimensional
weights~\cite{CMH-stacs06,BBCFK-lmcs14} rather than single-dimensional as it is
traditionally the case. Multi-dimensional MDPs are useful to analyze systems
with {\em multiple objectives} that are potentially conflicting and make
necessary the analysis of trade-offs. For instance, we may want to build a
control strategy that both ensures some good quality of service and minimizes the energy consumption. Second, there are works
that aim at synthesizing strategies enforcing {\em richer properties}. For
example, we may want to construct a strategy that both ensures some minimal
threshold with certainty (or probability one) and a good expectation~\cite{DBLP:conf/stacs/BruyereFRR14}. An illustrative survey of such extensions can be found in~\cite{DBLP:conf/vmcai/RandourRS15}.
 
Our paper participates in this general effort by providing algorithms and complexity results on the synthesis of strategies that enforce {\em multiple percentile constraints}.  A \textit{multi-percentile query} and the associated synthesis problem is as follows: given a multi-dimensionally weighted MDP $M$ and an initial state $\initState$, synthesize a strategy $\sigma$ such that it satisfies the conjunction of $q$ constraints:
\[
\query \coloneqq \bigwedge_{i = 1}^{q}\; \pr_{M,\initState}^\strat\big[f_{l_{i}} \geq v_i\big] \geq
	\alpha_i.
\]
where each $l_i$ refers to a dimension of the weight vectors, each $v_i$ is a value threshold, and $\alpha_i$ is a probability threshold, and $f$ is a payoff function. Each constraint $i$ expresses that the strategy ensures probability at least $\alpha_{i}$ to obtain payoff at least $v_{i}$ in dimension $l_{i}$.

We consider seven payoff functions: sup, inf, limsup, liminf, mean-payoff, truncated sum and discounted sum. This wide range covers most classical functions: our exhaustive study provides a \textit{complete picture} for the new multi-percentile framework and we focus on establishing meta-theorems and connections whenever possible. Some of our results are obtained by reduction to the previous work of~\cite{EKVY-lmcs08}, but for mean-payoff, truncated sum and discounted sum, that are {\em non-regular payoffs}, we need to develop original techniques.

Let us consider some examples. In an MDP that models a stochastic shortest path problem, we may want to obtain a strategy that ensures that the probability to reach the target within $d$ time units exceeds 50 percent: this is a single-constraint percentile query. With a {\em multi-constraint percentile query}, we can impose richer properties on strategies,
for instance, enforcing that the duration is less than $d_1$ in at least 50 percent of the cases, and less than $d_2$ in 95 percent of the cases, with $d_1 < d_2$. 
We may also consider percentile queries in multi-dimensional systems. If in the model, we add information about fuel consumption, we may also enforce that we arrive within $d$ time units in 95 percent of the cases, and that in half of the cases the fuel consumption is below some threshold $c$.

\paragraph{\bf Contributions.} We study percentile problems for a range of classical payoff functions: we establish algorithms and prove complexity and memory bounds.
Our algorithms can handle multi-constraint multi-dimensional queries, but
we also study interesting subclasses, namely, multi-constraint single-dimensional queries, single-constraint queries, and other classes depending on the payoff functions.
We present an overview of our results in Table~\ref{table}.
For all payoff functions but the discounted sum, they only require \textit{polynomial time in the size of the model} when the query size is fixed. 
In most applications, the query size is typically small while the model
can be very large. So our algorithms have clear potential to be useful in practice.

\def\arraystretch{1.2}
\begin{table}[t]
  \footnotesize
  \centering
  \begin{tabular}{|c||c|c|c|}
    \cline{2-4} \multicolumn{1}{c||}{} & \multirow{2}{*}{~Single-constraint~} & Single-dim. & ~Multi-dim.~ \\
    \multicolumn{1}{c||}{} & & ~Multi-constraint~ & ~Multi-constraint~\\
    \hline
    \hline
    Reachability & \PTIME~\cite{Puterman-wiley94} & P($M$)$\cdot$E($\query$)~\cite{EKVY-lmcs08}, \PSPACE-h & --- \\
    \hline
    \multirow{2}{*}{~$\scriptsize f \in \mathcal{F}$~} & \multirow{2}{*}{\PTIME~\cite{CH-ilc09}} & \multirow{2}{*}{\PTIME} & ~P($M$)$\cdot$E($\query$)~ \\
    & & & \PSPACE-h.\\
    \hline
    ~$\mpsup$ & ~\PTIME~\cite{Puterman-wiley94}~ & \PTIME & \PTIME\\
    \hline
    ~$\mpinf$ & ~\PTIME~\cite{Puterman-wiley94}~ & ~P($M$)$\cdot$E($\query$)~ & ~P($M$)$\cdot$E($\query$)~\\
    \hline
    \multirow{2}{*}{~SP~} &  ~P($M$)$\cdot$P$_{ps}$($\query$)~\cite{HaaseK14}~ & ~P($M$)$\cdot$P$_{ps}$($\query$) (one target)~ & ~P($M$)$\cdot$E($\query$)~\\
    & ~\PSPACE-h.~\cite{HaaseK14}~ & ~\PSPACE-h.~\cite{HaaseK14}~ & ~\PSPACE-h.~\cite{HaaseK14}~\\
    \hline
    \multirow{2}{*}{~$\varepsilon$-gap DS} & ~P$_{ps}$($M, \query, \varepsilon$)~ & ~P$_{ps}$($M,\varepsilon$)$\cdot$E($\query$)~ & ~P$_{ps}$($M,\varepsilon$)$\cdot$E($\query$)~\\
    & \NPTIME-h. & \NPTIME-h. & \PSPACE-h.\\
    \hline
  \end{tabular}
  \vspace{2mm}
  \caption{Some results for percentile queries. Here $\mathcal{F} = \{\inf, \sup, \liminf, \limsup\}$, $\mpsup$ (resp. $\mpinf$) stands for sup. (resp. inf.) mean-payoff, SP for shortest path, and DS for discounted sum. Parameters $M$ and $\query$ resp. represent model size and query size; P($x$), E($x$) and P$_{ps}$($x$) resp. denote polynomial, exponential and pseudo-polynomial time in parameter $x$. All results without reference are new.}
  \vspace{-6mm}
  \label{table}
\end{table}

We give a non-exhaustive list of contributions and highlight some links with related problems. 
\begin{itemize}
\item[A)] We show the \PSPACE-hardness of the multiple reachability problem with exponential dependency on the query size (Theorem~\ref{thm:asreach}), and
the \PSPACE-completeness of the almost-sure case, refining the results of~\cite{EKVY-lmcs08}. We also prove that in the case of \emph{nested} target sets, the problem admits polynomial-time solution (Theorem~\ref{thm:monotonic-reach}), and we use it to solve some of the multi-constraint percentile problems.
\item[B)] For payoff functions $\inf$, $\sup$, $\liminf$ and $\limsup$, we establish a polynomial-time algorithm for the single-dimension case (Theorem~\ref{thm:quant_reg_single_dim}), and an algorithm that is only exponential in the size of the query for the general case (Theorem~\ref{thm:quant_reg_multi_dim}).
We prove the \PSPACE-hardness of the problem for $\sup$ (Theorem~\ref{thm:quant_reg_multi_dim_pspace}), and give a polynomial time algorithm for $\limsup$ 
(Theorem~\ref{lemma:limsup}).
\item[C)] In the mean-payoff case, we distinguish $\mpsup$ defined by the limsup of the average weights, and~$\mpinf$ by their liminf. For the former, we give a polynomial-time algorithm for the general case (Theorem~\ref{thm:mpsup}). For the latter, our algorithm is polynomial in the model size and exponential in the query size (Theorem~\ref{thm:mpinf}).
\item[D)] The truncated sum function computes the \emph{sum} of weights until a target is reached. It models \emph{shortest path} problems. We prove the multi-dimensional percentile problem to be undecidable when both negative and positive weights are allowed (Theorem~\ref{thm:truncated_undec}). Therefore, we concentrate on the case of non-negative weights, and establish an algorithm that is polynomial in the model size and exponential in the query size (Theorem~\ref{thm:sp_overview}). We derive from~\cite{HaaseK14} that even the single-constraint percentile problem is \PSPACE-hard.
\item[E)] The discounted sum case turns out to be difficult, and linked to a long-standing open problem, not known to be decidable (Lemma~\ref{lem:ds_precise}). Nevertheless, we give algorithms for an approximation of the problem, called $\varepsilon$-gap percentile problem. Our algorithm guarantees correct answers up to an arbitrarily small zone of uncertainty (Theorem~\ref{thm:ds_overview}). We also prove that this $\varepsilon$-gap problem is \PSPACE-hard in general, and already \NPTIME-hard for single-constraint queries (Lemma~\ref{lem:ds_pspace_hard} and Lemma~\ref{lem:ds_np_hard}). According to a very recent paper by Haase and Kiefer~\cite{HaasePP}, our reduction even proves \textsf{PP}-hardness of single-contraint queries, which suggests that the problem does not belong to $\NPTIME$ at all otherwise the polynomial hierarchy would collapse.
\end{itemize}

We systematically study the memory requirement of strategies. We build our
algorithms using different techniques. Here are a few of them. For $\inf$ and
$\sup$ payoff functions, we reduce percentile queries to multiple reachability
queries, and rely on the algorithm of \cite{EKVY-lmcs08}: those are the
easiest cases. For $\liminf$, $\limsup$ and $\mpsup$, we additionally need to
resort to maximal end-component decomposition of MDPs. For the following cases,
there is no simple reduction to existing problems and we need non-trivial
techniques to establish algorithms.
For $\mpinf$, we use linear programming techniques to characterize winning strategies, borrowing ideas from~\cite{EKVY-lmcs08,BBCFK-lmcs14}. For shortest path and discounted sum, we consider unfoldings of the MDP, with particular care to bound their sizes, and for the latter, to analyze the cumulative error due to necessary roundings.

\paragraph{{\bf Related work.}}
There are several works in the literature that study multi-dimensional MDPs: for discounted sum, see~\cite{CMH-stacs06}, and for mean-payoff, see~\cite{BBCFK-lmcs14,FKR-ieee95}.
In the latter papers, 
the following threshold problem is studied in multi-dimensional MDPs: given a threshold vector~$\vec{v}$ and a probability threshold $\nu$, does there exist a strategy $\sigma$ such that $\pr_s^\sigma[\vec{r} \geq \vec{v}] \geq \nu$, where $\vec{r}$ denotes the mean-payoff vector.
The work~\cite{FKR-ieee95} solves this problem for the single dimensional case, and the multi-dimensional for the \emph{non-degenerate} case (w.r.t.~the solutions of a linear program). 
A general algorithm was later given in \cite{BBCFK-lmcs14}.
This problem asks for a bound on the \emph{joint probability} of the thresholds, that is, the probability of satisfying \emph{all} constraints simultaneously. In contrast,
in our problem we bound the \emph{marginal probabilities} separately, which may allow for more modeling flexibility.
The problem of maximizing the \emph{expectation vector} was also solved in~\cite{BBCFK-lmcs14}. Recently, and independently from our work, the problem of bounding the marginal probabilities was considered in~\cite{CKK-lics15} for $\mpinf$.
The given algorithm consists in a single linear program
(while we use a two-step procedure), has the same ingredients as ours, and has the same complexity. In addition, the algorithm of~\cite{CKK-lics15} also allows one to add
a constraint on the expectation vector.

Multiple reachability objectives in MDPs were considered in~\cite{EKVY-lmcs08}: given an MDP and multiple targets~$T_i$, thresholds $\alpha_i$, decide if there exists a strategy that forces each $T_i$ with a probability larger than $\alpha_i$. This work is the closest to our work and we show here that their problem is inter-reducible with our problem for the sup measure. In~\cite{EKVY-lmcs08} the complexity results are given only for the size of the model and not for the size of the query: we refine those results here and answer questions that were left open in that paper.

Several works consider percentile queries but only for \textit{one} dimension and {\em one} constraint (while we consider multiple constraints and possibly multiple dimensions) and particular payoff functions. Single-constraint queries for $\limsup$ and $\liminf$ were studied in~\cite{CH-ilc09}. 
The threshold probability problem for truncated sum was studied in MDPs with either all non-negative or all non-positive weights in~\cite{Ohtsubo-amc2004,SO-jcta13}. \textit{Quantile queries} in the single-constraint case were studied for
the shortest path with non-negative weights 
in~\cite{DBLP:conf/fossacs/UmmelsB13}, and for
energy-utility objectives in~\cite{BDDKK-fm14}. They have been recently extended to \textit{cost problems}~\cite{HaaseK14}, in a direction orthogonal to ours. For fixed finite horizon, \cite{XM-ijcai11} considers the problem of ensuring a single-contraint percentile query for the discounted sum, and that of maximizing the expected value subject to a single percentile constraint. Still for the discounted case, there is a series of works studying \textit{threshold problems}~\cite{White1993634,WL99} and \textit{value-at-risk problems}~\cite{DBLP:conf/fsttcs/BrazdilCFNS13}. All can be related to single-constraint percentiles queries.

This paper extends previous work presented in a conference~\cite{RRS-cav15}: it gives a full presentation of the technical details, along with additional results.

%% file: sp.tex
\section{Shortest Path}
\label{sec:sp}
We study shortest path problems in MDPs, which generalize the classical graph problem. In MDPs, the problem consists in finding a strategy ensuring that a target set is reached with bounded truncated sum with high probability.
This problem has been studied in the context of games and MDPs (e.g.,~\cite{bertsekas_MOR1991,DBLP:conf/concur/Alfaro99,DBLP:conf/stacs/BruyereFRR14}).
We consider percentile queries of the form $\query \coloneqq \bigwedge_{i = 1}^{q}\; \pr_{M,\initState}^\strat\big[\truncatedSum{\truncatedTarget_{i}}_{l_{i}} \leq v_i\big] \geq
\alpha_i$ (inner inequality $\leq$ is more natural but $\geq$ could be used by negating all weights). Observe that each constraint $i$ may relate to a different target set $\truncatedTarget_{i} \subseteq \states$.

\subsection{MDPs with Arbitrary Weights}

We prove that without further restriction, the multi-dimen\-sion\-al percentile problem is undecidable, even for a fixed number of dimensions. Our proof is inspired by the approach of Chatterjee et al.~for the undecidability of two-player multi-dimensional total-payoff games~\cite{DBLP:journals/iandc/Chatterjee0RR15} but requires additional techniques to adapt to the stochastic case. 

\begin{theorem}
\label{thm:truncated_undec}
The multi-dimensional percentile problem is undecidable for the truncated sum payoff function, for MDPs with both negative and positive weights and four dimensions, even with a unique target set.
\end{theorem}

\begin{proof}
We reduce the halting problem for two-counter machines (2CMs) to a multi-dimensional percentile problem for the truncated sum payoff function over an MDP with weights in $\integ^{4}$, with a unique target set.

Counters of a 2CM take values $(v_{1}, v_{2}) \in \nat^{2}$ along an execution, and can be incremented or decremented (if positive). A counter can be tested for equality to zero, and the machine can branch accordingly.
The halting problem for 2CMs is well-known to be undecidable~\cite{minsky1961}.

Consider a 2CM $\twoCM$. From this 2CM, we construct an MDP $\markovProcess = (\states, A, \delta, \weight)$ and a target set of states $\truncatedTarget \subset \states$, with an initial state $\initState \in \states$ such that there exists a strategy $\strat \in \strats$ satisfying the four-dimensional percentile query
\begin{equation*}
\query \coloneqq \bigwedge_{i = 1}^{4}\; \pr_{M,\initState}^\strat\big[\truncatedSum{\truncatedTarget}_{l_{i}} \leq 0\big] =
	1.
\end{equation*}
if and only if the machine does not halt.

Intuitively, this MDP is built such that strategies that do not faithfully simulate the 2CM~$\twoCM$ cannot satisfy the percentile query. To ensure that this is the case, we will implement checks through probabilistic transitions that will produce bad runs with positive probability against unfaithful strategies.

The MDP $\markovProcess$ is built as follows. The states of $\markovProcess$ are copies of the control states of $\twoCM$ (plus some special states discussed in the following). Actions in the MDP represent transitions between these control states. The weight function maps actions to $4$-dimensional vectors of the form $(c_{1}, -c_{1}, c_{2}, -c_{2})$, that is, two dimensions for the first counter $C_{1}$ and two for the second counter $C_{2}$. Each increment of counter $C_{1}$ (resp. $C_{2}$) in $\twoCM$ is implemented in $\markovProcess$ as an action of weight $(1, -1, 0, 0)$ (resp. $(0, 0, 1, -1)$). For decrements, we have weights respectively $(-1, 1, 0, 0)$ and $(0, 0, -1, 1)$ for $C_{1}$ and $C_{2}$. Therefore, the current value of counters $(v_{1}, v_{2})$ along an execution of the 2CM $\twoCM$ is represented in the MDP as the current sum of weights, $(v_{1}, -v_{1}, v_{2}, -v_{2})$. Hence, along a faithful execution, the 1st and 3rd dimensions are always non-negative, while the 2nd and 4th are always non-positive. The two dimensions per counter are used to enforce faithful simulation of non-negativeness of counters and zero test.

\begin{figure}[tb]
        \centering
        \begin{subfigure}[b]{0.26\textwidth}
        \centering
               \scalebox{1.1}{\begin{tikzpicture}[->,>=stealth',shorten >=1pt,auto,node
    distance=2.5cm,bend angle=45, scale=0.6, inner sep=0pt,font=\scriptsize]
    \tikzstyle{p1}=[draw,circle,text centered,minimum size=5mm,text width=4mm]
    \tikzstyle{p2}=[fill,circle,text centered,minimum size=1.5mm]
    \node[p1]  (1) at (0, 0)  {};
    \node[p2]  (3)  at (0, -2.5) {};
    \node[p1,dashed]  (2)  at (1.6, -4) {};
    \node[draw,rectangle,dashed,inner sep=2pt] (4) at (-1.6, -4) {escape gadget};
    \path
    (3) edge[shorten <=1pt](2)
    (3) edge[shorten <=1pt] (4)
    (1) edge node[left] {$(1, -1, 0, 0)$} (3);
      \end{tikzpicture}}
                \caption{Increment $C_{1}$.}
        \end{subfigure}%
        $\quad$ 
        \begin{subfigure}[b]{0.7\textwidth}
        \centering
               \scalebox{1.1}{\begin{tikzpicture}[->,>=stealth',shorten >=1pt,auto,node
    distance=2.5cm,bend angle=45, scale=0.6, inner sep=0pt,font=\scriptsize]
    \tikzstyle{p1}=[draw,circle,text centered,minimum size=5mm,text width=4mm]
    \tikzstyle{p2}=[fill,circle,text centered,minimum size=1.5mm]
    \node[p1]  (1) at (0, 0)  {};
    \node[p2]  (2)  at (4, 0) {};
    \node[p1,dashed]  (3) at (8, 2)  {};
    \node[p1]  (4) at (8, -2)  {};
    \node[p1,double]  (5) at (12, -2)  {};
    \node[draw,rectangle,dashed,inner sep=2pt] (6) at (3, -2) {escape gadget};
    \path
    (1) edge node[above] {$(-1, 1, 0, 0)$} (2)
    (2) edge[shorten <=1pt](3)
    (2) edge[shorten <=1pt] (4)
    (2) edge (6)
    (4) edge (5)
    (4) edge [loop above, out=130, in=50,looseness=3, distance=2cm] node [above] {$(-1, 0, -1, -1)$} (4)
    (5) edge [loop above, out=130, in=50,looseness=3, distance=2cm] node [above] {$(0, 0, 0, 0)$} (5);
      \end{tikzpicture}}
                \caption{Decrement $C_{1}$.}\label{fig:sp_undec_gadgets_dec}
        \end{subfigure}
        
\vspace{4mm}       

\begin{subfigure}[b]{0.26\textwidth}
        \centering
               \scalebox{1.1}{\begin{tikzpicture}[->,>=stealth',shorten >=1pt,auto,node
    distance=2.5cm,bend angle=45, scale=0.6, inner sep=0pt,font=\scriptsize]
    \tikzstyle{p1}=[draw,circle,text centered,minimum size=5mm,text width=4mm]
    \tikzstyle{p2}=[draw,rectangle,text centered,minimum size=5mm,text width=4mm]
    \node[p1]  (1)  at (0, 0) {};
    \path
    (0, 2) edge (1)
    (1) edge [loop below, out=230, in=310,looseness=3, distance=2cm] node [below] {$(0, 0, 0, 0)$} (1);
      \end{tikzpicture}}
                \caption{Halting.}
        \end{subfigure}%
        $\quad$ 
        \begin{subfigure}[b]{0.7\textwidth}
        \centering
               \scalebox{1.1}{\begin{tikzpicture}[->,>=stealth',shorten >=1pt,auto,node
    distance=2.5cm,bend angle=45, scale=0.6, inner sep=0pt,font=\scriptsize]
    \tikzstyle{p1}=[draw,circle,text centered,minimum size=5mm,text width=4mm]
    \tikzstyle{p2}=[fill,circle,text centered,minimum size=1.5mm]
    \node[p1]  (1) at (0, 0)  {};
    \node[p2]  (2)  at (4, 0) {};
    \node[p1,dashed]  (3) at (8, 2)  {};
    \node[p1]  (4) at (8, -2)  {};
    \node[p1,double]  (5) at (12, -2)  {};
    \node[draw,rectangle,dashed,inner sep=2pt] (6) at (3, -2) {escape gadget};
    \path
    (1) edge (2)
    (2) edge[shorten <=1pt](3)
    (2) edge[shorten <=1pt] (4)
    (4) edge (5)
    (2) edge (6)
    (4) edge [loop above, out=130, in=50,looseness=3, distance=2cm] node [above] {$(0, -1, -1, -1)$} (4)
    (5) edge [loop above, out=130, in=50,looseness=3, distance=2cm] node [above] {$(0, 0, 0, 0)$} (5);
      \end{tikzpicture}}
                \caption{Checking that $C_{1}$ is equal to zero.}\label{fig:sp_undec_gadgets_zero}
        \end{subfigure}
\vspace{4mm}       

\begin{subfigure}[b]{1\textwidth}
        \centering
             \scalebox{1.1}{\begin{tikzpicture}[->,>=stealth',shorten >=1pt,auto,node
    distance=2.5cm,bend angle=45, scale=0.6, inner sep=0pt,font=\scriptsize]
    \tikzstyle{p1}=[draw,circle,text centered,minimum size=5mm,text width=4mm]
    \tikzstyle{p2}=[draw,rectangle,text centered,minimum size=5mm,text width=4mm]
    \node[p1]  (1) at (0, 0)  {};
    \node[p1,double]  (2)  at (4, 0) {};
    \path
    (-2, 0) edge (1)
    (1) edge node[above] {$(0, 0, 0, 0)$} (2)
    (1) edge [loop above, out=130, in=50,looseness=3, distance=2cm] node [above] {$(-1, -1, -1, -1)$} (1)
    (2) edge [loop above, out=130, in=50,looseness=3, distance=2cm] node [above] {$(0, 0, 0, 0)$} (2);
      \end{tikzpicture}}
                \caption{Escape gadget reachable by every action of the MDP.}
        \end{subfigure}
        \caption{Gadgets encoding 2CM halting problem in a multi-dimensional percentile problem for truncated sum payoff function.}\label{fig:sp_undec_gadgets}    
\end{figure}
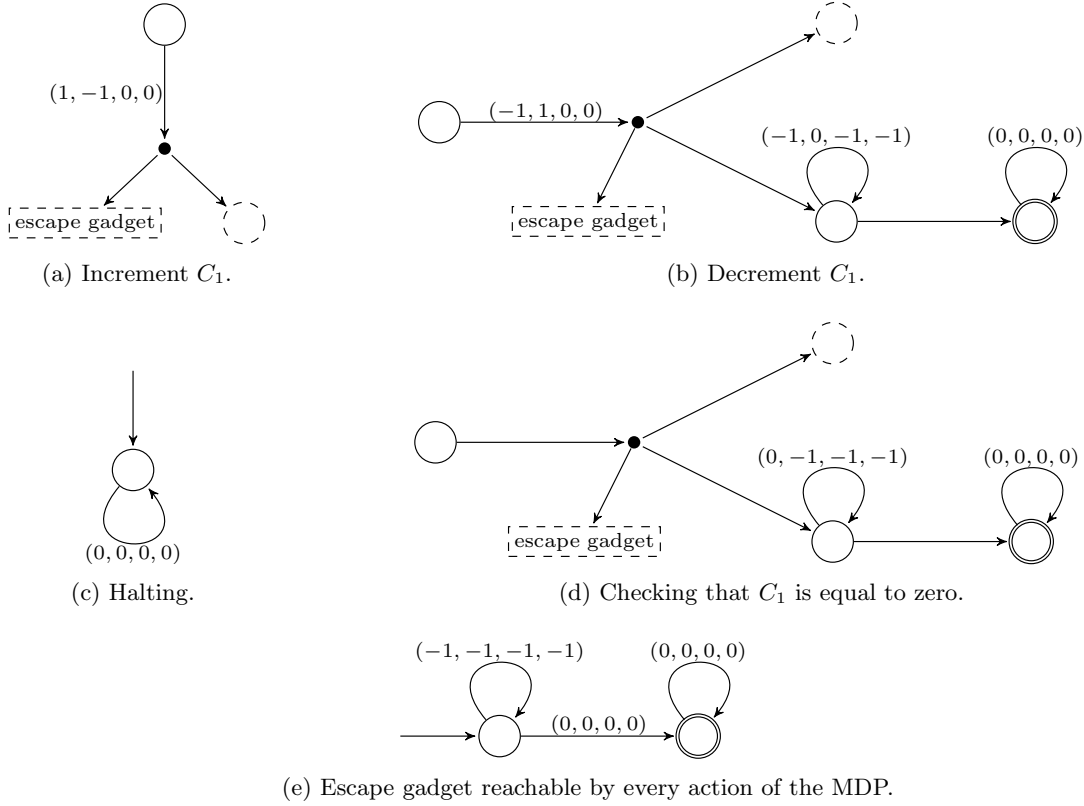

We now discuss how this MDP $\markovProcess$ ensures faithful simulation of the 2CM~$\twoCM$ by the controller through the use of the gadgets represented in Fig.~\ref{fig:sp_undec_gadgets}. Small filled circles represent equiprobable stochastic transitions, double circles depict states of the target set $\truncatedTarget$.
\begin{itemize}
\item An \textit{escape gadget} has a positive probability to be reached whenever an instruction of the 2CM is simulated. When in this gadget, the controller can decrease the sum on all dimensions below zero by cycling long enough before deciding to reach the target, hence making the run acceptable for the considered percentile query.

\item \textit{Increment and decrement} of counter values are easily simulated using the first four dimensions.

\item \textit{Values of counters may never go below zero}. To ensure this, we make every decrement action probabilistic with three equiprobable outcomes. Consider the decrement of $C_{1}$ as depicted in Fig.~\ref{fig:sp_undec_gadgets_dec}. Either the simulation continues (dashed control state), or it branches to the escape gadget, or it branches to the bottom right part of the decrement gadget. In that case, the controller can cycle long enough in the first state to ensure a negative sum of weights in all dimensions except for the second one, before reaching the target (runs \textit{have to} reach the target or their truncated sum will be $\infty$). If the controller is not faithful and a negative value is reached on counter $C_{1}$ when decrementing, this branching will induce a run which is losing because the second dimension will be strictly positive (recall it has value $-c_{1}$). Notice that the controller can never cheat, otherwise this branching will happen with strictly positive probability (i.e., after a finite prefix). On the contrary, if the controller never cheats, this branching is harmless and induces acceptable runs w.r.t. the percentile query. The gadget is similar for decrements of $C_{2}$ using the fourth dimension.

\item \textit{Zero tests are correctly executed}. In the same spirit, we allow a probabilistic branching after the controller claims a counter is equal to zero. Consider Fig.~\ref{fig:sp_undec_gadgets_zero} for counter $C_{1}$. If a zero test is passed while $c_{1} > 0$, the sum on the first dimension will stay strictly positive and the run will not be acceptable. On the contrary, if the controller is faithful, this branching is again harmless as it is possible to make all sums negative except for the first dimension for which it would already be equal to zero. We use a similar gadget for $C_{2}$ based on the third dimension. We also need to ensure that the controller cannot cheat by claiming that a counter is strictly positive while it is not. To achieve this, we follow each claim that $C_i$ is strictly positive by a decrement on $C_i$ (using our gadget) followed by an increment (idem). Thus, if $C_i$ is equal to zero while the controller claims the opposite, the decrement gadget will yield non-acceptable runs, as seen before. On the other hand, if the controller is faithful, visiting those two gadgets is safe (and does not modify the counters for the rest of the simulation).

\item \textit{Halting}. The end of a 2CM execution is modeled in the MDP by an halting state. This state does not belong to the target set $\truncatedTarget$: any run corresponding to an halting execution will have its truncated sum equal to $\infty$ on all dimensions, which makes it bad for the percentile query.

\end{itemize}

Now, we have argued that if the simulation is not faithful, bad runs will be produced with strictly positive probability, and the percentile query will not be satisfied. Furthermore, if the machine halts, then with a strictly positive probability, the halting state will be reached (because the machine halts after a \textit{finite} number of operations), which also results in bad runs. Hence if the percentile query is satisfied by a strategy $\strat \in \strats$, then this strategy describes a faithful infinite execution of $\twoCM$.

It remains to show that if the 2CM does not halt, the percentile query is satisfiable. Clearly, the halting state will never be reached, and gadgets cannot produce bad runs as the simulation is faithful. However, runs that never reach any target state (i.e., runs that never branch away from the simulation) are still bad runs as they yield an infinite truncated sum. Nonetheless, observe that each action taken in the MDP yields a strictly positive probability to branch to the escape gadget or to branch inside decrement and zero-test gadgets. Hence, if the 2CM does not halt, such actions are taken infinitely often and with probability one, the simulation eventually branches toward the target states (with a good truncated sum as argued before). We conclude that the strategy that simulates a never-halting machine does yield good runs with probability one.

Consequently, we have that the studied multi-dimensional percentile problem is equivalent to the 2CM halting problem, and thus, undecidable.
\end{proof}

\subsection{MDPs with Non-Negative Weights}

In the light of this result, we will restrict our setting to non-negative weights, a setting closer to the original interpretation of the shortest path problem (we could equivalently consider non-positive weights with inequality $\geq$ inside percentile constraints). We first discuss recent related work.

\smallskip\noindent\textbf{Comparison with quantiles and cost problems.} In~\cite{DBLP:conf/fossacs/UmmelsB13}, Ummels and Baier study \textit{quantile queries} over non-negatively weighted MDPs. They are equivalent to minimizing $v \in \nat$ in a single-constraint percentile query $\pr_{M,\initState}^\strat\big[\truncatedSum{\truncatedTarget} \leq v\big] \geq
	\alpha$ such that there still exists a satisfying strategy, for some fixed $\alpha$. Very recently, Haase and Kiefer extended quantile queries by introducing \textit{cost problems}~\cite{HaaseK14}. They can be seen as single-constraint percentile queries where inequality $\truncatedSum{\truncatedTarget} \leq v$ is replaced by an arbitrary Boolean combination of inequalities~$\varphi$. Hence, it can be written as $\pr_{M,\initState}^\strat\big[\truncatedSum{\truncatedTarget}  \models \varphi\big] \geq
	\alpha$.
Cost problems are studied on single-dimensional MDPs and all the inequalities relate to the same target $\truncatedTarget$, in contrast to our setting which allows both for multiple dimensions and multiple target sets. The single probability threshold bounds the probability of the whole event~$\varphi$.
	
Both settings are incomparable. 
Cost problems consist in a unique query that checks that with probability $\alpha$, paths satisfy the Boolean combination $\varphi$. In percentile queries, we have several constraints and we check that \textit{each} inequality is satisfied with the corresponding probability $\alpha_{i}$: paths do not need to satisfy \textit{all} inequalities at the same time. In full generality, for a fixed probability threshold $\alpha$, it is easier to satisfy a unique constraint over a disjunction of inequalities than to satisfy a disjunction of constraints over single inequalities: in the second case, $\alpha$ percent of the paths must satisfy the \textit{same} unique inequality, not in the first one. Similarly, it is harder to satisfy a unique constraint over a conjunction of inequalities than to satisfy a conjunction of constraints over single inequalities: in the first case, $\alpha$ percent of the paths must satisfy \textit{all} inequalities, not in the second one.

Still, our queries share common subclasses with cost problems: atomic formulae~$\varphi$ exactly correspond to our single-constraint queries. Moreover, cost problems for such formulae are inter-reducible with quantile queries~\cite[Proposition 2]{HaaseK14}. Cost problems with atomic formulae are \PSPACE-hard, so this also holds for \textit{single-constraint} percentile queries. The best known algorithm in this case is in \EXPTIME. In the following, we establish an algorithm that still only requires exponential time while allowing for \textit{multi-constraint multi-dimensional multi-target} percentile queries.

\smallskip\noindent\textbf{Overview.} 
Our main contributions for the shortest path are summarized in Theorem~\ref{thm:sp_overview}. In the following, we detail each of them and discuss some subclasses of queries with interesting complexities.

\begin{theorem}
\label{thm:sp_overview}
The percentile problem for the shortest path with non-negative weights can be solved in time polynomial in the model size
and exponential in the query size (exponential in the number of constraints and pseudo-polynomial in the largest threshold). 
The problem is \PSPACE-hard even for single-constraint queries. Exponential-memory strategies are sufficient and in general necessary.
\end{theorem}

\smallskip\noindent\textbf{Algorithm.} The sketch of our algorithm is as follows. Consider a $d$-dimensional MDP $M$ and a $q$-query percentile problem, with potentially different targets for each query. 
Let $v_{{\sf max}}$ be the maximum of the thresholds $v_i$. Because weights are non-negative, extending a finite history never decreases the sum of its weights.
Thus, any history ending with a sum exceeding $v_{\sf max}$ in all dimensions is surely losing under any strategy.

Based on this, we build an MDP $M'$ by unfolding $M$ and integrating the sum for each dimension in states of $M'$. We ensure its finiteness thanks to the above observation and we reduce its overall size to a \textit{single}-exponential by defining a suitable equivalence relation between states of $M'$: we only care about the current sum in each dimension, and we can forget about the actual path that led to it. Precisely, 
the states of~$M'$ are in $S \times \{0,\ldots,v_{\sf max}+1\}^d$. 
Now, for each constraint, we compute a set of target states in $M'$ that exactly captures all runs satisfying the inequality of the constraint. Thus, we are left with a multiple reachability problem on~$M'$: we look for a strategy~$\strat'$ that ensures that each of these sets $R_{i}$ is reached with probability $\alpha_{i}$. This query can be answered in time polynomial in $\vert M'\vert$ but exponential in the number of sets $R_{i}$, i.e., in $q$ (Theorem~\ref{thm:absorbing-reachsafe}). 

We prove the correctness of our algorithm in the next lemma.

\begin{lemma}
\label{lem:sp_alg}
The percentile problem for the shortest path can be solved in time polynomial in the size of the MDP and the thresholds values, and exponential in the number of dimensions of the weight function and the number of constraints of the problem.
\end{lemma}

\begin{proof}
Let $\markovProcess = (\states, A, \delta, \weight)$ be the considered MDP,  with $w\colon A \rightarrow \nat^{d}$ its $\dimension$-dimensional non-negative weight function, $\initState \in \states$ the initial state. We consider a $\queries$-constraint query: we are looking for a strategy $\strat \in \strats$ such that
$\query \coloneqq \bigwedge_{i = 1}^{q}\; \pr_{M,\initState}^\strat\big[\truncatedSum{\truncatedTarget_{i}}_{l_{i}} \leq v_i\big] \geq
	\alpha_i$
for the given thresholds $v_i \in \nat$, $\alpha_i \in [0, 1] \cap \rat$ and target sets $\truncatedTarget_{i} \subseteq \states$. The algorithm is as follows.

Let $v_{{\sf max}}$ be the maximum of the thresholds $v_i$, $i \in \{1, \ldots{}, \queries\}$. Observe that given any prefix of a run, extending it can never decrease the sum of weights (as all weights are non-negative) and that any run for which the truncated sum exceeds~$v_{{\sf max}}$ in all dimensions is not interesting for the controller.

Based on those observations, we unfold the MDP~$\markovProcess$, creating a tree-like structure in the nodes of which we integrate the current sum of weights. That is, nodes are labeled by elements of $\states \times \nat^{\dimension}$. We stop a branch as soon as the sum reaches $v_{{\sf max}}+1$ in all dimensions (we do not care about what happens after the sum hits this value as it is now a bad outcome for all the percentile constraints). Now, this unfolding is not exactly a tree because we allow actions of weight zero in the original MDPs. Hence we may have to introduce cycles in the unfolding: whenever a branch visits a node with a label identical to one of its ancestors, we stop this branch and introduce a cycle to the corresponding ancestor. Those two cutting criteria guarantee an unfolding which is finite and of maximum height $h = \mathcal{O}(\vert\states\vert \cdot (v_{{\sf max}}+2) \cdot \dimension)$. That is because every cycle (in the original MDP) that does not result in a cycle in the unfolding has to increase at least one dimension, by at least one, and has at most length $\vert \states \vert$.

Now consider the overall size of this unfolding. Recall that we want to build an unfolding which is at most exponential. If no special care is taken, the size of the unfolding could be as high as $\mathcal{O}(b^{h})$, where $b$ denotes the branching degree of $\markovProcess$, defined as
\begin{equation*}
b = \max_{s \in \states} \big\vert \{(a, s') \mid a \in A(s), s' \in \states, \delta(s,a,s') > 0\} \big\vert.
\end{equation*}
In particular, the overall size could be exponential in $v_{{\sf max}}$, that is, doubly-exponential in its encoding. To avoid that, we reduce the size of the unfolding by merging equivalent nodes.

What are equivalent nodes? First, we declare two nodes to be equivalent if they relate to the same state and describe identical sums on all dimensions. Second, observe that for any node of the unfolding, the sum on any dimension can theoretically grow up to $h\cdot W$, with $W$ the largest weight appearing on any action of $\markovProcess$. That is, it can grow larger than $(v_{{\sf max}}+1)$ as we stop only when \textit{all} dimensions are larger than this bound. Nonetheless, w.r.t.~satisfaction of the percentile query, we do not need to recall exactly what is the value reached after exceeding $v_{{\sf max}}$ as in any case, such a sum in a given dimension is not acceptable for any related constraint. Hence, we can also merge nodes by replacing any label larger than $(v_{{\sf max}}+1)$ by label $(v_{{\sf max}}+1)$.

By merging nodes equivalent according to this definition, we ensure that the overall size of the unfolding is at most $u = \mathcal{O}(\vert \states \vert \cdot (v_{{\sf max}}+2)^{\dimension})$. Indeed, the possible values for sums on any dimension in the unfolding run from $0$ to $(v_{{\sf max}}+1)$. Observe that this overall size $u$ is, as desired, polynomial in the number of states~$\vert \states \vert$ and in the largest threshold $v_{{\sf max}}$, and exponential in the number of dimensions $d$.

Interestingly, this merging process can be executed on the fly while building the unfolding hence does not hinder the total execution time of the algorithm (i.e., one does not have to fully build the doubly-exponential unfolding to construct the single-exponential merged one).

Now notice that this unfolding is itself an MDP, denoted~$\markovProcess'$. For each constraint $i \in \{1, \ldots{}, q\}$, we can compute the set $R_{i}$ of nodes that are labeled by a state in the corresponding target set $T_{i}$ and have a label less than or equal to $v_{i}$ on the corresponding sum dimension $l_{i}$. Hence, such a set $R_{i}$ actually captures all branches satisfying the inequality of constraint~$i$ (a branch is captured if it possesses a node of $R_{i}$). Observe that we only consider dimensions related to constraint $i$ when computing the set~$R_i$ (e.g., it is not a problem to exceed $v_{{\sf max}}$ in other dimensions). This computation takes time $\mathcal{O}(u \cdot \queries)$ in the worst case.

Now we are left with a multiple reachability problem on $M'$: we have to decide the existence of a strategy~$\strat'$ satisfying the query
\begin{equation*}
\query' \coloneqq \bigwedge_{i = 1}^{q}\; \pr_{M',\initState'}^{\strat'}\big[\diamondsuit R_{i}\big] \geq
	\alpha_i.
\end{equation*}
If such a strategy $\strat'$ exists in $\markovProcess'$, it is easy to see that the equivalent strategy~$\strat$ in the original MDP $\markovProcess$ satisfies the shortest path percentile query~$\query$. Indeed, the probability of reaching set $R_{i}$ in $M'$ following strategy $\strat'$ is exactly the probability of satisfying constraint $i$ in $M$ following the equivalent strategy $\strat$. On the contrary, if no strategy satisfies the multiple reachability query~$\query'$, it implies that the original percentile query~$\query$ cannot be satisfied either.

Solving this multiple reachability query~$\query'$ can be done in polynomial time in the size of the unfolding MDP $\markovProcess'$ but exponential in the number of sets $R_{i}$, i.e., in the number of constraints (Theorem~\ref{thm:absorbing-reachsafe}). Hence, the overall time required by the algorithm is polynomial in $\vert\states\vert$ and in the maximum threshold $v_{{\sf max}}$, and exponential in the number of dimensions~$\dimension$ and in the number of constraints $\queries$.
\end{proof}

It is worthwhile to mention that the exponential dependency on the number of constraints can be lifted when they all consider the same target set.

\begin{remark}
\textit{Percentile problems with unique target are solvable in time polynomial in the number of constraints but still exponential in the number of dimensions.}
\end{remark}

\begin{proof}
Consider the unfolding algorithm described in the proof of Lemma~\ref{lem:sp_alg}. Assume that all constraints of the percentile query relate to the same target set $\truncatedTarget \subseteq \states$. In that case, all branches can be stopped as soon as they reach a state of $\truncatedTarget$. Thus, when computing the sets $R_{i}$ for the multiple reachability problem on the unfolding $\markovProcess'$, all nodes that belong to these sets are actually leaves of the unfolding. Hence they are absorbing states of $\markovProcess'$. From Theorem~\ref{thm:absorbing-reachsafe}, it follows that the multiple reachability problem can be solved in polynomial time, which eliminates the exponential dependency on the number of constraints for solving the shortest path percentile problem with a single target set.
\end{proof}

For single-dimensional queries with a unique target set (but still potentially multi-constraint), our algorithm remains pseudo-polynomial as it requires polynomial time in the thresholds values (i.e., exponential in their encoding).
\begin{corollary}
\label{cor:sp_alg}
The single-dimensional percentile problem with a unique target set can be solved in pseudo-polynomial time.
\end{corollary}

\noindent\textbf{Lower bound.}
By equivalence with cost problems for atomic cost formulae, it follows
from~\cite[Theorem 7]{HaaseK14} that no truly-polynomial-time algorithm
exists for the single-constraint percentile problem unless \PTIME~=~\PSPACE.

\begin{lemma}
The single-constraint percentile problem for the shortest path is \PSPACE-hard.
\end{lemma}

\smallskip\noindent\textbf{Memory.} 
We now formally prove the need for (and sufficiency of) exponential memory in shortest path percentile queries.

\begin{lemma}
\label{lem:spmemory}
Exponential-memory strategies are both sufficient and, in general, necessary to satisfy percentile queries for the shortest path.
\end{lemma}

\begin{proof}
First, the algorithm given in Lemma~\ref{lem:sp_alg} solves the percentile problem by answering a multiple reachability problem over an unfolded MDP of exponential size. As stated in Theorem~\ref{thm:absorbing-reachsafe}, memory of size polynomial in the MDP (here, the unfolded one) and exponential in the number of constraints (which is untouched by our algorithm) is sufficient to satisfy such queries. Hence, it follows that exponential-memory strategies suffice for shortest path percentile queries.

Second, let us show that multiple reachability problems over an MDP $M$ can be reduced to shortest path percentile problems over the very same MDP, enriched with a trivial weight function. Consider an unweighted MDP $\markovProcess = (\states, A, \delta)$ and a multiple reachability query for sets $T_{i} \subseteq \states$ and thresholds $\alpha_{i} \in \left[ 0, 1\right] \cap \rat$, with $i \in \{1, \ldots{}, \queries\}$. Let $\markovProcess' = (\states, A, \delta, \weight)$ be a single-dimensional weighted version of the MDP~$\markovProcess$, where all actions are assigned weight zero. Then we trivially have that a strategy~$\strat$ satisfies the multiple reachability query on $\markovProcess$ if and only if it satisfies the percentile query $\bigwedge_{i = 1}^{q}\; \pr_{M',\initState}^\strat\big[\truncatedSum{\truncatedTarget_{i}} \leq 0\big] \geq	\alpha_i$ on $\markovProcess'$. Indeed, runs that are bad for this percentile query are exactly the ones that are assigned truncated sum $\infty$ because they do not reach the considered target sets. This concludes the reduction.

Finally, we know by Lemma~\ref{lem:multiReach_expMemoryLB} that exponential memory is needed in general for multiple reachability queries. This lower bound thus straightforwardly carries over to shortest path percentile queries.
\end{proof}

%% file: ds.tex
\section{Discounted Sum}
\label{sec:ds}

The \textit{discounted sum} accumulates weights using a discount factor to model that short-term rewards or costs are more important than long-term ones. It is well-studied in automata~\cite{DBLP:journals/corr/BokerH14} and MDPs~\cite{Puterman-wiley94,CMH-stacs06,DBLP:conf/lpar/ChatterjeeFW13}. We consider queries of the form $\query \coloneqq \bigwedge_{i = 1}^{q}\; \pr_{M,\initState}^\strat\big[\discSum{\discount_{i}}_{l_{i}} \geq v_i\big] \geq
	\alpha_i$, for discount factors $\discount_{i} \in \left] 0, 1\right[ \cap \rat$ and the usual thresholds. That is, we study multi-dimensional MDPs and possibly distinct discount factors for each constraint.
	
Unfortunately, our setting encompasses a much simpler question which is still not known to be decidable. We discuss this question and its reduction to our problem in Sect.~\ref{subsec:precise}. We also argue that solving this problem would require an important breakthrough. Then, in Sect.~\ref{subsec:approx}, we establish a conservative algorithm that, in some sense, can approximate the answer to the percentile problem.
	
\subsection{Precise Discounted Sum is Hard}
\label{subsec:precise}

Consider the \textit{precise discounted sum problem}: given a rational $t$, and a rational discount factor $\discount \in \left] 0, 1\right[ $, does there exist an infinite binary sequence $\tau = \tau_{1}\tau_{2}\tau_{3}\ldots{} \in \{0, 1\}^{\omega}$ such that $\sum_{j = 1}^{\infty} \discount^{j} \cdot \tau_{j} = t$? In~\cite{bokerTDS}, this problem is related to several long-standing open questions, such as decidability of the \textit{universality problem for discounted-sum automata}~\cite{DBLP:journals/corr/BokerH14}. A slight generalization to paths in graphs is also mentioned by Chatterjee et al.~as a key open problem in~\cite{DBLP:conf/lpar/ChatterjeeFW13}.

\begin{lemma}
\label{lem:ds_precise}
The precise discounted sum problem can be reduced to an almost-sure percentile problem over a two-dimensional MDP with only one state.
\end{lemma}

\begin{proof}
Assume we have a precise discounted sum problem with discount factor $\discount \in \left] 0, 1\right[ \cap \rat$ and target value $t \in \rat$. Let $M$ be an MDP with only one state $s$ and two actions, $a$ and $b$ (that both cycle on $s$ with probability one, obviously). Consider the two-dimensional weight function $w\colon A \rightarrow \integ^{2}$ such that $w(a) = (0, 0)$ and $w(b) = (1, -1)$. The role of action $a$ (resp. $b$) is to represent the choice of $0$ (resp. $1$) in the binary sequence.

We define the percentile problem asking for the existence of a strategy $\strat \in \strats$ such that
\begin{equation*}
\pr_{M,s}^\strat \big[\discSum{\discount}_{1} \geq t\big] = 1\quad \wedge \quad \pr_{M,s}^\strat \big[\discSum{\discount}_{2} \geq -t\big] = 1.
\end{equation*}
By definition of the weight function, the second term of the conjunction is equivalent to $\pr_{M,s}^\strat \big[\discSum{\discount}_{1} \leq t\big] = 1$. Hence, if a satisfying strategy $\strat$ exists, it does satisfy $\pr_{M,s}^\strat \big[\discSum{\discount}_{1} = t\big] = 1$. We claim that the answer to the precise discounted sum problem is \textsf{Yes} if and only if the answer to the percentile problem is \textsf{Yes}.

First, assume a satisfying strategy $\strat$ exists. In general, our percentile problems do not require strategies to be pure. However, even if $\strat$ is randomized, we can extract a run $\rho$ induced by this strategy and such that $\discSum{\discount}_{1}(\rho) = t$ (such a run exists otherwise the strategy would not satisfy the percentile query). This run can be seen as a sequence of actions $\rho_{A} \in \{a, b\}^{\omega}$, which we translate in a corresponding sequence $\tau \in \{0, 1\}^{\omega}$ satisfying the precise discounted sum problem.

Conversely, assume there exists a sequence $\tau \in \{0, 1\}^{\omega}$ satisfying the precise discounted sum problem. Then this sequence defines a (possibly infinite-memory) pure strategy $\strat$ that ensures a discounted sum equal to $t$, and thus the percentile query is satisfied.
\end{proof}

This suggests that answering percentile problems for the discounted sum would require an important breakthrough.

\subsection{Approximation Algorithm}
\label{subsec:approx}

\smallskip\noindent\textbf{Approaching an answer.} As shown in Sect.~\ref{subsec:precise}, an exact algorithm is currently out of reach. Fortunately, we are still able to establish an algorithm that can ``approximate'' a solution. Since we consider decision problems, the notion of approximation should not be understood \textit{sensu stricto}. We will formalize the output of the algorithm in the following but we first give an intuitive sketch.

\smallskip\noindent\textbf{The $\varepsilon$-gap problem.} Our algorithm takes as input a percentile query and an arbitrarily small \textit{precision factor} $\varepsilon > 0$ and has three possible outputs: \textsf{Yes}, \textsf{No} and \textsf{Unknown}. If it answers \textsf{Yes}, then a satisfying strategy exists and can be synthesized. If it answers \textsf{No}, then no such strategy exists. Finally, the algorithm may output \textsf{Unknown} for a specified ``zone'' close to the threshold values involved in the problem and of width which depends on $\varepsilon$.
It is possible to incrementally reduce the uncertainty zone, but it cannot be eliminated as the case~$\varepsilon=0$ would answer the precise discounted sum problem, which is not known to be decidable.

We actually solve an \emph{$\varepsilon$-gap problem}, a particular case of \emph{promise problems}~\cite{DBLP:conf/birthday/Goldreich06a},
where the set of inputs is partitioned in three subsets: yes-inputs, no-inputs and the rest of them. The promise problem then asks to answer \textsf{Yes} for all yes-inputs and \textsf{No} for all no-inputs, while the answer may be arbitrary for the remaining inputs. 
In our setting, the set of inputs for which no guarantee is given can be taken arbitrarily small, parametrized by value $\varepsilon > 0$: this is an $\varepsilon$-gap problem. This notion is later formalized in Theorem~\ref{thm:ds_gap}.

\smallskip\noindent\textbf{Related work: single-constraint case.} There are papers considering models related to \textit{single-constraint} percentile queries. Consider a single-dimensional MDP and a single-constraint query, with thresholds $v$ and $\alpha$. The \textit{threshold problem} fixes $v$ and maximizes $\alpha$~\cite{White1993634,WL99}. The \textit{value-at-risk problem} fixes~$\alpha$ and maximizes $v$~\cite{DBLP:conf/fsttcs/BrazdilCFNS13}. This is similar to \textit{quantiles} in the shortest path setting~\cite{DBLP:conf/fossacs/UmmelsB13}.

Paper~\cite{DBLP:conf/fsttcs/BrazdilCFNS13} is the first to provide an exponential-time algorithm to approximate the optimal value $v^{\ast}$ under a fixed $\alpha$ in the general setting. The authors also rely on approximation. While we do not consider optimization, we do extend the setting to \textit{multi-constraint}, \textit{multi-dimensional}, \textit{multi-discount} problems, and we are able to remain in the same complexity class, namely \EXPTIME.

\smallskip\noindent\textbf{Overview.} Our main contributions for the discounted sum are summarized in Theorem~\ref{thm:ds_overview}. In the following, we provide a thorough discussion for each of them, and prove several intermediate results of interest.
 
\begin{theorem}
\label{thm:ds_overview}
The $\varepsilon$-gap percentile problem for the discounted sum can be solved in time pseudo-polynomial in the model size and the precision factor, and exponential in the query size: polynomial in the number of states, the weights, the discount factors and the precision factor, and exponential in the number of constraints. It is \PSPACE-hard for two-dimensional MDPs and already \NPTIME-hard for single-constraint queries. Exponential-memory strategies are both sufficient and in general necessary to satisfy $\varepsilon$-gap percentile queries.
\end{theorem}

\smallskip\noindent\textbf{Cornerstones of the algorithm.} Our approach is similar to the shortest path: we want to build an unfolding capturing the needed information w.r.t.~the discounted sums, and then reduce the percentile problem to a multiple reachability problem over this unfolding. However, several challenges have to be overcome.

First, we need a \textit{finite} unfolding. This was easy in the shortest path due to non-decreasing sums and corresponding upper bounds. Here, it is not the case as we put no restriction on weights. 
Nonetheless, thanks to the discount factor, weights contribute less and less to the sum along a run. In particular, cutting all runs after a pseudo-polynomial length changes the overall sum by at most~$\varepsilon/2$.

Second, we reduce the overall size of the unfolding.
For the shortest path we took advantage of integer labels to define equivalence. Here, the space of values taken by the discounted sums is too large for a straightforward equivalence.
To reduce it, we introduce a \textit{rounding} scheme of the numbers involved. This idea is inspired by~\cite{DBLP:conf/fsttcs/BrazdilCFNS13}. We bound the error due to cumulated roundings by $\varepsilon/2$.

So, we control the amount of information lost to guarantee exact answers except inside an arbitrarily small $\varepsilon$-zone.
Given a $q$-constraint query $\query$ for thresholds $v_{i}$, $\alpha_{i}$,
dimensions $l_{i}$ and discounts $\discount_{i}$, we define the
\textit{$x$-shifted query} $\query_{x}$, for $x \in \rat$, as the exact same
problem for thresholds $v_{i}+x$, $\alpha_{i}$, dimensions~$l_{i}$ and discounts
$\discount_{i}$. Our algorithm satisfies the following theorem, which formalizes
the $\varepsilon$-gap percentile problem mentioned in
Theorem~\ref{thm:ds_overview}. .
\begin{theorem}
\label{thm:ds_gap}
There is an algorithm that, given an MDP, a percentile query $\query$ for the discounted sum and a precision factor $\varepsilon > 0$, solves the following $\varepsilon$-gap problem in exponential time. It answers
\begin{itemize}
\item \textsf{Yes} if there is a strategy satisfying the $(2\cdot\varepsilon)$-shifted percentile query $\query_{2\cdot \varepsilon}$;

\item \textsf{No} if there is no strategy satisfying the $(-2\cdot\varepsilon)$-shifted percentile query $\query_{-2\cdot\varepsilon}$;

\item and arbitrarily otherwise.
\end{itemize}
\end{theorem}

We first state a more precise result and proceed with the technical discussion of the algorithm in the following paragraphs.

\begin{theorem}
\label{thm:disc_algo}
There is an algorithm satisfying the following properties.
\begin{enumerate}
\item It takes as input an MDP, a percentile query $\query$ for the discounted sum and a precision factor $\varepsilon > 0$.
\item If it outputs \textsf{Yes}, then there exists a strategy satisfying the percentile query $\query$.
\item If it outputs \textsf{No}, then there exists no such strategy.
\item If it outputs \textsf{Unknown}, then there exists a strategy satisfying at least the $(-2\cdot\varepsilon)$-shifted percentile query $\query_{-2\cdot\varepsilon}$ and there exists no strategy satisfying the $(2\cdot\varepsilon)$-shifted percentile query $\query_{2\cdot\varepsilon}$.
\item It runs in time polynomial in the size of the MDP, the weights, the discount factors and the precision factor, and exponential in the number of constraints.
\end{enumerate}
\end{theorem}

It suffices to prove Theorem~\ref{thm:disc_algo} for Theorem~\ref{thm:ds_gap} to follow as an immediate corollary for the $\varepsilon$-gap formulation of the problem.

\smallskip\noindent\textbf{Technical discussion.} Let $M = (S, A, \delta, \weight)$ be a $d$-dimensional MDP. We consider the $q$-constraint percentile query $\query \coloneqq \bigwedge_{i = 1}^{q}\; \pr_{M,\initState}^\strat\big[\discSum{\discount_{i}}_{l_{i}} \geq v_i\big] \geq
	\alpha_i$, where for $i \in \{1, \ldots{}, q\}$, we have that $v_{i} \in \rat$, $\alpha_{i} \in \left[ 0, 1\right] \cap \rat$, $\discount_{i} \in \left] 0, 1\right[ \cap \rat$ and $l_{i} \in \{1, \ldots{}, d\}$. Let $\varepsilon$ be an arbitrarily small precision factor. We assume w.l.o.g.~that $\varepsilon \in \rat_{0}$, i.e., we always use rational precision factors.
	
We now describe the algorithm and establish intermediate results related to the construction operated by the algorithm. We conclude by proving that all properties stated in Theorem~\ref{thm:disc_algo} are satisfied.

Our first step is building an unfolding of $M$, in the classical way. We denote it by $U$. Each node of $U$ is labeled by the corresponding state of $M$ and the discounted sum \textit{related to each query}, computed over the descending path from the root to the node. Observe that we have $q$ numerical dimensions in $U$ and not $d$ as in the shortest path. This will prove useful because we may have different discount factors for each constraint, hence the same dimension may induce different discounted sums depending on the considered constraint. This building scheme induces an infinite tree $U$ with nodes labeled by elements of $S \times \rat^{q}$.

In order to obtain a finite tree, we compute a bound $h$ on the height such that we do not lose too much information by cutting all branches at level $h$ (assuming the root node is at level $1$). Let $U_{h}$ denote the cut of $U$ at level $h$.

\begin{lemma}
\label{lem:ds_height}
There exists a pseudo-polynomial height $h$ such that for any infinite branch of $U$, its discounted sum on any dimension and w.r.t.~any of the discount factors is at most $\varepsilon/2$ far from the discounted sum of its prefix branch in $U_{h}$.
\end{lemma}

\begin{proof}
Consider any branch of $U_{h}$, for some $h \in \nat_{0}$. We denote the corresponding prefix of a run by $\pi = s_{1}a_{1}s_{2}a_{2}\ldots{}a_{h-1}s_{h}$. Its discounted sum w.r.t.~discount factor $\discount_{i}$ and dimension $l_{i}$ is $\discSum{\discount_{i}}_{l_{i}}(\pi) = \sum_{j = 1}^{h-1} \discount_{i}^{j}\cdot w_{l_{i}}(a_j)$. This branch could be extended in $U$ to any infinite branch that represents a prolonging run $\rho = s_{1}a_{1}\ldots{}a_{h-1}s_{h}a_{h+1}s_{h+1}\ldots{}$ of which $\pi$ is a prefix. We have that $\discSum{\discount_{i}}_{l_{i}}(\rho) = \sum_{j = 1}^{\infty} \discount_{i}^{j}\cdot w_{l_{i}}(a_j)$ and we want to pick $h$ such that
\begin{equation*}
\left\vert \discSum{\discount_{i}}_{l_{i}}(\rho) - \discSum{\discount_{i}}_{l_{i}}(\pi) \right\vert \leq \dfrac{\varepsilon}{2},
\end{equation*}
for any prolonging run $\rho$. That is, we want
\begin{equation*}
\left\vert \sum_{j = h}^{\infty} \discount_{i}^{j}\cdot w_{l_{i}}(a_j) \right\vert \leq \dfrac{\varepsilon}{2}.
\end{equation*}
Let $\discount = \max_{i} \discount_{i}$ be the largest discount factor (i.e., the one for which the discounting effect if the slowest) and let $W$ be the largest absolute weight appearing in the MDP $M$. We obtain that
\begin{equation*}
\left\vert \sum_{j = h}^{\infty} \discount_{i}^{j}\cdot w_{l_{i}}(a_j) \right\vert \leq W \cdot \sum_{j = h}^{\infty} \discount^{j} = W \cdot \left(\sum_{j = 0}^{\infty} \discount^{j} - \sum_{j = 0}^{h-1} \discount^{j}\right) = W \cdot \dfrac{\discount^{h}}{1-\discount}.
\end{equation*}
It thus suffices to take $h$ large enough to have that $W \cdot \frac{\discount^{h}}{1-\discount} \leq \frac{\varepsilon}{2}$. We assume that $W > 0$ otherwise the discounted sum is always zero and the percentile problem is trivial. We also recall that $0 < \discount < 1$. Hence the inequality becomes $\discount^{h} \leq \frac{\varepsilon\cdot (1 - \discount)}{2\cdot W}$. Applying the binary logarithm, we get the following inequality:
\begin{equation*}
h \cdot \log_{2}(\discount) \leq \log_{2}(\varepsilon) + \log_{2}(1 - \discount) - \log_{2}(W) - 1.
\end{equation*}
Since $\discount < 1$, we have that $\log_{2}(\discount) < 0$ and we finally obtain that
\begin{equation*}
h \geq \dfrac{\log_{2}(\varepsilon) + \log_{2}(1 - \discount) - \log_{2}(W) - 1}{\log_{2}(\discount)}.
\end{equation*}
Observe that this expression is always positive as $\varepsilon < 1$, $\discount < 1$ and $W \geq 1$. In the following, let us assume we take the ceiling of this expression as the value~$h$. What is the size of $h$ w.r.t.~the input of the algorithm? Since we are taking the binary logarithm of all involved values, it may seem that $h$ only needs to be polynomial in the encoding of the values. However, when $\discount \sim 1$, we have that $\log_{2} \discount \sim 1 - \discount$. Therefore, $h$ can be polynomial in the value of $\discount$, that is, exponential in its encoding.
\end{proof}

We now have a finite tree $U_{h}$, of pseudo-polynomial height, and such that all discounted sums labeled in its leaves are at most $\varepsilon/2$ far from the one of any prolonging run. In other words, once such a leaf has been reached, the controller may use any arbitrary strategy and its discounted sum will not vary by more than $\varepsilon/2$. This implies that we only care about devising a strategy for the $h$ first steps, as we will use later.

Consider the overall size of the tree $U_{h}$. As discussed for the shortest path, this size can be as high as $\mathcal{O}(b^{h})$, where $b$ denotes the branching degree of $\markovProcess$, defined as $b = \max_{s \in \states} \big\vert \{(a, s') \mid a \in A(s), s' \in \states, \delta(s,a,s') > 0\} \big\vert$. Thus, the overall size could be pseudo-exponential. Again, we want to reduce this tree $U_{h}$ to a compressed tree of truly-exponential size by merging equivalent nodes.

However, in this case it does not suffice to look for nodes with the exact same labels. Indeed, the range of possible labels is in general pseudo-exponential. Observe that the set of labels of any tree $U_{h}$ is a finite subset of $\states \times \left[ -W \cdot h, W \cdot h\right]^{q}$ (this characterization can be narrowed but it suffices for our needs). We introduce a value $\gamma \in \rat$ and maps the set of possible labels to $\states \times \{-W \cdot h, -W\cdot h + \gamma,  -W\cdot h + 2\cdot\gamma, \ldots{}, W\cdot h - \gamma, W \cdot h\}^{q}$ by rounding the values appearing in $U_{h}$ to multiples of $\gamma$ (we assume w.l.o.g.~that $W\cdot h$ is such a multiple). To that end, we define the function $\round\colon \rat \rightarrow \rat$ that rounds any rational $x \in \rat$ to the closest multiple of $\gamma$, i.e., the closest value in the new set of labels. The idea of rounding numbers to reduce the complexity is inspired by~\cite{DBLP:conf/fsttcs/BrazdilCFNS13}, but the technique differs.

Assume we apply this label mapping on $U_{h}$, for some fixed $\gamma$. Then, we define $U_{h, \sim_{\gamma}}$ as the MDP obtained by merging nodes having identical labels after the mapping. This is the unfolded MDP we are looking for \textit{if $\gamma$ is chosen adequately}, and it can be built on the fly by rounding each node (and potentially merging) at the moment it is created. Intuitively, $\gamma$ should not be too large to be able to keep the resulting rounding error low, but it should be large enough to induce a range of labels which is at most of exponential size. The following lemma states the existence of such a value $\gamma \in \rat$.

\begin{lemma}
\label{lem:ds_rounding}
There exists a value $\gamma \in \rat$ such that
\begin{enumerate}
\item $\big\vert \states \times \{-W \cdot h, -W\cdot h + \gamma, \ldots{}, W\cdot h - \gamma, W \cdot h\}^{q}\big\vert$ is at most exponential;
\item for all branch $\pi$ in $U_{h}$, for all $\discount_{i}, l_{i}$, $i \in \{1, \ldots{}, q\}$, we have that
\begin{equation*}
\left\vert \discSum{\discount_{i}}_{l_{i}}(\pi) - \roundedDiscSum{\discount_{i}}_{l_{i}}(\pi) \right\vert \leq \dfrac{\varepsilon}{2},
\end{equation*}
where $\roundedDiscSum{\discount_{i}}_{l_{i}}(\pi)$ denotes the rounded discounted sum of the corresponding branch $\pi'$ in $U_{h,\sim_{\gamma}}$ (i.e., the label of the corresponding leaf in $U_{h,\sim_{\gamma}}$).
\end{enumerate}
\end{lemma} 

\begin{proof}
We choose $\gamma = \dfrac{\varepsilon}{h-1}$ and prove the two assumptions. Observe that we assume $h > 1$ otherwise $U_{h}$ contains only the root node with all discounted sums equal to zero and no rounding is needed.

First, consider \textit{assumption 1}. The size of the set is $\vert \states \vert \cdot \left(\dfrac{2\cdot W \cdot h + 1}{\gamma} \right) ^{q}$.
Hence it suffices to prove that $\left(2\cdot W \cdot h + 1\right) \cdot \gamma^{-1}$ is at most exponential. Since both $h$ and $W$ are at most exponential (in the encoding of values), this boils down to proving that $\gamma^{-1} = \dfrac{h-1}{\varepsilon}$ is at most exponential, which is the case.

Second, let us prove \textit{assumption 2}. Recall that our rounding scheme maps each value to the closest multiple of~$\gamma$ whenever the label of a node is computed. It is important to understand that this rounding is executed on the fly, and not after building the tree $U_{h}$ fully (otherwise we would require pseudo-exponential time). Consequently, when a discounted sum for a node of level $2 \leq n \leq h$ is computed, we have to take into account that the label of its father of level $n-1$ has already been rounded: the rounding errors add up along a branch.

We claim that the total error over a branch of height $h$ is bounded by the expression $(h-1)\cdot \dfrac{\gamma}{2}$. That is, for all height-$h$ branch $\pi$ of $U_{h}$, for all $\discount_{i}$, $l_{i}$,
\begin{equation*}
\left\vert \discSum{\discount_{i}}_{l_{i}}(\pi) - \roundedDiscSum{\discount_{i}}_{l_{i}}(\pi) \right\vert \leq (h-1) \cdot \dfrac{\gamma}{2}.
\end{equation*}
We prove it by induction. Let $\pi = s_{1}a_{1}s_{2}\ldots{}s_{h}$ in the following.

The base case is $h = 2$. We ask whether
\begin{equation*}
\left\vert \discount_{i} \cdot \weight_{l_{i}}(a_{1}) - \round\big(\discount_{i} \cdot \weight_{l_{i}}(a_{1})\big) \right\vert \leq \dfrac{\gamma}{2}.
\end{equation*}
This is clearly true by definition of $\round$, which maps any rational to the closest multiple of $\gamma$.

Now assume our claim is true up to level $2 \leq h-1$. We prove it is still satisfied for level $h$. Let us rewrite $\left\vert \discSum{\discount_{i}}_{l_{i}}(\pi) - \roundedDiscSum{\discount_{i}}_{l_{i}}(\pi) \right\vert$ as follows:
\begin{align*}
\Big\vert \discSum{\discount_{i}}_{l_{i}}(s_{1}\ldots{}s_{h-1}) + &\discount_{i}^{h-1} \cdot \weight_{l_{i}}(a_{h-1}) - \round\left( \roundedDiscSum{\discount_{i}}_{l_{i}}(s_{1}\ldots{}s_{h-1}) +  \discount_{i}^{h-1} \cdot \weight_{l_{i}}(a_{h-1}) \right)  \Big\vert.
\end{align*}
Using the equality $\round(n\cdot\gamma + x) = n\cdot\gamma + \round(x)$ for $n \in \nat$ and $x \in \rat$, along with the fact that $\roundedDiscSum{\discount_{i}}_{l_{i}}(s_{1}\ldots{}s_{h-1})$ is already rounded by construction, we rewrite this as:
\begin{align*}
\Big\vert \discSum{\discount_{i}}_{l_{i}}(s_{1}\ldots{}s_{h-1}) + &\discount_{i}^{h-1} \cdot \weight_{l_{i}}(a_{h-1}) - \roundedDiscSum{\discount_{i}}_{l_{i}}(s_{1}\ldots{}s_{h-1}) - \round\left(\discount_{i}^{h-1} \cdot \weight_{l_{i}}(a_{h-1}) \right)  \Big\vert.
\end{align*}
By the subadditivity of $\vert \cdot \vert$, we bound this expression by
\begin{align*}
\Big\vert \discSum{\discount_{i}}_{l_{i}}(s_{1}\ldots{}s_{h-1}) &- \roundedDiscSum{\discount_{i}}_{l_{i}}(s_{1}\ldots{}s_{h-1})\Big\vert + \Big\vert\discount_{i}^{h-1} \cdot \weight_{l_{i}}(a_{h-1}) - \round\left(\discount_{i}^{h-1} \cdot \weight_{l_{i}}(a_{h-1}) \right)  \Big\vert.
\end{align*}
Finally, using the induction hypothesis for the first term and the definition of $\round$ for the second one, we can bound this sum by
\begin{equation*}
(h-2)\cdot \dfrac{\gamma}{2} + \dfrac{\gamma}{2} = (h-1) \cdot \dfrac{\gamma}{2},
\end{equation*}
which proves our initial claim.

Now, by our choice of $\gamma$, this implies that the total rounding error over any branch is at most $\dfrac{\varepsilon}{2}$, which proves the correctness of \textit{assumption 2}.
\end{proof}

Let us sum up the situation: given an MDP, a percentile query and a precision factor $\varepsilon > 0$, we are able to construct an unfolded MDP $U_{h,\sim_{\gamma}}$ of at most exponential size such that all leaves have labels in $\states \times \{-W \cdot h, -W\cdot h + \gamma, \ldots{}, W\cdot h - \gamma, W \cdot h\}^{q}$, where each of the $q$ numerical dimensions approximate the discounted sum of corresponding infinite branches within an error bounded by $\varepsilon$ ($\varepsilon/2$ due to truncating the branches and $\varepsilon/2$ due to the rounding of values).

The last step of our algorithm is as follows. Consider the $2\cdot q$ following target sets of nodes in $U_{h,\sim_{\gamma}}$.
\begin{itemize}
\item $\forall\, i \in \{1, \ldots{}, q\}$, ${\sf Sure}_{i}$ is the set of leaves for which the label on numerical dimension $i$ is greater than or equal to $v_{i} + \varepsilon$. Essentially, we have that $\roundedDiscSum{\discount_{i}}_{l_{i}}(\pi) \geq v_{i} + \varepsilon$, where $\pi$ denotes a corresponding descending branch.
\item $\forall\, i \in \{1, \ldots{}, q\}$, ${\sf Maybe}_{i}$ is the set of leaves for which the label on numerical dimension $i$ is greater than or equal to $v_{i} - \varepsilon$. Essentially, we have that $\roundedDiscSum{\discount_{i}}_{l_{i}}(\pi) \geq v_{i} - \varepsilon$, where $\pi$ denotes a corresponding descending branch.
\end{itemize}
Observe that ${\sf Sure}_{i} \subseteq {\sf Maybe}_{i}$ for all query $i$. Our algorithm proceeds as follows.
\begin{itemize}
\item[\textit{A)}] We execute the multiple reachability problem checking the existence of a strategy $\strat'$ such that
\begin{equation*}
\bigwedge_{i = 1}^{q}\; \pr_{U_{h,\sim_{\gamma}},\initState'}^{\strat'}\big[\diamondsuit {\sf Sure}_{i}\big] \geq
	\alpha_i,
\end{equation*}
with $\initState'$ the root node of the unfolded MDP. If the answer is \textsf{Yes}, then we answer \textsf{Yes} to the percentile problem. Otherwise, we proceed to the next step.
\item[\textit{B)}]  We execute the multiple reachability problem checking the existence of a strategy $\strat'$ such that 
\begin{equation*}
\bigwedge_{i = 1}^{q}\; \pr_{U_{h,\sim_{\gamma}},\initState'}^{\strat'}\big[\diamondsuit {\sf Maybe}_{i}\big] \geq
	\alpha_i,
\end{equation*}
with $\initState'$ the root node of the unfolded MDP. If the answer is \textsf{Yes}, then we answer \textsf{Unknown} to the percentile problem. Otherwise, we answer \textsf{No}.
\end{itemize}

The intuition is threefold. First, if a leaf of ${\sf Sure}_{i}$ is reached, then whatever the strategy that is played afterwards, any prolonging run will have a discounted sum at least equal to $v_{i}$ w.r.t.~the corresponding discount factor $\discount_{i}$ and dimension $l_{i}$. Hence, all prolonging runs are acceptable for constraint $i$. Second, if a leaf of ${\sf Maybe}_{i}$ is reached, then some prolonging runs may satisfy the constraint while other do not: we need to compute the unfolding for a smaller precision factor $\varepsilon$ in order to obtain useful information from nodes that are currently in ${\sf Maybe}_{i} \setminus {\sf Sure}_{i}$. Third, if a leaf does not belong to ${\sf Maybe}_{i}$, then any prolonging run is guaranteed to falsify constraint $i$ as adding error $\varepsilon$ does not suffice to make the discounted sum at least equal to $v_{i}$. We are finally able to prove Theorem~\ref{thm:disc_algo}.

\begin{proof}[Proof of Theorem~\ref{thm:disc_algo}] We consider each of \textit{properties 2-5} separately.

\textit{Property 2}. Our algorithm answers \textsf{Yes} if and only if there exists a strategy $\strat'$ satisfying the multiple reachability query $\bigwedge_{i = 1}^{q}\; \pr_{U_{h,\sim_{\gamma}},\initState'}^{\strat'}\big[\diamondsuit {\sf Sure}_{i}\big] \geq
	\alpha_i$. We define the strategy $\strat$ on the original MDP $M$ that plays as follows: it chooses the $(h-1)$ first actions according to $\strat'$ and then plays an arbitrary memoryless strategy. By Lemma~\ref{lem:ds_height}, Lemma~\ref{lem:ds_rounding}, and by definition of ${\sf Sure}_{i}$, this strategy guarantees that for all $i \in \{1, \ldots{}, q\}$, a discounted sum (w.r.t.~$\discount_{i}$ and $l_{i}$) at least equal to $v_{i}$ is achieved with probability at least equal to $\alpha_{i}$. Hence this finite-memory strategy $\strat$ satisfies the discounted sum percentile query on the original MDP $M$.
	
\textit{Property 3}. Our algorithm answers \textsf{No} if and only if there exists no strategy $\strat'$ satisfying the multiple reachability query $\bigwedge_{i = 1}^{q}\; \pr_{U_{h,\sim_{\gamma}},\initState'}^{\strat'}\big[\diamondsuit {\sf Maybe}_{i}\big] \geq
	\alpha_i$. By contradiction, assume the multiple reachability query cannot be satisfied, yet there exists a strategy $\strat$ in the original MDP $M$ that satisfies the percentile query for the discounted sum. That is, for all $i$ and associated $\discount_{i}, l_{i}$, this strategy achieves discounted sum at least~$v_{i}$ with probability at least $\alpha_{i}$. By Lemma~\ref{lem:ds_height} and Lemma~\ref{lem:ds_rounding}, we know that such a strategy reaches with probability at least~$\alpha_{i}$ leaves in $U_{h,\sim_{\gamma}}$ that are labeled with a value at least equal to $v_{i} - \varepsilon$ in numerical dimension $i$. That is, $\strat$ reaches each set ${\sf Maybe}_{i}$ with probability at least $\alpha_{i}$, which contradicts the hypothesis and proves the property.

\textit{Property 4}. Applying the same argument as for \textit{property 1}, if there exists a strategy $\strat'$ for the multiple reachability query $\bigwedge_{i = 1}^{q}\; \pr_{U_{h,\sim_{\gamma}},\initState'}^{\strat'}\big[\diamondsuit {\sf Maybe}_{i}\big] \geq
	\alpha_i$, then this strategy can be translated into a strategy $\strat$ over $M$ that ensures the percentile query where all value thresholds $v_{i}$ are replaced by their shifted version $v_{i} - 2\cdot \varepsilon$. Indeed, observe that the threshold gap between sets ${\sf Maybe}_{i}$ and ${\sf Sure}_{i}$ is exactly $2\cdot\varepsilon$. Conversely, we apply the argument of \textit{property~2} to deduce that if there exists no strategy for the multiple reachability query $\bigwedge_{i = 1}^{q}\; \pr_{U_{h,\sim_{\gamma}},\initState'}^{\strat'}\big[\diamondsuit {\sf Sure}_{i}\big] \geq
	\alpha_i$ (which is the case otherwise the answer of the algorithm would have been \textsf{Yes}), then there is no strategy for the percentile query shifted by $2\cdot\varepsilon$. 

\textit{Property 5}. It remains to study the complexity of our algorithm. Recall that the unfolded MDP $U_{h,\sim_{\gamma}}$ can be constructed in time
\begin{equation*}
\mathcal{O}\left( \vert \states \vert \cdot \left(\dfrac{2\cdot W \cdot h + 1}{\gamma} \right) ^{q} \right),
\end{equation*}
while $h$ is polynomial in $\discount = \max_{i} \discount_{i}$, $\log_{2}(\varepsilon)$ and $\log_{2}(W)$ and $\gamma$ is polynomial in both $\varepsilon$ and $h$. Moreover, multiple reachability queries executed by the algorithm only require polynomial time in $\vert U_{h,\sim_{\gamma}}\vert$ as all target states are absorbing (they are leaves in the unfolding). Overall, this shows that our algorithm requires time that is polynomial in $\vert \states \vert$, $W$, $\discount$ and $\varepsilon$, and exponential in $q$. This proves the property and finally concludes our proof of correctness for the algorithm.
\end{proof}

\smallskip\noindent\textbf{Lower bounds.} The $\varepsilon$-gap percentile problem is \PSPACE-hard by reduction from subset-sum games~\cite{DBLP:journals/tcs/Travers06}. Those are two-player games defined by a finite list of pairs of natural numbers $(a_{1}, b_{1})$, $(a_{2}, b_{2})$, $\ldots{}$, $(a_{n}, b_{n})$, and a target $t \in \nat$. Players take turns choosing between $a_{j}$ and $b_{j}$. After $n$ rounds, if the sum of the chosen numbers equals $t$, then player~1 wins, otherwise player~2 wins. Deciding if player~1 has a winning strategy in a subset-sum game is \PSPACE-complete~\cite{DBLP:journals/tcs/Travers06}.

\begin{lemma}
\label{lem:ds_pspace_hard}
The $\varepsilon$-gap problem defined in Theorem~\ref{thm:ds_gap} is \PSPACE-hard, already for two-dimensional MDPs and fixed values of discount and precision factors.
\end{lemma}

Two tricks are important. First, counterbalancing the discount effect via adequate weights. Second, simulating an equality constraint. This cannot be achieved directly because it requires to handle $\varepsilon = 0$. Still, by choosing weights carefully we restrict possible discounted sums to integer values only. Then we choose the thresholds and $\varepsilon > 0$ such that no run can take a value within the uncertainty zone. This circumvents the limitation due to uncertainty.

\begin{proof}
Consider a subset-sum game defined by pairs $(a_{1}, b_{1})$, $\ldots{}$, $(a_{n}, b_{n}) \in \nat^{2}$, and target $t \in \nat$. Assume that we have an algorithm, called \textsf{Algo}$_{\varepsilon}$, that solves the $\varepsilon$-gap problem of Theorem~\ref{thm:ds_gap}. We claim that this algorithm can also decide if player~1 has a winning strategy in the subset-sum game, through a polynomial-time reduction of the subset-sum game to a discounted sum percentile problem.

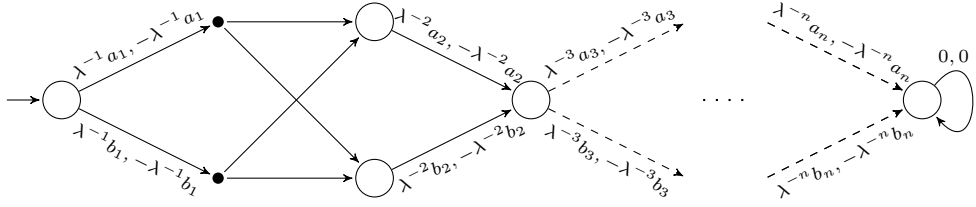
\begin{figure}[tb]
        \centering
               \scalebox{1}{\begin{tikzpicture}[->,>=stealth',shorten >=1pt,auto,node
    distance=2.5cm,bend angle=45, scale=0.52, inner sep=0pt,font=\scriptsize]
    \tikzstyle{p1}=[draw,circle,text centered,minimum size=5mm,text width=4mm]
    \tikzstyle{p2}=[fill,circle,text centered,minimum size=1.5mm]
    \tikzstyle{empty}=[]
    \node[p1]  (1) at (0, 0)  {};
    \node[p2]  (2)  at (4, 2) {};
    \node[p2]  (3)  at (4, -2) {};
    \node[p1]  (4)  at (8, 2) {};
    \node[p1]  (5)  at (8, -2) {};
    \node[p1]  (6)  at (12, 0) {};
    \node[empty]  (7)  at (16, 2) {};
    \node[empty]  (8)  at (16, -2) {};
    \node[empty]  (10)  at (18, 2) {};
    \node[empty]  (11)  at (18, -2) {};
    \node[p1]  (9) at (22, 0)  {};
    \path
    (-1.4,0) edge (1)
     (1) edge node[above,sloped,yshift=1mm] {$\discount^{-1} a_{1}, -\discount^{-1} a_{1}$} (2)
     (1) edge node[below,sloped,yshift=-1mm] {$\discount^{-1} b_{1}, -\discount^{-1} b_{1}$} (3)
    (2) edge[shorten <=1pt] (4)
    (2) edge[shorten <=1pt] (5)
    (3) edge[shorten <=1pt] (4)
    (3) edge[shorten <=1pt] (5)
    (4) edge node[above,sloped,yshift=1mm] {$\discount^{-2} a_{2}, -\discount^{-2} a_{2}$} (6)
    (5) edge node[below,sloped,yshift=-1mm] {$\discount^{-2}  b_{2}, -\discount^{-2}  b_{2}$} (6)
     (6) edge[dashed] node[above,sloped,yshift=1mm] {$\discount^{-3}  a_{3}, -\discount^{-3}  a_{3}$} (7)
     (6) edge[dashed] node[below,sloped,yshift=-1mm] {$\discount^{-3}  b_{3}, -\discount^{-3}  b_{3}$} (8)
    (10) edge[dashed] (9)
    (11) edge[dashed] (9)
    (10) edge[dashed] node[above,sloped,yshift=1mm] {$\discount^{-n}  a_{n}, -\discount^{-n}  a_{n}$} (9)
    (11) edge[dashed] node[below,sloped,yshift=-1mm] {$\discount^{-n}  b_{n}, -\discount^{-n}  b_{n}$} (9)
    (9) edge [loop right, out=50, in=310,looseness=3, distance=2cm] node [above,xshift=-2.5mm,yshift=4.5mm] {$0,0$} (9)
    ;
    \draw[loosely dotted,thick,-] (16.4,0) -- (17.6,0);
      \end{tikzpicture}}
        \caption{Encoding of subset-sum game into 2-dimensional percentile problem for the discounted sum.}\label{fig:ds_subset}
\end{figure}

We construct a 2-dimensional MDP $\markovProcess = (\states, A, \delta, \weight)$. Our construction is illustrated in Fig.~\ref{fig:ds_subset}. Filled circles represent equiprobable stochastic transitions. Controllable states simulate choices of player~1 in the game: the controller can choose between $a_{j}$ and $b_{j}$ when $j$ is odd. Conversely, stochastic transitions simulate choices of player~2: when $j$ is even, $a_{j}$ and $b_{j}$ are chosen with the same probability $1/2$. Each action corresponding to choosing $a_{j}$ (resp. $b_{j}$) has a 2-dimensional weight $(\discount^{-j} \cdot a_{j}, -\discount^{-j} \cdot a_{j})$ (resp. $(\discount^{-j} \cdot b_{j}, -\discount^{-j} \cdot b_{j})$). The discount factor can be fixed arbitrarily, say $\discount = 1/2$ for the sake of concreteness. Note that those weights only require an encoding which is polynomial in the size of the input. We add a self-loop with weight $(0, 0)$ on the terminal state.

Observe that any run in this MDP has a discounted sum which is exactly equal to the sum of the chosen elements $a_{j}$, $b_{j}$, thanks to the countereffect of $\discount^{-j}$ in the weights definition. Hence we also have that all runs have integer discounted sums. 

Our goal is to find a 2-dimensional percentile query that can express the winning condition of the subset-sum game, taking into account that algorithm \textsf{Algo}$_{\varepsilon}$ can only solve the $\varepsilon$-gap problem. 

Intuitively, we would like to express that the discounted sum must be exactly equal to $t$, in all possible runs. First observe that given the structure of the MDP, the terminal state and its zero loop is guaranteed to be reached in $n$ steps. Therefore, any strategy ensuring the required property almost-surely (i.e., with probability one) also ensures it surely (i.e., over all possible runs). Ideally, we would like to execute the $2$-constraint percentile problem asking for the existence of a strategy that satisfies query
\begin{equation*}
\query^{A} \coloneqq \quad \pr_{M,s}^\strat \big[\discSum{\discount}_{1} \geq t\big] = 1\quad \wedge \quad \pr_{M,s}^\strat \big[\discSum{\discount}_{2} \geq -t\big] = 1.
\end{equation*}
Let us call it \textit{Problem A}. Any strategy satisfying $\query^{A}$ would be a winning strategy for player-1, and conversely. Still, this would only be useful if we could take $\varepsilon = 0$, which we cannot.

Instead, consider \textit{Problem B}, asking for the existence of a strategy satisfying
\begin{equation*}
\query^{B} \coloneqq \quad \pr_{M,s}^\strat \big[\discSum{\discount}_{1} \geq t - 1/2\big] = 1\quad \wedge \quad \pr_{M,s}^\strat \big[\discSum{\discount}_{2} \geq -t - 1/2\big] = 1.
\end{equation*}
Furthermore, let us choose the precision factor $\varepsilon = 1/6$. Recall we assume that \textsf{Algo}$_{\varepsilon}$ solves the $\varepsilon$-gap problem. Consider the execution of \textsf{Algo}$_{\varepsilon}$ over query $\query^{B}$. By definition of the $\varepsilon$-gap problem (Theorem~\ref{thm:ds_gap}), we have that:
\begin{itemize}
\item[(1)] if there exists a strategy $\strat$ satisfying
\begin{equation*}
\query^{B}_{2\cdot\varepsilon} \coloneqq \quad \pr_{M,s}^\strat \big[\discSum{\discount}_{1} \geq t - 1/6\big] = 1\quad \wedge \quad \pr_{M,s}^\strat \big[\discSum{\discount}_{2} \geq -t - 1/6\big] = 1,
\end{equation*}
then the answer of \textsf{Algo}$_{\varepsilon}$ is \textsf{Yes};
\item[(2)] if there exists no strategy $\strat$ satisfying
\begin{equation*}
\query^{B}_{-2\cdot\varepsilon} \coloneqq \quad \pr_{M,s}^\strat \big[\discSum{\discount}_{1} \geq t - 5/6\big] = 1\quad \wedge \quad \pr_{M,s}^\strat \big[\discSum{\discount}_{2} \geq -t - 5/6\big] = 1,
\end{equation*}
then the answer of \textsf{Algo}$_{\varepsilon}$ is \textsf{No};
\item[(3)] otherwise the answer can be either \textsf{Yes} or \textsf{No}.
\end{itemize}
Now let us review the possible answers of \textsf{Algo}$_{\varepsilon}$.

Assume the answer is \textsf{Yes}. By (2), we have that there exists a strategy $\strat$ that satisfies $\query^{B}_{-2\cdot\varepsilon}$ otherwise the answer would have been \textsf{No}. Since all runs have integer discounted sums, this necessarily implies that $\strat$ also satisfies $\query^{A}$. Indeed, we have that $t = \lceil t-5/6\rceil$ and $-t = \lceil -t-5/6\rceil$. Hence player-1 has a winning strategy in the subset-sum game.

Assume the answer is \textsf{No}. By (1), we have that there exists no strategy $\strat$ that satisfies $\query^{B}_{2\cdot\varepsilon}$ otherwise the answer would have been \textsf{Yes}. Obviously, there exists no more strategy satisfying $\query^{A}$ since it is harder to satisfy (its thresholds are higher). Hence player-1 has no winning strategy in the subset-sum game.

Finally, we see that the answer of \textsf{Algo}$_{\varepsilon}$ is \textsf{Yes} if and only if the answer to \textit{Problem A} is also \textsf{Yes}. Since algorithm \textsf{Algo}$_{\varepsilon}$ can decide \textit{Problem A}, we also have that it can decide if player-1 has a winning strategy in the subset-sum game, which concludes our proof.
\end{proof}

For single-constraint $\varepsilon$-gap problems, we prove \NPTIME-hardness, even for Markov chains. Our proof is by reduction from the $K$-th largest subset problem~\cite{garey_FNY1979}, inspired by~\cite[Theorem 11]{DBLP:conf/stacs/BruyereFRR14}. A recent paper by Haase and Kiefer~\cite{HaasePP} shows that this $K$-th largest subset problem is actually \textsf{PP}-complete. This suggests that the single-constraint problem does not belong to $\NPTIME$ at all, otherwise the polynomial hierarchy would collapse to $\PTIME^{\NPTIME}$ by Toda's theorem~\cite{toda1991pp}.

\begin{lemma}
\label{lem:ds_np_hard}
The $\varepsilon$-gap problem defined in Theorem~\ref{thm:ds_gap} is \NPTIME-hard for single-constraint queries. This holds even for Markov chains, i.e., MDPs with only one available action in every state.
\end{lemma}

\begin{proof}
The $K$-th largest subset problem is as follows. Given a finite set $X = \{x_{1}, \ldots{}, x_{n}\}$ (hence $n = \vert X\vert$), a size function $h\colon X \rightarrow \nat$ assigning non-negative integer values to elements of $X$, and two naturals $K, L \in \nat$, decide if there exist $K$ distinct subsets $Y_{i} \subseteq X$, $1 \leq i \leq K$, such that $h(Y_{i}) = \sum_{x \in Y_{i}} h(x) \leq L$ for all $K$ subsets. The \NPTIME-hardness of this problem was proved in~\cite{johnson_JACM1978}.

We assume w.l.o.g.~that $K \leq 2^{n}$ otherwise the answer to the problem is trivially \textsf{No} since we cannot find a sufficient number of \textit{distinct} subsets.

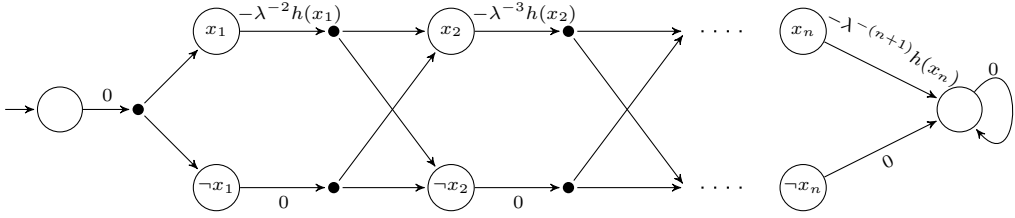
\begin{figure}[htb]
  \centering   
 \scalebox{1}{\begin{tikzpicture}[->,>=stealth',shorten >=1pt,auto,node
    distance=2.5cm,bend angle=45, scale=0.52, inner sep=0pt,font=\scriptsize]
    \tikzstyle{p1}=[draw,circle,text centered,minimum size=5mm,text width=6mm]
    \tikzstyle{p2}=[fill,circle,text centered,minimum size=1.5mm]
    \tikzstyle{empty}=[]
    \node[p1]  (1) at (0, 0)  {};
    \node[p2]  (1b)  at (2, 0) {};
    \node[p1]  (2)  at (4, 2) {$x_{1}$};
    \node[p1]  (3)  at (4, -2) {$\neg x_{1}$};
    \node[p2]  (2b)  at (7, 2) {};
    \node[p2]  (3b)  at (7, -2) {};
    \node[p1]  (4)  at (10, 2) {$x_{2}$};
    \node[p1]  (5)  at (10, -2) {$\neg x_{2}$};
    \node[p2]  (4b)  at (13, 2) {};
    \node[p2]  (5b)  at (13, -2) {};
    \node[empty]  (6)  at (16, 2) {};
    \node[empty]  (7)  at (16, -2) {};
    \node[p1]  (8)  at (19, 2) {$x_{n}$};
    \node[p1]  (9)  at (19, -2) {$\neg x_{n}$};
    \node[p1]  (10)  at (23, 0) {};
    \path
    (-1.4,0) edge (1)
    (1) edge node[above,yshift=1mm] {$0$} (1b)
    (1b) edge[shorten <=1pt] (2)
	(1b) edge[shorten <=1pt] (3)
    (9) edge node[below,sloped,yshift=-1mm] {$0$} (10)
    (8) edge node[above,sloped,yshift=2mm] {$-\discount^{-(n+1)}h(x_{n})$} (10)
    (4) edge node[above,yshift=1mm,xshift=0.8mm] {$-\discount^{-3}h(x_{2})$} (4b)
    (5) edge node[below,yshift=-1mm] {$0$} (5b)
    (2) edge node[above,yshift=1mm,xshift=0.8mm] {$-\discount^{-2}h(x_{1})$} (2b)
    (3) edge node[below,yshift=-1mm] {$0$} (3b)
    (2b) edge[shorten <=1pt] (4)
    (2b) edge[shorten <=1pt] (5)
    (3b) edge[shorten <=1pt] (4)
    (3b) edge[shorten <=1pt] (5)
    (4b) edge[shorten <=1pt] (6)
    (4b) edge[shorten <=1pt] (7)
    (5b) edge[shorten <=1pt] (6)
    (5b) edge[shorten <=1pt] (7)
    (10) edge [loop right, out=50, in=310,looseness=3, distance=2cm] node [above,xshift=-2.5mm,yshift=4.5mm] {$0$} (10)
    ;
    \draw[loosely dotted,thick,-] (16.4,2) -- (17.6,2);
    \draw[loosely dotted,thick,-] (16.4,-2) -- (17.6,-2);
      \end{tikzpicture}}
      \caption{Reduction from $K$-th largest subset problem to $\varepsilon$-gap problem for a single-constraint discounted sum percentile problem over a Markov chain.}
\label{fig:ds_np_hard}
  \end{figure}
  
Given an instance of the $K$-th largest subset problem, we build a Markov chain as depicted in Fig.~\ref{fig:ds_np_hard}. Observe that this is indeed a Markov \textit{chain} as there is a unique action available in all states. As usual, the filled circles represent equiprobable transitions. In the first step, element $x_{1}$ is either selected (upper transition) or not selected (lower one), with equal probability. This is repeated for every element up to reaching the terminal state with a zero loop. Hence, there is a bijection between runs in this Markov chain and subsets of~$X$. Moreover, all subsets are equiprobable: they have probability $1/2^{n}$ to be selected.

The discount factor can be chosen arbitrarily. For the sake of concreteness, assume it is $\discount = 1/2$. Now, observe that the weight function is defined such that the discounted sum over a run representing a subset $Y \subseteq X$ is exactly equal to $-h(Y) = -\sum_{x \in Y} h(x)$. To achieve this, we again use the trick of multiplying values $-h(x_{i})$ by $\discount^{-(i+1)}$ (the shift is due to the first transition). By definition of our weight function, it is clear that all runs take integer values. Also, the size of the Markov chain is polynomial in the size of the original problem.

Consider the single-constraint percentile query asking if
\begin{equation*}
\pr_{M,s} \big[\discSum{\discount} \geq -L - 1/2\big] \geq \dfrac{K}{2^{n}},
\end{equation*}
with $s$ the initial state of the Markov chain. Note that we drop the existential quantification on strategies since there exists a unique - and trivial - strategy in a Markov chain. Recall that we only have access to an algorithm, say \textsf{Algo}$_{\varepsilon}$, that solves the $\varepsilon$-gap problem, not the exact one. Consider $\varepsilon = 1/6$ and let us review the possible answers given by the execution of \textsf{Algo}$_{\varepsilon}$ on this query.

Assume \textsf{Algo}$_{\varepsilon}$ answers \textsf{Yes}. By definition of the $\varepsilon$-gap problem (Theorem~\ref{thm:ds_gap}), we have that
\begin{equation*}
\pr_{M,s} \big[\discSum{\discount} \geq -L-5/6\big] \geq \dfrac{K}{2^{n}} \quad\Rightarrow\quad \pr_{M,s} \big[\discSum{\discount} \geq -L\big] \geq \dfrac{K}{2^{n}}.
\end{equation*}
The implication follows from the fact that all runs take integer values and by equality $\lceil -L - 5/6\rceil = -L$ since $L \in \nat$. This implies that there are at least $K$ distinct runs representing subsets $Y_{i} \subseteq X$ for which $-h(Y_{i}) \geq -L \Leftrightarrow h(Y_{i}) \leq L$. Hence the answer to the $K$-th largest subset problem is also \textsf{Yes}.

Now assume \textsf{Algo}$_{\varepsilon}$ answers \textsf{No}. By definition of the $\varepsilon$-gap problem, we have that
\begin{equation*}
\pr_{M,s} \big[\discSum{\discount} \geq -L-1/6\big] < \dfrac{K}{2^{n}} \quad\Rightarrow\quad \pr_{M,s} \big[\discSum{\discount} \geq -L\big] < \dfrac{K}{2^{n}},
\end{equation*}
using the fact that the second inequality is harder to satisfy. This implies that there are strictly less than $K$ distinct runs representing subsets $Y_{i} \subseteq X$ for which $-h(Y_{i}) \geq -L \Leftrightarrow h(Y_{i}) \leq L$. Hence the answer to the $K$-th largest subset problem is also \textsf{No}.

In summary, we have that \textsf{Algo}$_{\varepsilon}$ answers \textsf{Yes} if and only if the answer to the $K$-th largest subset problem is also \textsf{Yes}. This concludes our proof.
\end{proof}

\smallskip\noindent\textbf{Memory.} For the precise discounted sum and
generalizations, infinite memory is
needed~\cite{DBLP:conf/lpar/ChatterjeeFW13}. For $\varepsilon$-gap problems,
the exponential upper bound follows from the algorithm while the lower bound is
shown via a family of problems that emulate the ones used for multiple
reachability (Theorem~\ref{thm:asreach}). 

\begin{lemma}
\label{lem:ds_memory}
Exponential-memory strategies are both sufficient and, in general, necessary to satisfy $\varepsilon$-gap percentile problems for the discounted sum.
\end{lemma}

\begin{proof}
First, the algorithm of Theorem~\ref{thm:disc_algo} solves the $\varepsilon$-gap percentile problem by answering a multiple reachability problem over an unfolded MDP of exponential size. As stated in Theorem~\ref{thm:absorbing-reachsafe}, memory of size polynomial in the MDP (here, the unfolded one) and exponential in the number of contraints (which is untouched by our algorithm) is sufficient to satisfy such queries. Moreover, once the first $h$ steps have been played according to such a strategy, any arbitrary strategy may be used, in particular a memoryless one suffices. Hence, it follows that exponential-memory strategies suffice for the discounted sum $\varepsilon$-gap percentile problem.

Second, for the lower bound we use a family of MDPs based on the one defined to prove the exponential memory requirements of multiple reachability problems (Lemma~\ref{lem:multiReach_expMemoryLB}). Consider the unweighted MDP depicted in Fig.~\ref{fig:multiReach_exp_mem}. Recall it is composed of $k$ stochastic gadgets followed by $k$ controllable gadgets. We transform this MDP into a $2\cdot k$-dimensional MDP $M$ as follows. First, we remove the self-loops on states $s'_{k,L}$ and $s'_{k,R}$ and replace them by actions going to a terminal state~$s_{t}$ with probability one: this is for technical convenience. Second, we associate actions to $2\cdot k$-dimensional weight vectors:
\begin{itemize}
\item the action leaving $s_{1}$ has weight $-\discount^{-1}$ in all $2\cdot k$ dimensions,
\item actions leaving a state $s_{i,L}$ have weight $\discount^{-2\cdot i}$ in dimension $i$ and weight zero in all other dimensions,
\item actions leaving a state $s_{i,R}$ have weight $\discount^{-2\cdot i}$ in dimension $k + i$ and weight zero in all other dimensions,
\item actions leaving a state $s'_{i,L}$ have weight $\discount^{-(k+ 2\cdot i)}$ in dimension $i$ and weight zero in all other dimensions,
\item actions leaving a state $s'_{i,R}$ have weight $\discount^{-(k+ 2\cdot i)}$ in dimension $k + i$ and weight zero in all other dimensions,
\item all remaining actions have weight zero in all dimensions.
\end{itemize}

As usual, the discount factor can be taken equal to $1/2$. While this may seem technical, the goal is simple: emulating the multiple reachability problem used in Lemma~\ref{lem:multiReach_expMemoryLB}. Each dimension $l \in \{1, \ldots{}, 2\cdot k\}$ will get a $-1$ by the first action. Then, a dimension $l \in \{1, \ldots{}, k\}$ (resp. $l \in \{k + 1, \ldots{}, 2\cdot k\}$) will get a $1$ when $s_{l,L}$ or $s'_{l,L}$ (resp. when $s_{l-k,R}$ or $s'_{l-k,R}$) is visited. All other actions have no impact on the discounted sum over dimension $l$. Therefore, one can easily check if a run $\rho$ has visited a set $\{s_{i,L}, s'_{i,L}\}$ (resp. $\{s_{i,R}, s'_{i,R}\}$): it suffices to check if the discounted sum on dimension $i$ (resp. $k+i$) is at least zero.

Now consider the percentile query
\begin{equation*}
\query \coloneqq \quad \bigwedge_{l = 1}^{2\cdot k}\; \pr_{M,s_{1}}^\strat\big[\discSum{\discount}_{l} \geq -1/2\big] = 1,
\end{equation*}
and in particular, its $\varepsilon$-gap version, with $\varepsilon = 1/6$. Applying the same reasoning as for proofs of Lemma~\ref{lem:ds_pspace_hard} and Lemma~\ref{lem:ds_np_hard}, we can prove that the answer to this $\varepsilon$-gap problem is \textsf{Yes} if and only if all target sets
\begin{equation*}
T_{l} = \{s_{1,L}, s'_{1,L}\}, \{s_{1,R}, s'_{1,R}\}, \{s_{2,L}, s'_{2,L}\}, \ldots{}, \{s_{k,L}, s'_{k,L}\}, \{s_{k,R}, s'_{k,R}\}
\end{equation*} 
are reached almost-surely. By Lemma~\ref{lem:multiReach_expMemoryLB}, we know that this requires a strategy encoded as a Moore machine with no less than $2^{k}$ memory states. This shows the exponential lower bound for the $\varepsilon$-gap problem and concludes our proof.
\end{proof}